\documentclass{llncs}

\usepackage{setspace}
\usepackage{graphicx}
\usepackage{url}
\usepackage{xspace}
\usepackage{enumitem}
\usepackage{algorithmicx}
\usepackage{algcompatible}
\usepackage{algpseudocode, algorithm}
\algrenewcommand\alglinenumber[1]{\tiny #1:}
\usepackage{multicol}
\usepackage{multirow}
\usepackage{float}
\usepackage{lineno}
\usepackage{float}
\usepackage{array}
\usepackage{xcolor}
\usepackage{mathtools}
\usepackage{subcaption}
\algnewcommand\algorithmicswitch{\textbf{switch}}
\algnewcommand\algorithmiccase{\textbf{case}}
\algnewcommand\algorithmicassert{\texttt{assert}}
\algnewcommand\Assert[1]{\State \algorithmicassert(#1)}%
\algdef{SE}[SWITCH]{Switch}{EndSwitch}[1]{\algorithmicswitch\ #1\ \algorithmicdo}{\algorithmicend\ \algorithmicswitch}%
\algdef{SE}[CASE]{Case}{EndCase}[1]{\algorithmiccase\ #1}{\algorithmicend\ \algorithmiccase}%
\algtext*{EndSwitch}%
\algtext*{EndCase}%

\newcommand{\blank}[1]{\hspace*{#1}}
\algdef{SE}[DOWHILE]{Do}{doWhile}{\algorithmicdo}[1]{\algorithmicwhile\ #1}%

\newcommand{\punt}[1]{}
\newcommand{\cmnt}[1]{}
\newcommand{\ignore}[1]{}

\newcommand{\secref}[1]{Section~\ref{sec:#1}}
\newcommand{\figref}[1]{Fig.~\ref{fig:#1}}

\newcommand{\thmref}[1]{Theorem~\ref{thm:#1}}

\newcommand{\algoref}[1]{{Algorithm \ref{alg:#1}}}
\newcommand{\subsecref}[1]{SubSection~\ref{subsec:#1}}

\newcommand{\apnref}[1]{Appendix~\ref{apn:#1}}

\newcommand{\Lineref}[1]{Line~\ref{lin:#1}}

\newcommand{\remove}[1]{}

\newcommand{\op} {operation\xspace}

\newcommand{\comm}[1] {\textit{committed(#1)}}
\newcommand{\aborted}[1] {\textit{aborted}(#1)}
\newcommand{\txns}[1] {\textit{txns}(#1)}

\newcommand{\evts}[1] {evts(#1)}

\newcommand{\tobj} {tobj\xspace}

\newcommand{\opq} {opaque\xspace}

\newcommand{\opty} {opacity\xspace}

\newcommand{\coopty} {co-opacity\xspace}

\newcommand{\sbty} {serializability\xspace}

\newcommand{\mvto} {MVTO\xspace}
\newcommand{\bto} {BTO\xspace}

\newcommand{\cl} {\emph{conflictList}\xspace}

\newcommand{\cminer} {Concurrent Miner\xspace}
\newcommand{\cvalidator} {Concurrent Validator\xspace}

\newcommand{\bg} {block graph\xspace}
\newcommand{\blg} {BG\xspace}
\newcommand{\aul} {\emph{auList[]}\xspace}
\newcommand{\confg} {\emph{confGraph}\xspace}
\newcommand{\gconfl} {\emph{getConflictfList}\xspace}
\newcommand{\vp} {\emph{vrtPred}\xspace}
\newcommand{\vc} {\emph{vrtCurr}\xspace}
\newcommand{\vn} {\emph{vrtNext}\xspace}
\newcommand{\ep} {\emph{egPred}\xspace}
\newcommand{\ec} {\emph{egCurr}\xspace}
\newcommand{\en} {\emph{egNext}\xspace}
\newcommand{\vl} {\emph{vrtList}\xspace}
\newcommand{\el} {\emph{egList}\xspace}
\newcommand{\inc} {\emph{indegree}\xspace}

\newcommand{\adde} {\emph{addEdge}\xspace}
\newcommand{\addv} {\emph{addVertext}\xspace}
\newcommand{\egn} {\emph{egNode}\xspace}
\newcommand{\vgn} {\emph{vrtNode}\xspace}
\newcommand{\enode} {\emph{egNode}\xspace}
\newcommand{\vnode} {\emph{vrtNode}\xspace}

\newcommand{\searchl} {\emph{localSearch}\xspace}
\newcommand{\searchg} {\emph{globalSearch}\xspace}
\newcommand{\remex} {\emph{remExNode}\xspace}
\newcommand{\gind} {gIndex}

\newcommand{\bc} {blockchain\xspace}

\newcommand{\scontract} {\emph{smart contract\xspace}}
\newcommand{\SContract} {smart contract\xspace}

\newcommand{\exec} {\emph{executeScFun}\xspace}

\newcommand{\nc} {\emph{sctCount}\xspace}

\newcommand{\vt} {\emph{vrtTail}\xspace}
\newcommand{\vh} {\emph{vrtHead}\xspace}
\newcommand{\eh} {\emph{eHead}\xspace}
\newcommand{\et} {\emph{eTail}\xspace}
\newcommand{\tl} {\emph{thLog}\xspace}
\newcommand{\cachel} {\emph{cacheList}\xspace}

\newcommand{\begtrans} {\emph{STM\_begin}\xspace}

\newcommand{\mvve} {MVVE\xspace}
\newcommand{\vie} {VE\xspace}
\newcommand{\ce} {CE\xspace}

\newcommand{\rs}{rset\xspace}
\newcommand{\ws}{wset\xspace}

\newcommand{\begt} {STM\_begin\xspace}
\newcommand{\tread} {STM\_read\xspace}
\newcommand{\twrite} {STM\_write\xspace}
\newcommand{\tlu} {STM\_lookup\xspace}
\newcommand{\tins} {STM\_insert\xspace}
\newcommand{\tdel} {STM\_delete\xspace}

\newcommand{\tryc} {STM\_tryC\xspace}

\newcommand{\commit}{\mathcal{C}}
\newcommand{\abort}{\mathcal{A}}
\newcommand{\mth} {method\xspace}

\newcommand {\incomp}[1] {#1.incomp}
\newcommand {\live}[1] {#1.live}

\newcommand{\inv} {$inv$}
\newcommand{\rsp} {$rsp$}

\newcommand{\sctl} {sctList\xspace}

\newcommand{\egnode} {\emph{egNode}\xspace}
\newcommand{\vrtnode} {\emph{vrtNode}\xspace}
\newcommand{\vrtlist} {vrtList\xspace}
\newcommand{\eglist} {egList\xspace}

\newcommand{\rwstm} {RWSTM\xspace}
\newcommand{\ostm} {OSTM\xspace}

\newcommand{\svotm} {SVOSTM\xspace}
\newcommand{\mvotm} {MVOSTM\xspace}

\newcommand{\oconf} {oconflict\xspace}
\newcommand{\mvoconf} {mvoconflict\xspace}
\newcommand{\rwconf} {rwconflict\xspace}

\newcommand{\sctrn} {SCT\xspace}
\newcommand{\fbr} {FBR\xspace}
\newcommand{\emb} {EMB\xspace}
\newcommand{\scv} {SMV\xspace}

\newcommand{\guc}[1] {#1.gUC\xspace}
\newcommand{\glc}[1] {#1.gLC\xspace}
\newcommand{\llc}[2] {#1.lLC_{#2}\xspace}
\newcommand{\luc}[2] {#1.lUC_{#2}\xspace}

\newcommand{\gucntr} {gUC\xspace}
\newcommand{\glcntr} {gLC\xspace}
\newcommand{\llcntr} {lUC\xspace}
\newcommand{\lucntr} {lLC\xspace}
\newcommand{\fb} {FBin\xspace}

\newcommand{\mthr} {multi-threaded\xspace}
\newcommand{\Mthr} {Multi-threaded\xspace}

\newcommand{\stat} {StaticBin\xspace}
\newcommand{\spec} {SpecBin\xspace}

\begin{document}
\title{Efficient Concurrent Execution of Smart Contracts in Blockchains using Object-based Transactional Memory\thanks{This paper is eligible for Best Student Paper award as Parwat is full-time Ph.D. student.}
}
%

\titlerunning{Efficient Concurrent Execution of Smart Contracts in Blockchains using OSTMs}
%
%

\author{Parwat Singh Anjana\inst{1}\and Hagit Attiya\inst{2}\and Sweta Kumari\inst{2}\and Sathya Peri\inst{1}\and Archit Somani\inst{2}}
%

\institute{Department of Computer Science \& Engineering, IIT Hyderabad, India \\
	\texttt{cs17mtech11014@iith.ac.in, sathya\_p@cse.iith.ac.in} \and Department of Computer Science, Technion, Israel \\
	\texttt{(hagit, sweta, archit)@cs.technion.ac.il}}



%


%
\maketitle              
\begin{abstract}
	Several popular blockchains such as Ethereum execute \emph{complex transactions} through user-defined scripts. A block of the chain typically consists of multiple \emph{smart contract transactions (SCTs)}. To append a block into the blockchain, a miner executes these SCTs. On receiving this block, other nodes act as \emph{validators}, who re-execute these SCTs as part of the consensus protocol to validate the block. In Ethereum and other blockchains that support cryptocurrencies, a miner gets an incentive every time such a valid block is successfully added to the \bc. 
When executing SCTs sequentially, miners and validators fail to harness the power of multiprocessing offered by the prevalence of multi-core processors, thus degrading throughput. By leveraging multiple threads to execute SCTs, we can achieve better efficiency and higher throughput. Recently, \emph{Read-Write Software Transactional Memory Systems (RWSTMs)} were used for concurrent execution of SCTs. It is known that \emph{Object-based STMs (OSTMs)}, using higher-level objects (such as hash-tables or lists), achieve better throughput as compared to RWSTMs. Even greater concurrency can be obtained using \emph{Multi-Version OSTMs (MVOSTMs)}, which maintain multiple versions for each shared data-item as opposed to \emph{Single-Version OSTMs (SVOSTMs)}. 

This paper proposes an efficient framework to execute SCTs concurrently based on object semantics, using \emph{optimistic} SVOSTMs and MVOSTMs. 
In our framework, a \mthr miner constructs a \textit{Block Graph (BG)}, capturing the \emph{object-conflicts} relations between SCTs, and stores it in the block. Later, validators re-execute the same SCTs concurrently and deterministically relying on this BG. 
A malicious miner can modify the BG to harm the blockchain, e.g., to cause \textit{double spending}. To identify malicious miners, we propose \emph{Smart \Mthr Validator (\scv)}. Experimental analysis shows that proposed \mthr miner and validator achieve significant performance gains over state-of-the-art SCT execution framework.



\end{abstract}

\section{Introduction}
\label{sec:intro}


Blockchains like Bitcoin \cite{Nakamoto:Bitcoin:2009} and Ethereum \cite{ethereum:url} have become very popular. Due to their usefulness, they are now considered for automating and securely storing user records such as land sale documents, vehicle, and insurance records. \emph{Clients}, external users of the system, send requests to nodes to execute on the blockchain, as \emph{smart contracts transactions (\sctrn{s})}. An SCT is similar to the methods of a class in an object-oriented langugage, which encode business logic relating to the contract \cite{Solidity,Dickerson+:ACSC:PODC:2017}.

Blocks are added to the blockchain by \emph{block-creator} nodes also known as \emph{miners}. A miner $m$ packs some number of \sctrn{s} received from various (possibly different) clients, to form a block $B$. Then, $m$ executes the \sctrn{s} of the block sequentially to obtain the final state of the blockchain, which it stores in the block. To maintain the chain structure, $m$ adds the hash of the previous block to the current block $B$ and proposes this new block to be added to the blockchain.


On receiving the block $B$, every other node acts as a \emph{validator}. The validators execute a global consensus protocol to decide the order of $B$ in the blockchain. As a part of the consensus protocol, validators validate the contents of $B$. They re-execute all the \sctrn{s} of $B$ sequentially to obtain the final state of the blockchain, assuming that $B$ will be added to the blockchain. If the computed final state matches the one stored in $B$ by the miner $m$ then $B$ is accepted by the validators. In this case, the miner $m$ gets an incentive for adding $B$ to the blockchain (in Ethereum and other cryptocurrency-based blockchains). Otherwise, $B$ is rejected, and $m$ does not get any reward. Ethereum follows order-execute model \cite{Androulaki+:Hyperledger:Eurosys:2018}, as do several other blockchains such as Bitcoin \cite{Nakamoto:Bitcoin:2009}, EOS \cite{eos:url}.

\vspace{.5mm}
\noindent
\textbf{Related Work:} Dickerson et al. \cite{Dickerson+:ACSC:PODC:2017} observed that both miner and validators can execute \sctrn{s} concurrently and harness the power of multi-core processors. They observed another interesting advantage of concurrent execution of \sctrn{s} in Ethereum, where only the miner receives an incentive for adding a valid block while all the validators execute the \sctrn{s} in the block. Given a choice, it is natural for a validator to pick a block that supports concurrent execution and hence obtain higher throughput. 



Concurrent execution of \sctrn{s} poses challenge. Consider a miner $m$ that executes the \sctrn{s} in a block concurrently. Later, a validator $v$ may re-execute same \sctrn{s} concurrently, in an order that may yield a different final state than given by $m$ in $B$. In this case, $v$ incorrectly rejects the valid block $B$ proposed by $m$. We denote this as \emph{False Block Rejection} (\emph{\fbr}), noting that \fbr may negate the benefits of concurrent execution.

\ignore{
Concurrent execution of \sctrn poses a challenge. Consider a miner that executes the \sctrn{s} in a block concurrently with $S$ being the equivalent serial execution \textcolor{blue}{(("this is new point which is made by the way" comment is not clear.))}and obtains a final state. Later, a validator may re-execute same \sctrn{s} concurrently according to another equivalent serial execution $S'$ and observe a different final state than given by the miner. Hence, it incorrectly rejects the valid block proposed by miner, causing \emph{False Block Rejection} (\emph{\fbr}), which may negate the benefits of concurrent execution. 

\cmnt{\setlength{\intextsep}{0pt}
\begin{figure*}
	\centering
	\centerline{\scalebox{0.38}{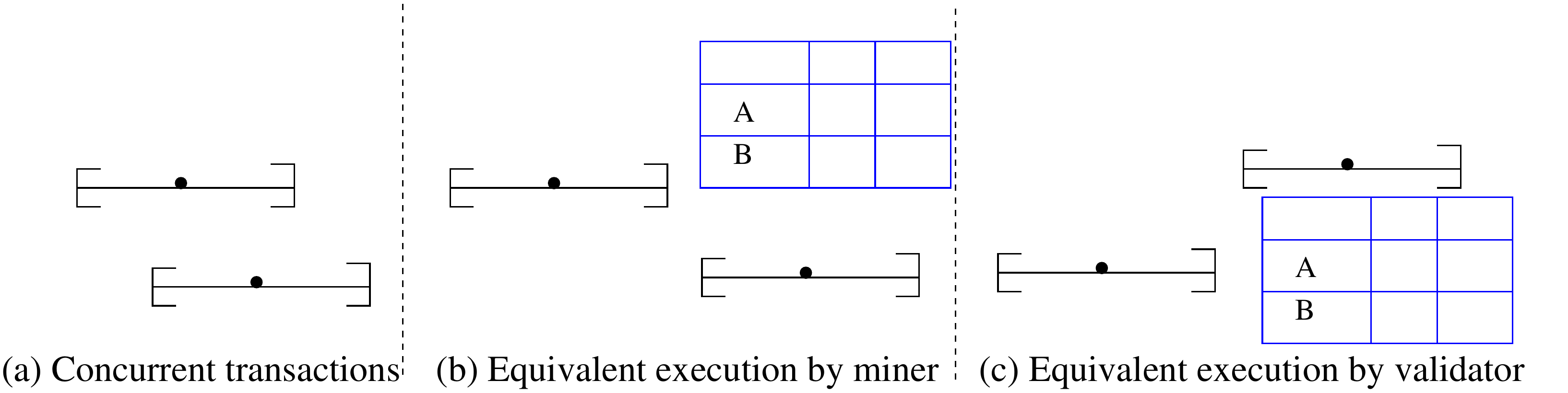}}
	  \caption{Challenges in Concurrent execution of \sctrn{s}}
	\label{fig:vfalse}
\end{figure*}}
}

Dickerson et al. \cite{Dickerson+:ACSC:PODC:2017} proposed a multi-threaded miner algorithm that is based on a \emph{pessimistic Software Transactional Memory (STM)} and uses locks for synchronization between threads executing SCTs. To avoid \fbr, the miner identifies the \emph{dependencies} between \sctrn{s} in the block while executing them by multiple threads. Two \sctrn{s} are \emph{dependent} if they are \emph{conflicting}, i.e., both of them access the same data-item and at least one of them is a write. These dependencies among \sctrn{s} are recorded in the block in form of a \emph{Block Graph (\blg)}. Two \sctrn{s} that have a path in the \blg are dependent on each other and cannot be executed concurrently. Later, a validator $v$ relies on the \blg to identify the dependencies among the \sctrn{s}, and concurrently execute \sctrn{s} only if there is no path between them in the \blg. In the course of the execution by $v$, the size of \blg dynmically decreases and the dependencies change. Dickerson et al. \cite{Dickerson+:ACSC:PODC:2017} use a \emph{fork-join} approach to execute the \sctrn{s}, where a master thread allocates \sctrn{s} without dependencies to different slave threads to execute.

Anjana et al. \cite{Anjana:OptSC:PDP:2019} used an \emph{optimistic} Read-Write STM (RWSTM), which identifies the conflicts between \sctrn{s} using timestamps. Those are used by miner threads to build the \blg. A validator processes a block using \blg in a completely decentralized manner using multiple threads, unlike the centralized fork-join approach of \cite{Dickerson+:ACSC:PODC:2017}. Each validator thread identifies an independent \sctrn and executes it concurrently with other threads. They showed that the decentralized approach yields significant performances gain over fork-join \cite{Dickerson+:ACSC:PODC:2017}. 

\ignore{
Anjana et al. \cite{Anjana:OptSC:PDP:2019} used \emph{optimistic} Read-Write STMs (RWSTMs). The STM identifies the conflicts between \sctrn{s} using timestamps, which miner threads use to build the \blg. 
A validator process a block using \blg in a completely decentralized manner using multiple threads, unlike the fork-join approach of \cite{Dickerson+:ACSC:PODC:2017}. Each validator thread identifies an independent \sctrn and concurrently executes it with other threads, yielding significant performances gain. 
}

Saraph and Herlihy \cite{Vikram&Herlihy:EmpSdy-Con:Tokenomics:2019} used a \emph{speculative bin} approach to execute \sctrn{s} of Ethereum in parallel. A miner maintains two bins for storing \sctrn{s}: \emph{concurrent} and \emph{sequential}. The \sctrn{s} are sorted into these bins using read-write locks. The \emph{concurrent bin} stores non-conflicting \sctrn{s} while the \emph{sequential bin} stores the remaining \sctrn{s}. If an \sctrn $T_i$ requests a lock held by an another \sctrn $T_j$ then $T_i$ is rolled back and placed in the sequential bin. Otherwise, $T_i$ is placed in the concurrent bin. To save the cost of rollback and retries of \sctrn{s}, they have used \emph{static conflict prediction} which identifies conflicting \sctrn{s} before executing them speculatively. The multi-threaded validator in this approach executes all the \sctrn{s} of the concurrent bin concurrently and then executes the \sctrn{s} of the sequential bin sequentially. We call this the \emph{Static Bin} approach. Zhang and Zhang \cite{ZangandZang:ECSC:WBD:2018} used \emph{multi-version timestamp order (MVTO)} for the concurrent execution of SCTs, in a pessimistic manner. 

\vspace{.5mm}
\noindent
\textbf{Exploiting Object-Based Semantics:} The STM-based solution of Anjana et al. \cite{Anjana:OptSC:PDP:2019} and others \cite{ZangandZang:ECSC:WBD:2018}, rely on \emph{read-write conflicts (\rwconf{s})} for synchronization. In contrast, \emph{object-based STMs (OSTMs)} track higher-level, more advanced conflicts between operations like insert, delete, lookup on a hash-table, enqueue/dequeue on queues, push/pop on the stack \cite{Hassan+:OptBoost:PPoPP:2014,HerlihyKosk:Boosting:PPoPP:2008,Peri+:OSTM:Netys:2018}. It has been shown in literature that \ostm{s} provide greater concurrency than \rwstm{s} (see \figref{ex1} in \apnref{ostm-adv}). This observation is important since Solidity \cite{Solidity}, the langugage used for writing \sctrn{s} for Ethereum, extensively uses a hash-table structure called \emph{mapping}. This indicates that a hash-table based OSTM is a natural candidate for concurrent execution of these \sctrn{s}.\footnote[1]{For clarity, we denote smart contract transactions as \sctrn{s} and an STM transaction as a transaction in the paper.} 

The lock-based solution proposed by Dickerson et al. \cite{Dickerson+:ACSC:PODC:2017} used abstract locks on hash-table keys, exploiting the object-based semantics with locks. In this paper, we want to exploit the object semantics of hash-tables using optimistic STMs to improve the performance obtained. 




To capture the dependencies between the \sctrn{s} in a block, miner threads construct the \blg concurrently and append it to the block. The dependencies between the transactions are given by the \emph{object-conflicts (\oconf{s})} (as opposed to \rwconf{s}) which ensure that the execution is correct, i.e., satisfies \emph{conflict-opacity} \cite{Peri+:OSTM:Netys:2018}. It has been shown \cite{Hassan+:OptBoost:PPoPP:2014,HerlihyKosk:Boosting:PPoPP:2008,Peri+:OSTM:Netys:2018} that there are fewer \oconf{s} than \rwconf{s}. Since there are fewer \oconf{s}, the \blg has fewer edges which in turn, allows validators to execute more \sctrn{s} concurrently. This also reduces the size of the \blg leading to a smaller communication cost.

\begin{figure}[tb]
	\centerline{\scalebox{0.35}{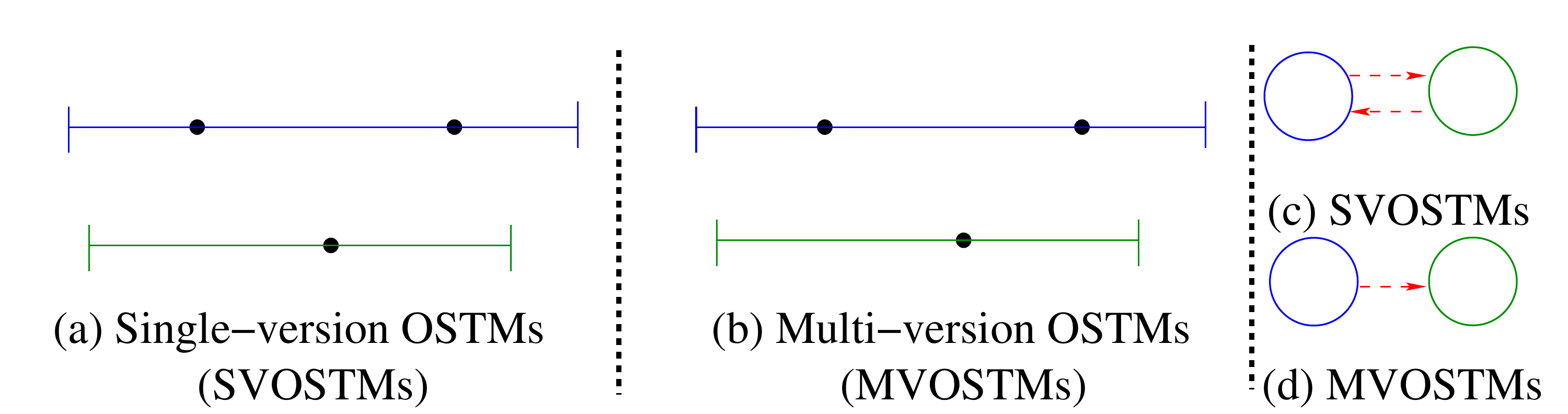}}
	\caption{\small(a) demonstrates two transactions $T_1$ and $T_2$. Between two \emph{get\_balance()} (or \emph{get()}) of $T_1$ on account $A_1$ and $A_2$ (initially both accounts maintain \$10 in each), $T_2$ \emph{sends} money from account $A_1$ to $A_2$. Thus, $T_1$ gets older balance for $A_1$ while newer balance of $A_2$. Hence, it cannot be serialized \cite{Papad:1979:JACM} (or \opq \cite{GuerKap:Opacity:PPoPP:2008}). The corresponding conflict-graph has a cycle as shown in (c). To ensure the correctness, \svotm{s} abort $T_1$. (b) shows the execution using \mvotm under the same scenario as (a). By maintaining multiple versions, \mvotm allows transaction $T_1$ to get the older balance for both accounts $A_1$ and $A_2$. Hence, for this execution the equivalent serial schedule is $T_1T_2$ as shown in (d). 
	}
	\label{fig:ex2}
\end{figure}

\emph{Multi-version object STMs (\mvotm{s})} \cite{Juyal+:MVOSTM:SSS:2018} demonstrate that by maintaining multiple versions for each shared data-item (object), even greater concurrency can be obtained as compared to traditional \emph{single-version \ostm{s} (SVOSTMs)}. \figref{ex2} illustrates the benefits of concurrent execution of \sctrn{s} by miner using \mvotm over \svotm. 
Thus a \blg based on \mvotm will have fewer edges than an \svotm-based \blg, and will further reduce the size of the \blg stored in the block. These advantages motivated us to use \mvotm{s} for concurrent execution of \sctrn{s} by miners.

Concurrent executions of SCTs may cause inconsistent behaviors such as \emph{infinite loops}, \emph{divide by zero}, \emph{crash failures}. Some of these behaviors, such as crash failures can be mitigated when SCTs are executed in a controlled environment, for example, the \emph{Ethereum Virtual Machine (EVM)} \cite{ethereum:url}. However, not all anomalies such as infinite loop can be prevented by the virtual machine. The inconsistent executions can be prevented by ensuring that the executions produced by the STM system are \opq \cite{GuerKap:Opacity:PPoPP:2008} or one of its variants such as \coopty \cite{Peri+:OSTM:Netys:2018}. Our \mvotm satisfies the former condition, opacity, while our \svotm satisfies the latter one, co-opacity.

\vspace{.5mm}
\noindent
\textbf{Handling a Malicious Miner: } A drawback of some of the approaches mentioned above is that a malicious miner can make the final state of the blockchain be inconsistent. In the \blg approach, the miner proposes an incorrect \blg which does not include all necessary edges. With the bin-based approach, the miner could place the conflicting transactions in the concurrent bin \cite{Vikram&Herlihy:EmpSdy-Con:Tokenomics:2019}. This can result in inconsistent states in the \bc due to \emph{double spending}, e.g., when two concurrent transactions incorrectly transfer the same amount of money simultaneously from a source account to two different destination accounts. 
If a malicious miner $mm$ does not add an edge between these two transactions in the \blg\cite{Anjana:OptSC:PDP:2019} or put these two transactions in concurrent bin \cite{Vikram&Herlihy:EmpSdy-Con:Tokenomics:2019} then both \sctrn{s} can execute concurrently by validators. 
Similarly, if a majority of validators accept the block containing these two transactions, then the state of the \bc becomes inconsistent. We denote this problem as \emph{edge missing \blg} (\emph{\emb}) in the case of the \blg approach \cite{Anjana:OptSC:PDP:2019} and \emph{faulty bin} (\emph{\fb}) in the case of the bin-based approach \cite{Vikram&Herlihy:EmpSdy-Con:Tokenomics:2019}. In \secref{result}, we show  the effect of malicious miners (\emb or \fb) through experiments on the blockchain system. 


To handle \emb and \fb errors, the validator must reject a block when edges are missing in the \blg or when conflicting \sctrn{s} are in the concurrent bin. Execution of such a graph or concurrent bin by the validator threads can lead to an inconsistent state. To detect such an execution, the validator threads watch and identify transactions performing conflicting access on the same data-items while executing concurrently. In \secref{pm}, we propose a \emph{Smart \Mthr Validator (\scv)} which uses $counters$ to detect this condition and reject the corresponding blocks. 

Dickerson et al. \cite{Dickerson+:ACSC:PODC:2017} suggest a lock-based solution to handle \emb errors. The miner generates and stores the lock profile required to execute the \sctrn{s} of a block along with the \blg. The validator then records a trace of the locks each of its thread would have acquired, had it been executing speculatively independent of the \blg. The validator would then compare the lock profiles it generated with the one provided by the miner present in the block. If the profiles are different then the block is rejected. This check is in addition to the check of the final state generated and the state in the block. This solution is effective in handling \emb errors caused by malicious miners. However, it is lock-based and cannot be used for preventing \emb issue in optimistic approaches such as \cite{Anjana:OptSC:PDP:2019}. The advantage of our \scv solution is that it works well with both optimistic and lock-based approaches. 

\noindent
\textbf{Our Contributions:}
This paper develops an efficient object semantics framework to execute \sctrn{s} concurrently by a miner using optimistic hash-table (both single and multi-version) OSTM. 
We use two methodologies to re-execute the \sctrn{s} concurrently by validators. In addition to the \emph{fork-join approach} employed by Dickerson et al. \cite{Dickerson+:ACSC:PODC:2017}, we also use a \emph{decentralized approach} \cite{Anjana:OptSC:PDP:2019} in which the validator threads execute independent \sctrn{s} concurrently in a decentralized manner. To handle \emb and \fb errors, we propose a decentralized \emph{Smart \Mthr Validator}.
To summarize: 
\begin{itemize}
	\item We introduce an efficient object-based framework for the concurrent execution of \sctrn{s} by miners (\secref{cminer}). We propose a novel way to execute the \sctrn{s} efficiently using optimistic \svotm by miner while ensuring \emph{\coopty} \cite{Peri+:OSTM:Netys:2018}. To further increase concurrency, we propose a new way for the execution of \sctrn{s} by the miner using optimistic \mvotm \cite{Juyal+:MVOSTM:SSS:2018} while satisfying opacity \cite{GuerKap:Opacity:PPoPP:2008}. 

	\item We propose the concurrent execution of \sctrn{s} by validators using \blg given by miner to avoid \fbr error (\secref{cvalidator}). The validator executes the \sctrn{s} using either fork-join or decentralized approaches. 
	
	\item We propose a Smart \Mthr Validator to handle \emb and \fb errors caused by malicious miners (\secref{malminer}).
	
	\item  Extensive simulations (\secref{result}) show that concurrent execution of \sctrn{s} by \svotm and \mvotm miner provide an average speedup of 3.41$\times$ and 3.91$\times$ over serial miner, respectively. \svotm and \mvotm based decentralized validator provide on average of 46.35$\times$ and 48.45$\times$ over serial validator, respectively.
\end{itemize}

\vspace{-.3cm}
\section{System Model}
\label{sec:model}
As in \cite{tm-book,KuzSat:NI:TCS:2016}, we consider $n$ threads, $p_1,\ldots,p_n$ in a system that access shared data-items (or objects/keys) in a completely asynchronous fashion. We assume that none of the threads/processes will crash or fail unexpectedly. 
 

\noindent
\textbf{Events:} A thread invokes the transactions and the transaction calls object level (or higher-level) methods which internally invokes read/write atomic events on the shared data-items to communicate with other threads. Method invocations (or \emph{\inv}) and responses (or \emph{\rsp}) are also considered as events.

\noindent
\textbf{History:} It is a sequence of invocations and responses of different transactional methods. We consider \emph{sequential history} in which invocation on each transactional method follows the immediate matching response. 

In this paper, we consider only \emph{well-formed} histories in which a new transaction will not begin until the invocation of previous transaction has not been committed or aborted.

\noindent
\textbf{Software Transactional Memory (STM):} STM \cite{KuzSat:NI:TCS:2016,Shavit:1995:STM:224964.224987} is a convenient concurrent programming interface for a programmer to access the shared memory using multiple threads. 
A typical STM works at lower-level (read-write) and exports following \mth{s}: (1) \begt{()}: begins a transaction with unique id. (2) \tread{($k$)} (or $r(k)$): reads the value of data-item $k$ from shared memory. (3) \twrite{($k,v$)} (or $w(k,v)$): writes the value of data-item $k$ as $v$ in its local log. (4) \tryc{()}: validates the transaction. If all updates made by the transaction is consistent then the updates will be reflected onto shared memory and transaction returns \emph{commit} (or $\commit$). Otherwise, transaction returns \emph{abort} (or $\abort$). Transaction $T_i$ starts with \begt{()} and completes when any of its methods return abort (or $\abort$) or commit (or $\commit$). The \tread{()} and \tryc{()} \mth{s} may return $\abort{}$. 

\ostm{s} export higher-level \mth{s}: (1) \begt{()}: begins a transaction with unique id. (2) \tlu{($k$)} (or $l(k)$): does a lookup on data-item $k$ from shared memory. (3) \tins{($k,v$)} (or $i(k,v)$): inserts the value of data-item $k$ as $v$ in its local log. (4) \tdel{($k$)} (or $d(k)$): deletes the data-item $k$. (5) \tryc{()}: validates the transaction. After successful validation, the actual effects of \emph{STM\_insert()} and \emph{STM\_delete()} will be visible in the shared memory and transaction returns $\commit$. Otherwise, it will return $\abort$. We represent \emph{\tlu{()}}, and \emph{\tdel{()}} as \emph{return-value ($rv{}$)} methods because both methods return the value from hash-table. We represent \emph{\tins{()}}, and \emph{\tdel{()}} as \emph{update} ($upd{}$) \mth{s} as on successful \emph{STM\_tryC()} both methods update the shared memory. Methods \textit{rv{()}} and \textit{\tryc}$()$ may return $\abort$. For a transaction $T_i$, we denote all the objects accessed by its $rv_i()$ and $upd_i()$ \mth{s} as $rvSet_i$ and $updSet_i$, respectively. 

\noindent
\textbf{Valid and Legal History:} If the  successful $rv_j(k, v)$ (i.e., $v \neq \abort$) method of a transaction $T_j$  \emph{returns} the value from any of previously committed transaction $T_i$ that has performed $upd()$ on key $k$ with value $v$  then such $rv_j(k, v)$ method is \emph{valid}. If all the \emph{rv()} methods of history $H$ is valid then $H$ is valid history \cite{Peri+:OSTM:Netys:2018}.

If the successful $rv_j(k, v)$ (i.e., $v \neq \abort$) method of a transaction $T_j$  \emph{returns} the value from  previous closest committed transaction $T_i$ that $k\in updSet_i$ ($T_i$ can also be $T_0$) and updates the $k$ with value $v$ then such $rv_j(k, v)$ method is \emph{legal}.   If all the \emph{rv()} methods of history $H$ is legal then $H$ is legal history \cite{Peri+:OSTM:Netys:2018}.  A legal history is also valid history.

Two histories H and H$'$ are \emph{equivalent} if they have the same set of events. H and H$'$ are \emph{multi-version view equivalent} \cite[Chap. 5]{WeiVoss:TIS:2002:Morg} if they are valid and equivalent. H and H$'$ are \emph{view equivalent} \cite[Chap. 3]{WeiVoss:TIS:2002:Morg} if they are legal and equivalent. Additional definitions are in \apnref{ap-model}.

\vspace{-.1cm}
\section{Proposed Mechanism}
\label{sec:pm}
This section describes our approach to the construction, data structures, and methods of concurrent BG, concurrent execution of \sctrn{s} by \mthr miner using optimistic object-based STMs, \mthr validator, and detection of a malicious miner.

\vspace{-.2cm}
\subsection{The Block Graph}
\label{sec:bg}
The \mthr miner executes the \sctrn{s} concurrently and stores the dependencies among them in a \blg. Each committed transaction corresponding to an \sctrn is a vertex in the \blg while edges capture the dependencies, based on the STM protocol. \Mthr miner uses \svotm and \mvotm to execute the \sctrn{s} concurrently, using timestamps. The challenge here is to construct the BG concurrently without missing any dependencies.


\svotm-based miner maintains three types of edges based on \oconf{s} between the transactions. An edge $T_i\rightarrow T_j$ between two transaction is defined when: \textbf{(1)} \textit{$rv_i(k,v) \text{ - } STM\_tryC_j()$ $edge:$} if $rv_i(k,v)<_HSTM\_tryC_j()$, $k\in updSet(T_j)$ and $v \neq \abort$; \textbf{(2)} \textit{$STM\_tryC_i() \text{ - } rv_j(k,v)$ $edge:$} if $STM\_tryC_i()<_H rv_j(k,v)$, $k\in updSet(T_i)$ and $v \neq \abort$; \textbf{(3)} \textit{$STM\_tryC_i() \text{ - } STM\_tryC_j()$ $edge:$} if $STM\_tryC_i() <_H STM\_tryC_j()$ and ($updSet(T_i)$ $\cap$ $updSet(T_j)\neq\emptyset$).


\mvotm-based miner maintains two types of edges based on \emph{multi-version \oconf{s}} (\emph{\mvoconf{s}}) \cite{Juyal+:MVOSTM:SSS:2018}. \textbf{(1)} \textit{return value from (rvf) edge:} if $STM\_tryC_i() <_H rv_j(k,v)$, $k\in updSet(T_i)$ and $v \neq \abort$ then there exist an \emph{rvf edge} from $T_i$ to $T_j$, i.e., $T_i \rightarrow T_j$; \textbf{(2)} \textit{multi-version (mv) edge:} consider a triplet, $STM\_tryC_i(), rv_m(k,v),$ $ STM\_tryC_j()$ in which ($updSet(T_i)$ $\cap$ $updSet(T_j)$ $\cap$ $rvSet(T_m)\neq \emptyset$), ($T_i$ and $T_j$ update the key $k$ with value $v$ and $u$ respectively) and ($u, v \neq \abort$); then there are two types of \emph{mv edge}: (a) if $STM\_tryC_i() <_H STM\_tryC_j()$ then there exist a  \emph{mv edge} from $T_m$ to $T_j$. (b)  if $STM\_tryC_j() <_H STM\_tryC_i()$ then there exist a  \emph{mv edge} from $T_j$ to $T_i$. 
We modified \svotm and \mvotm to capture \oconf{s} and \mvoconf{s} in the BG.

\noindent
\textbf{Data Structure for the Block Graph:} 
We use \emph{adjacency lists} to maintain the $BG(V, E)$. $V$ is the set of vertices (or SCTs) stored as a vertex list and $E$ is the set of edges (conflicts between SCTs) stored as edge list. Two lock-free methods build the BG (see details in \apnref{ap-dsBG}): \emph{addVertex()} adds a vertex and \emph{addEdge()} adds an edge in BG. To execute the \sctrn{s}, validator threads use three methods of block graph library: \searchg{()} identifies the independent vertex with indegree 0 to execute it concurrently, \remex{()} decrements the indegree of conflicting vertices and \searchl{()} identifies the independent vertex in thread local.


\cmnt{

Here $V$ is the set of vertices (\vrtnode{s}) is stored as a vertex list ($\vrtlist$). Similarly E is the set of Edges (\egnode{s}) is stored as edge list ($\eglist$ or conflict list) as shown in the \figref{graph}.(a), both $\vrtlist$ and $\eglist$ store between the two sentinel nodes \emph{Head}($-\infty$) and \emph{Tail}($+\infty$). Each \vrtnode{} maintains a tuple: \emph{$\langle$ts, scFun, indegree, egNext, vrtNext$\rangle$}. Here, \emph{ts} is the unique timestamp $i$ of the transaction $T_i$ to which this node corresponds to. scFun is the smart contract function executed by the transaction $T_i$ which is stored in \vrtnode. The number of incoming edges to the transaction $T_i$, i.e. the number of transactions on which $T_i$ depends, is captured by \emph{indegree}. Field \emph{egNext} and \emph{vrtNext} points the next \egnode{} and \vrtnode{} in the $\eglist$ and $\vrtlist$ respectively.      

Each \egnode{} of $T_i$ similarly maintains a tuple: \emph{$\langle$ts, vrtRef, egNext$\rangle$}. Here, \emph{ts} stores the unique timestamp $j$ of $T_j$ which has an edge coming from $T_i$ in the graph. BG maintains the conflict edge from lower timestamp transaction to higher timestamp transaction. This ensures that the \bg is acyclic. The \egnode{s} in $\eglist$ are stored in increasing order of the \emph{ts}. Field \emph{vrtRef} is a \emph{vertex reference pointer} which points to its own \vrtnode{} present in the $\vrtlist$. This reference pointer helps to maintain the \emph{indegree} count of \vrtnode{} efficiently. 

\begin{figure}
	\centerline{\scalebox{0.32}{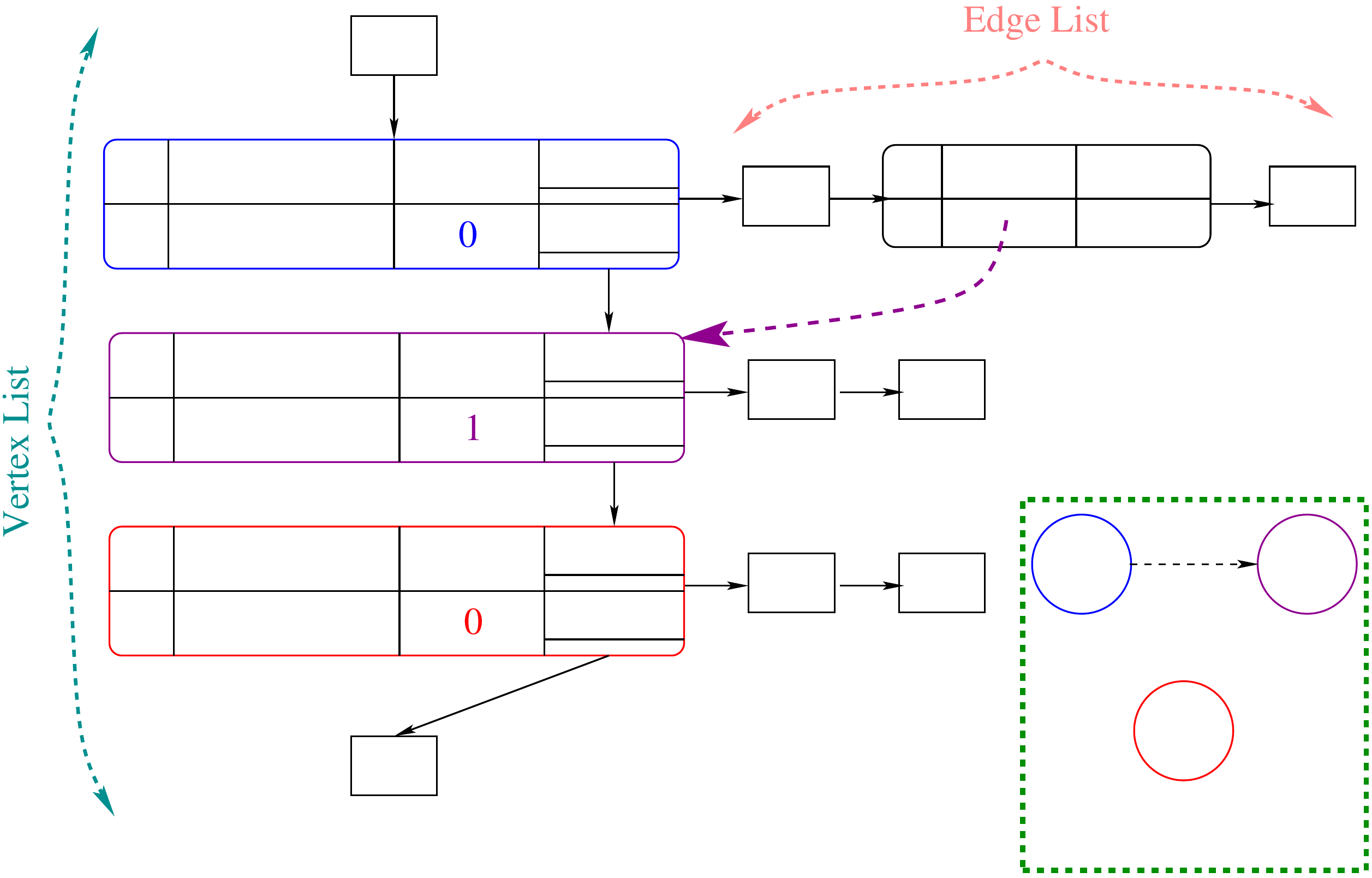}}
	 \caption{Construction of Block Graph}
	\label{fig:graph}
\end{figure}

\figref{graph}.(b) demonstrates the high level overview of BG which consist of three transaction $T_1$, $T_2$ and $T_3$. Here, $T_1$, $T_2$ are in conflict while $T_3$ is independent. The underlying representation of it illustrated in \figref{graph}.(a). For each transactions ($T_1, T_2$ and $T_3$) there exists a \vrtnode{} in the $\vrtlist$ of BG along with their conflicts. Since there is en edge from $T_1$ to $T_2$, an \egnode{} corresponding to $T_2$ is in the $\eglist$ of $T_1$. As mentioned earlier, the conflict edges go from lower timestamp to higher timestamp to ensure acyclicity of the \bg. After adding the \egnode, the \emph{indegree} of the \vrtnode of $T_2$ in the $\vrtlist$ is incremented as shown in \figref{graph}.(a). 

\noindent
\textbf{Block Graph Library Methods Accessed by Concurrent Miner:} Concurrent miner uses multiple threads to build the \bg. Specifically, the concurrent miner uses two methods to build the \bg: \emph{addVertex()} and \emph{addEdge()}. These two methods are \emph{lock-free} \cite{HerlihyShavit:Progress:Opodis:2011}. Here, \emph{addVertex(i)}, as the names suggests adds a \vrtnode with $ts=i$ for respective $scFun$ to the $vrtList$ of the BG if such a vertex is not already present. This node is atomically added to $vrtList$ using CAS operations.


The $addEdge(u, v)$ \mth creates an \egnode for $v$ in $u$'s \vrtnode if it does not already exist. First, it identifies the \egnode{} in the $\eglist$ of \vrtnode{}. If \egnode{} does not exist then it creates the node and adds into the $\eglist$ of \vrtnode{} atomically using CAS. The edges from $u$ to $v$ captures the conflicts between these transactions. This implies that $v$ is dependent on $u$ and the scFun of $v$ has to be executed only after $u$'s execution. 


\noindent
\textbf{Block Graph Library Methods Accessed by Concurrent Validator:} Concurrent validator uses multiple threads to re-executes the \sctrn{s} concurrently and deterministically with the help of BG given by the concurrent miner. To execute the \sctrn{s}, validator threads use three methods of block graph library: \searchg{()}, \remex{()} and \searchl{()}. First a validator thread $Th_i$ invokes the \searchg() \mth which searches for a \vrtnode{} $n$ in the \bg having $indegree$ 0 (i.e., source node). Such a node corresponds to a \sctrn which does not depend on other transactions and hence can be executed independently without worrying about synchronization issues. On identifying $n$, $Th_i$ atomically tries to claim it if not already claimed by some other thread. It does this by performing a CAS \op on the $indegree$ to -1. After successful execution of scFun of $n$, $Th_i$ invokes \remex \mth which decrements the \emph{inedgree} count for all the nodes which are have an incoming edge from $n$. This list of nodes are maintained in the \emph{\eglist} of $n$. 

While decrementing the \emph{indegree} count of conflicting nodes if the validator thread $Th_i$ finds any other \vrtnode{} with the \emph{indegree} as 0 then it adds that a reference to that node in its thread-local log $thLog_i$. The $thLog_i$ is used for optimization so that $Th_i$ need not to search in the global \blg to find the next source node. If a reference to the source node exists in the local log of validator, it is identified by the \searchl{()} \mth. $Th_i$ on identifying such a node $n$, atomically claims $n$ (if not already claimed by another thread). Then it executes the scFun of $n$ and then \remex as explained above. Please refer to the detailed description of BG methods along with pseudocode in \apnref{ap-bgcode}.  
}

\vspace{-.1cm}
\subsection{\Mthr Miner}
\label{sec:cminer}
A miner $m$ receives requests to execute \sctrn{s} from different clients. The miner $m$ then forms a block with several \sctrn{s} (the precise number of SCTs depend on the blockchain), $m$ execute these \sctrn{s} while executing the non-conflicting \sctrn{s} concurrently to obtain the final state of the blockchain. Identifying the non-conflicting \sctrn{s} statically is not straightforward because smart contracts are written in a turing-complete language \cite{Dickerson+:ACSC:PODC:2017} (e.g., Solidity \cite{Solidity} for Ethereum). We use optimistic STMs to execute the \sctrn{s} concurrently as in Anjana et al. \cite{Anjana:OptSC:PDP:2019} but adapted to object-based STMs on hash-tables to identify the conflicts. 

\ignore{
The multi-threaded miner receives \sctrn{s} requests from multiple clients to execute them and forms a block while executing the non-conflicting \sctrn{s} concurrently. Identifying the non-conflicting \sctrn{s} is not straightforward because smart contracts are written in a turing-complete language \cite{Dickerson+:ACSC:PODC:2017} (such as Solidity \cite{Solidity} for Etherum). We use optimistic STMs to execute the \sctrn{s} concurrently as in Anjana et al. \cite{Anjana:OptSC:PDP:2019} but adapted to object-based STMs on hash tables. 
}

\vspace{2mm}
\begin{algorithm}[tb]
	\scriptsize 
	\caption{\Mthr Miner{(\sctl{[]}, STM)}: $m$ threads concurrently execute the \sctrn{s} from \sctl with 
		STMs.}
	\label{alg:cminer}	
	\begin{algorithmic}[1]
		\makeatletter\setcounter{ALG@line}{0}\makeatother
		\Procedure{\emph{\Mthr Miner{ (\sctl{[]}, STM)}}}{} \label{lin:cminer1}
		\State $curInd = \gind.get\&Inc()$; // Atomically read the index and increment it. \label{lin:index}
		\While{($curInd < \sctl.length$)} // Execute until all SCTs have not been executed \label{lin:lt-chk}
		\State $curTrn = \sctl[curInd]$; // Get the current \sctrn to execute  \label{lin:sctrn}
		\State $T_i$ = \begtrans{()}; // Begins a new transaction. Here $i$ is unique id \label{lin:beg-tx}
		\ForAll{($curStep \in curTrn.scFun$)} // scFun is a list of steps \label{lin:curStep}
		\State \textbf{switch}(curStep) \label{lin:swt-cur}						
		\State \quad case lookup($k$): \label{lin:case-lookup}						
		\State \quad \quad $v$ $\gets$ \tlu{($k$)}; // Lookup data-item $k$ from a shared memory \label{lin:lu-ret} 
		\State \quad \quad \textbf{if}{($v == \abort$)} \textbf{then} goto \Lineref{beg-tx};\textbf{end if} break;		
		
		\State \quad case insert($k, v$): // Insert data-item k into $T_i$ local memory with value $v$ \label{lin:case-ins} 
		\State \quad \quad \tins{($k, v$)}; \label{lin:stm-ins} break;
		
		\State \quad case delete($k$): \label{lin:case-del}
		\State \quad \quad $v$ $\gets$ \tdel{($k$)}; // Actual deletion of data-item k happens in STM\_tryC()    \label{lin:del-ret} 
		\State \quad \quad \textbf{if}{($v == \abort$)} \textbf{then} goto \Lineref{beg-tx}; \textbf{end if} break;
		
		\State \quad default: Execute the step normally // Any step apart from lookup, insert, delete
		\State \textbf{endswitch}
		\EndFor
		\State $v$ $\gets$ STM\_tryC(); // Try to commit the transaction $T_i$\label{lin:tryC25}  
		\State \textbf{if}{($v == \abort$)} \textbf{then} goto \Lineref{beg-tx}; \label{lin:tryC} \textbf{end if}
		\State addVertex(i); // Create vertex node for $T_i$ with scFun \label{lin:addv}
		\State BG(i, STMs); // Add the conflicts of $T_i$ to \bg \label{lin:confs}
		\State $curInd = \gind.get\&Inc()$; // Atomically read the index and increment it. \label{lin:next-ind}				
		\EndWhile \label{lin:end-main_while}
		\State build-block(); // Here the miner builds the block. \label{lin:block}
		\EndProcedure \label{lin:cminer27}		
	\end{algorithmic}
\end{algorithm}

\algoref{cminer} shows how \sctrn{s} are executed by an $m$ threaded miner. The input is an array of \sctrn{s}, $\sctl$ and a object-based STM, (\svotm or \mvotm). We assume that both libraries support the \blg \mth{s} described above. The \mthr miner uses a global index into the \sctl $\gind$ which is accessed by all the threads. A thread $Th_x$ first reads the current value of $\gind$ into a local value $curInd$ and increments $\gind$ atomically (\Lineref{index}). 


Having obtained the current index in $curInd$, $Th_x$ gets the corresponding \sctrn, $curTrn$ from $\sctl[]$ (\Lineref{sctrn}). $Th_x$ then begins a STM transaction corresponding to $curTrn$ (\Lineref{beg-tx}). For every hash-table insert, delete and lookup encountered while executing the scFun of $curTrn$, $Th_x$ invokes the corresponding STM \mth{s}: \emph{\tlu{()}}, \emph{\tins{()}}, \emph{\tdel{()}}, either on an \svotm or on an \mvotm. Otherwise, it executes the step normally. If any of these steps fail, $Th_x$ begins a new STM transaction (\Lineref{beg-tx}) and re-executes these steps. 
 
Upon successful completion of transaction $T_i$, $Th_x$ creates a vertex node for $T_i$ in the \bg (\Lineref{addv}). Then, $Th_x$ obtains the transactions (\sctrn{s}) with which $T_i$ is conflicting from the \ostm, and adds the corresponding edges to the BG (\Lineref{confs}). $Th_x$ then gets the index of the next \sctrn to execute (\Lineref{next-ind}). 

An important step here is how the underlying \ostm{s} (either \svotm or \mvotm) maintain the conflicts among the transactions which is used by $Th_x$ (see \apnref{ap-sobjds}). Both \svotm and the \mvotm use timestamps to identify the conflicts.	 


\ignore{
	\begin{figure}[H]
		\centerline{\scalebox{0.25}{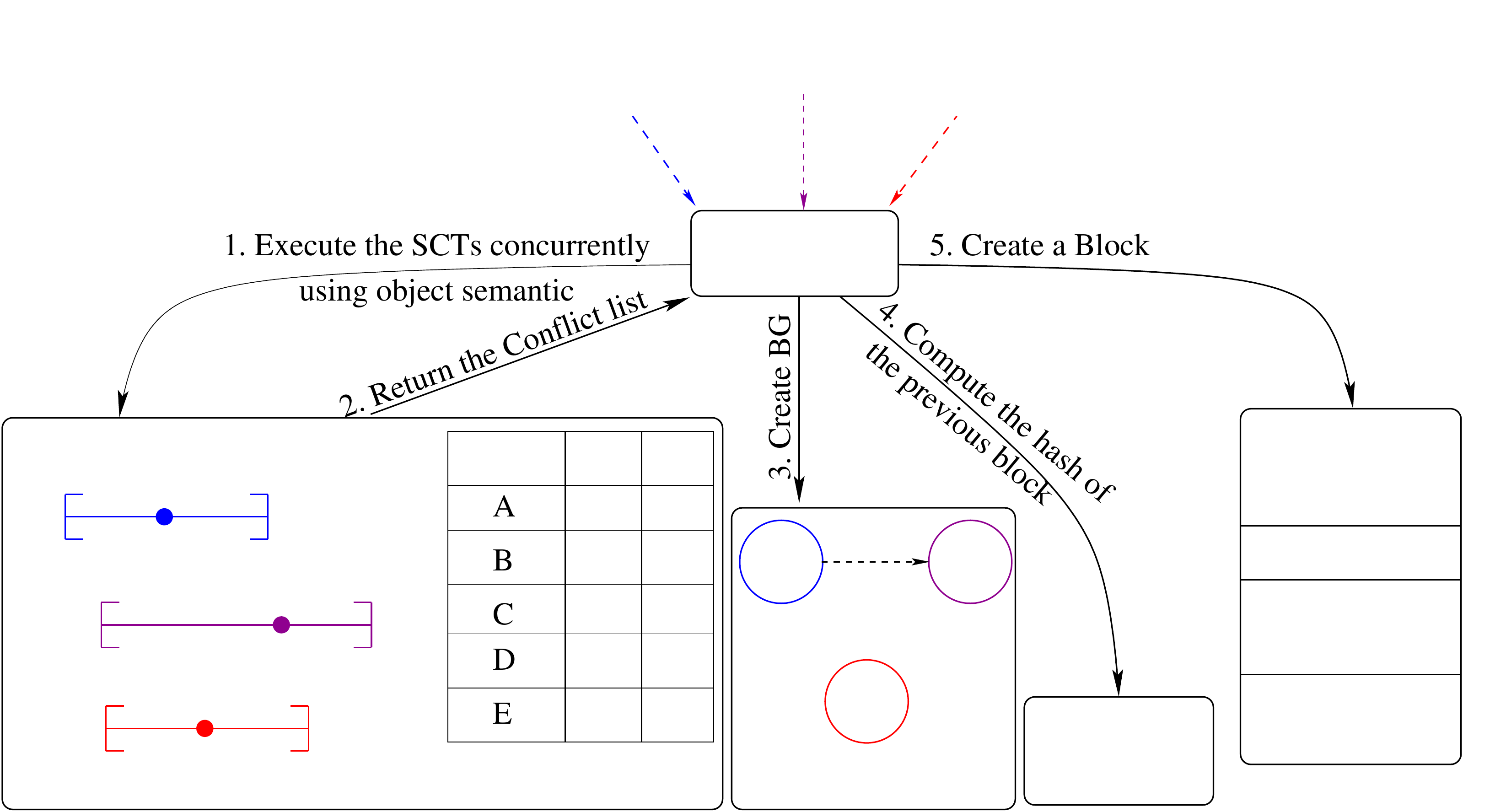}}
		\caption{Working of Concurrent Miner}
		\label{fig:cminer}
	\end{figure}
	
	\begin{figure}
		\centerline{\scalebox{0.4}{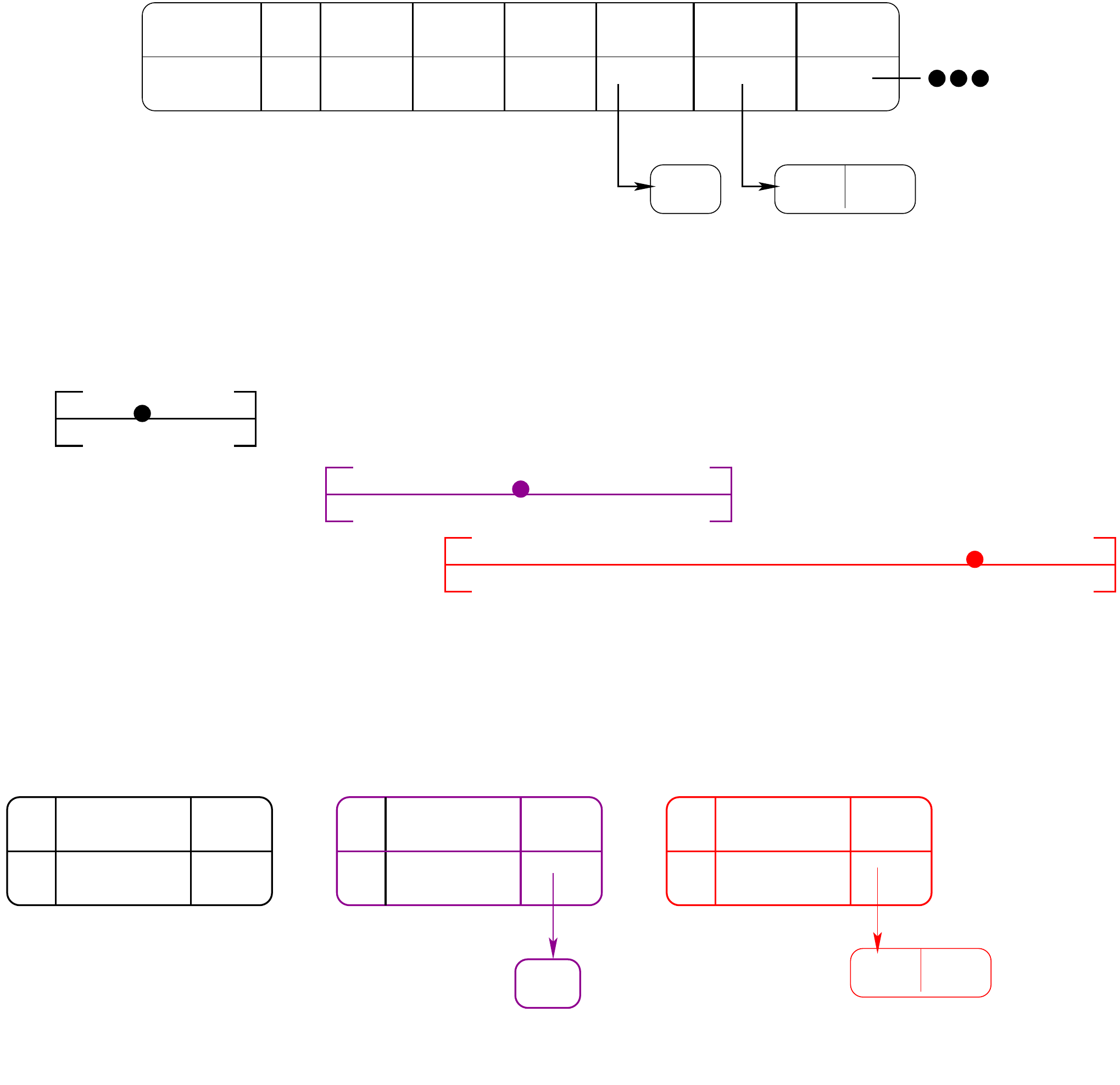}}
		\caption{\svotm conflict list}
		\label{fig:svostm-cl}
	\end{figure}
	
	\begin{figure}
		\centerline{\scalebox{0.4}{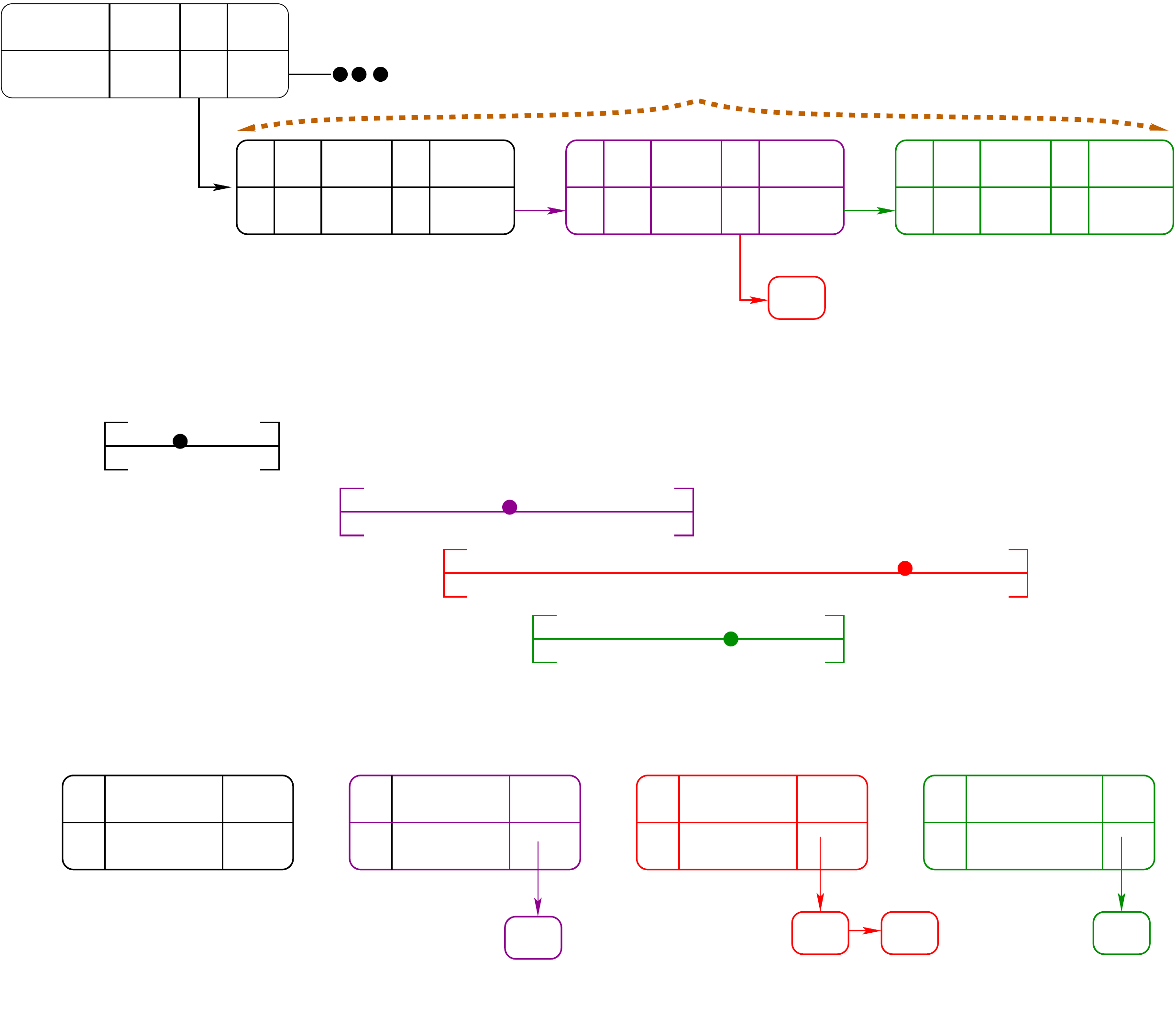}}
		\caption{\mvotm conflict list}
		\label{fig:mvostm-cl}
	\end{figure}
}

Once all the \sctrn{s} of \sctl have been executed successfully and the BG is constructed concurrently, it is stored in the proposed block. The miner then stores the final state ($FS_m$) of the blockchain (which is the state of all shared data-items), resulting from the execution of \sctrn{s} of \sctl in the block. The miner then computes the \op{s} related to the blockchain. For Ethereum, this would constitute the hash of the previous block. Then the \mthr miner proposes a block which consists of all the \sctrn{s}, BG, $FS_m$ of all the shared data-items and hash of the previous block (\Lineref{block}). The block is then broadcast to all the other nodes in the blockchain. 


\setlength{\intextsep}{0pt}

\apnref{ap-correctness} proves the following theorems:
\vspace{-.15cm}
\begin{theorem}
	\label{thm:BGdep}
	The BG captures all the dependencies between the conflicting nodes.
\end{theorem}
\vspace{-.25cm}
\begin{theorem}
	\label{thm:co-ostm}
	A history $H_m$ generated by the multi-threaded miner with \svotm satisfies co-opacity.
\end{theorem}
\vspace{-.25cm}
\begin{theorem}
	\label{thm:o-mvostm}
	A history $H_m$ generated by multi-threaded miner with \mvotm satisfies opacity.
\end{theorem}


\vspace{-.2cm}
\subsection{\Mthr Validator}
\label{sec:cvalidator}

The validator re-executes the \sctrn{s} deterministically relying on the \blg provided by the miner in the block. \blg consists of dependency among the conflicting \sctrn{s} and restrict validator threads to execute them serially to avoid the \emph{False Block Rejection (FBR) error} while non-conflicting SCTs execute concurrently to obtain greater throughput. The validator uses \searchg{()}, \searchl{()}, and \remex{()} \mth{s} of the BG library as described in \secref{bg}. 

\ignore{
High level overview of \algoref{val} shows the execution of \sctrn{s} by concurrent validator with the help of BG. First, multiple validator threads concurrently identify the source node (\emph{indegree} 0) in the BG using \emph{globalSearch()} at \Lineref{val5}. After identifying the source node, thread claims it (sets \emph{indegree} to -1) atomically so that other concurrent validator threads can not claim it. Then it executes the scFun of \sctrn corresponding to the source node. After successful execution of scFun, it decrements the \emph{indegree} count of conflicting node of source node using \emph{remExNode()} at \Lineref{val7}. While decrementing the \emph{indegree} of conflicting node, if node became source node then validator thread stores that node in its thread local log $thLog$ to execute next \sctrn at \Lineref{val10} efficiently. 


\begin{figure}
	\centerline{\scalebox{0.23}{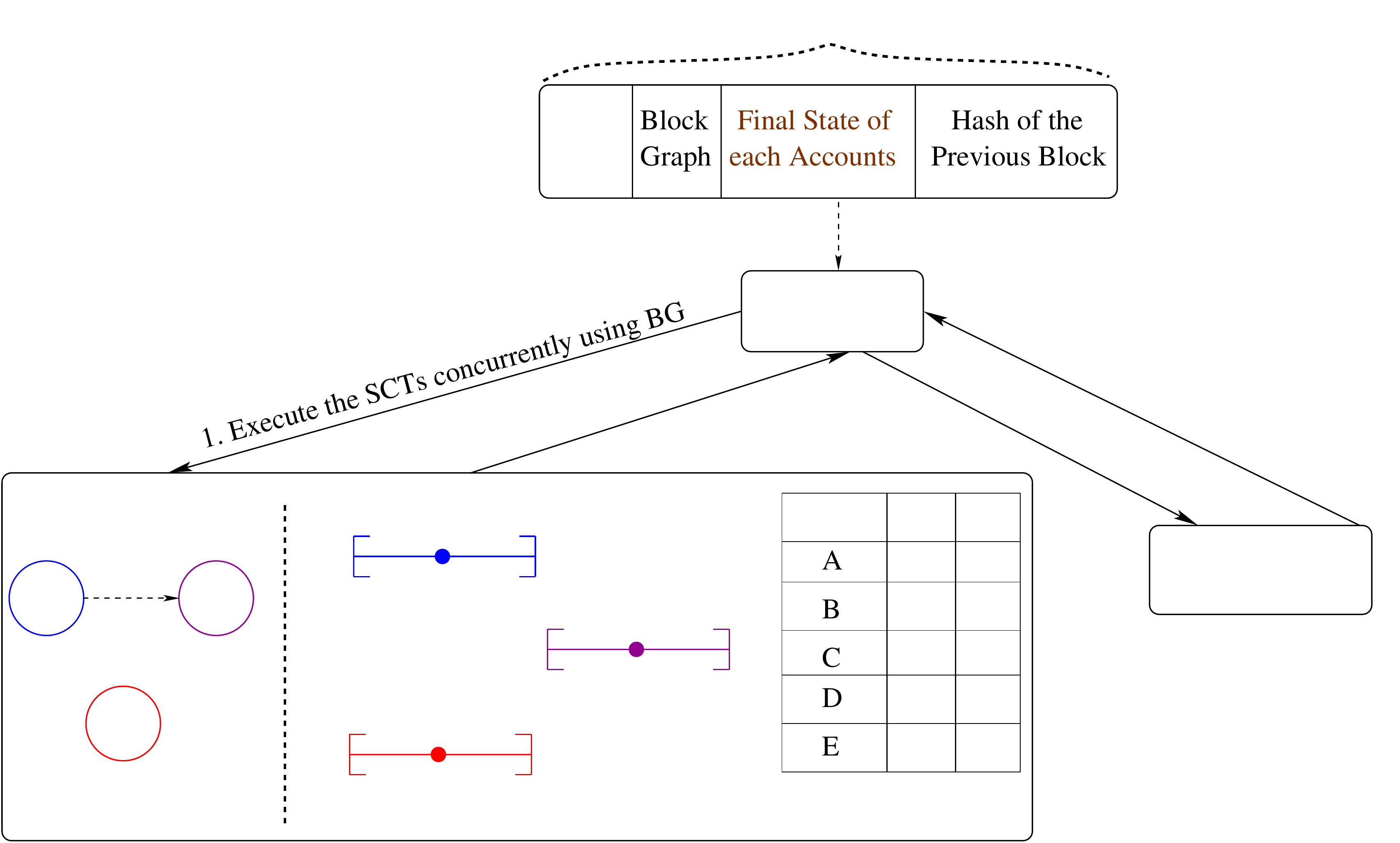}}
	\caption{Working of Concurrent Validator}
	\label{fig:cval}
\end{figure}

\figref{cval}.(a) represents the BG in which conflicting edge is from transaction $T_1$ to $T_2$ because they accessed the same account as B. So concurrent validator needs to respect the conflicting order of the transaction to avoid the \emph{false block rejection (or FBR)} error and execute them serially as depicted in \figref{cval}.(b). However, transaction $T_3$ don't have any conflict with other transactions so it can be executed concurrently.
} 

After successful execution of the \sctrn{s}, validator threads compute the final state ($FS_v$) of the blockchain which is the state of all shared data items. If it matches the final state $FS_m$ provided by the miner then the validator accepts the block. If a majority of the validators accept the block, then it is added to the blockchain. Detailed description appears in \apnref{ap-cvalidator}.

\apnref{ap-correctness} proves the following theorems:

\vspace{-.15cm}
\begin{theorem}
	\label{thm:hmve1}
	A history $H_m$ generated by the \mthr miner with \svotm and history $H_v$ generated by a \mthr validator are view equivalent.
\end{theorem}
\vspace{-.25cm}
\begin{theorem}
		\label{thm:hmmvve}
	A history $H_m$ generated by the \mthr miner with MVOSTM and history $H_v$ generated by a \mthr validator are multi-version view equivalent. 
\end{theorem}



\subsection{Detection of Malicious Miners by Smart \Mthr Validator (SMV)}
\label{sec:malminer}

We propose a technique to handle edge missing BG (\emb) and Faulty Bin (FBin) caused by the malicious miner as explained in \secref{intro}. A malicious miner $mm$ can remove some edges from the BG and set the final state in the block accordingly. A \mthr validator executes the \sctrn{s} concurrently relying on the BG provided by the $mm$ and results the same final state. Hence, incorrectly accepts the block. Similarly, if a majority of the validators accept the block then the state of the blockchain becomes inconsistent. For instance, a double spending can be executed.

A similar inconsistency can be caused by a $mm$ in bin-based approach: $mm$ can maliciously add conflicting \sctrn{s} to the concurrent bin resulting in \fb error. This may cause \mthr validator $v$ to access shared data items concurrently leading to synchronization errors. To prevent this, the \scv checks to see if two concurrent threads end up accessing the same shared data item concurrently. If this situation is detected, then the miner is malicious. 

To identify such situations, \scv uses $counters$, inspired by the \emph{basic timestamp ordering (\bto)} protocol in databases \cite[Chap. 4]{WeiVoss:TIS:2002:Morg}. \scv keeps track of each global data item that can be accessed across multiple transactions by different threads. Specifically, \scv maintains two global counters for each key of hash-table (shared data item) $k$ (a) $\guc{k}$ - global update counter (b) $\glc{k}$ - global lookup counter. These respectively keep track of number of \textbf{updates} and \textbf{lookups} that are concurrently performed by different threads on $k$. Both counters are initially 0. 


When an \scv thread $Th_x$ is executing an \sctrn $T_i$ it maintains two local variables corresponding to each global data item $k$ which is accessible only by $Th_x$ (c) $\luc{k}{i}$ - local update counter  (d) $\llc{k}{i}$ - local lookup counter. These respectively keep track of number of updates and lookups performed by $Th_x$ on $k$ while executing $T_i$. These counters are initialized to 0 before the start of $T_i$.  

Having described the counters, we will explain the algorithm at a high level. Suppose the next step to be performed by $Th_x$ is:

\vspace{-.25cm}
\begin{enumerate}
    \item $lookup(k)$: Thread $Th_x$ will check for equality of the local and global update counters, i.e., $(\luc{k}{i} == \guc{k})$. If they are not same then \scv will report the miner as malicious. Otherwise, (i) $Th_x$ will atomically increment $\glc{k}$. (ii) $Th_x$ will increment $\llc{k}{i}$. (iii) Perform the lookup on the key $k$ from shared memory. 
    \ignore{
    \begin{enumerate}
        \item $Th_x$ will atomically increment $\glc{k}$.
        \item $Th_x$ will increment $\llc{k}{i}$.
        \item Perform the lookup on the key $k$. 
    \end{enumerate}
    }

    \item $update(k, val)$: Here $Th_x$ wants to update (insert/delete) $k$ with value $val$. So, $Th_x$ will check for the equality of both global, local update and lookup counters, i.e., $(\luc{k}{i} == \guc{k})$ and $(\llc{k}{i} == \glc{k})$. If they are not same then \scv will report the miner as malicious. Otherwise, (i) $Th_x$ will atomically increment $\guc{k}$. (ii) $Th_x$ will increment $\luc{k}{i}$. (iii) Perform the update on the key $k$ with value $val$ on shared memory. 
\end{enumerate}
\vspace{-.23cm}

Once $T_i$ terminates, $Th_x$ will atomically decrements $\guc{k}, \glc{k}$ by the value of $\luc{k}{i}, \llc{k}{i}$, respectively. Then $Th_x$ will reset $\luc{k}{i}, \llc{k}{i}$ to 0. 

The reason for performing these steps and the correctness of the algorithm is as follows: if $Th_x$ is performing a lookup on $k$ then no other thread should be performing an update on $k$. Here, $\guc{k}$ represents the number of updates to $k$ currently executed by all the threads while $\luc{k}{i}$ represents the number of updates to $k$ on behalf of $T_i$ by $Th_x$. Thus the value of $\gucntr$ should be same as $\llcntr$. Otherwise, some other thread is also concurrently performing the updates to $k$. Similarly, if $Th_x$ is performing an update on $k$, then no other thread should be performing an update or lookup on $k$. This can be verified by checking if $\lucntr$, $\llcntr$ are respectively same as $\glcntr$, $\gucntr$. 

\vspace{-.15cm}
\begin{theorem}
	\label{thm:Hmm}
	Smart \Mthr Validator rejects malicious blocks with \blg that allow concurrent execution of dependent SCTs.
\end{theorem}
\vspace{-.15cm}
The same \scv technique can be applied to identify the \emph{faulty bin} error as explained in \secref{intro}. See \apnref{ap-mm} for detailed description along with the pseudo code of smart \mthr validator and \apnref{ap-correctness} for proof of \thmref{Hmm}.

\section{Experimental Evaluation}
\label{sec:result}
The goal of this section is to demonstrates the performance gains by proposed \mthr miner and validator against state-of-the-art miners and validators. To evaluate our approach, we considered Ethereum smart contracts. In Ethereum blockchain, contracts are written in Solidity \cite{Solidity} language and are executed on the \emph{Ethereum Virtual Machine (EVM)} \cite{ethereum:url}. EVM does not support multi-threading \cite{ethereum:url,Dickerson+:ACSC:PODC:2017}. So, we converted the smart contracts of Ethereum as described in Solidity documentation \cite{Solidity} into C++ multi-threaded contracts similar to the approach of \cite{Anjana:OptSC:PDP:2019,Dickerson+:ACSC:PODC:2017}. Then we integrated them into object-based STM framework (\svotm and \mvotm) for concurrent execution of \sctrn{s} by the miner. 

We chose a diverse set of smart contracts described in Solidity documentation \cite{Solidity} as benchmarks to analyze the performance of our proposed approach as was done in \cite{Anjana:OptSC:PDP:2019,Dickerson+:ACSC:PODC:2017}. The selected benchmark contracts are (1) \emph{Coin}: a financial contract, (2) \emph{Ballot}: an electronic voting contract, (3) \emph{Simple Auction}: an auction contract, and (4) finally, a \emph{Mix} contract: combination of three contracts mentioned above in equal proportion in which block consists of multiple \sctrn{s} belonging to different smart contracts and seems more realistic. 



We compared the proposed \svotm and \mvotm miner with state-of-the-art \mthr: \bto \cite{Anjana:OptSC:PDP:2019}, \mvto \cite{Anjana:OptSC:PDP:2019}, Speculative Bin (or \spec) \cite{Vikram&Herlihy:EmpSdy-Con:Tokenomics:2019}, Static Bin (or \stat) \cite{Vikram&Herlihy:EmpSdy-Con:Tokenomics:2019}, and Serial miner.\footnote[2]{Code is available here: https://github.com/PDCRL/ObjSC} We could not compare our work with Dickerson et al. \cite{Dickerson+:ACSC:PODC:2017} as their source code is not available in public domain. We converted the code of \stat and \spec \cite{Vikram&Herlihy:EmpSdy-Con:Tokenomics:2019} from Java to C++ for comparing with our algorithms. 

Concurrent execution of \sctrn{s} by the validator does not use any STM protocol; however it uses the \blg provided by the \mthr miner, which does use STM. To identify malicious miners and prevent any malicious block from being added to the blockchain, we proposed Smart \Mthr Validator (SMV) for \svotm, \mvotm as \svotm SMV, \mvotm SMV. Additionally, we proposed SMV for state-of-the-art validators as  \bto SMV, \mvto SMV, SpecBin SMV, and StaticBin SMV and analysed the performance.

\vspace{1mm}
\noindent
\textbf{Experimental Setup:} The experimental system consists of two sockets, each comprised of 14 cores 2.60 GHz Intel (R) Xeon (R) CPU E5-2690, and each core supports 2 hardware threads. Thus the system supports a total of 56 hardware threads. The machine runs Ubuntu 16.04.2 LTS operating system and has 32GB RAM.


To analyze the performance, we evaluated the speedup achieved by each contract on two workloads. In the first workload (W1), the number of \sctrn{s} varied from 50 to 300 while the number of threads fixed is at 50. The maximum number of \sctrn{s} in a block of Ethereum is approximately 250 \cite{EthereumAvgBlockSize,Dickerson+:ACSC:PODC:2017}, 
but is growing over time. In the second workload (W2), the number of threads varied from 10 to 60, while the number of \sctrn{s} is fixed at 100. The average number of \sctrn{s} in a block of Ethereum is around 100 \cite{EthereumAvgBlockSize}. The hash-table size and shared data-items are fixed to 30 and 500 respectively for both workloads. For accuracy, results are averaged over 26 runs in which the first run is discarded and considered as a warm-up run. The results of serial execution is treated as the baseline for evaluating the speedup. This section describes the detailed analysis for the mix contract and analysis of Coin, Ballot and Simple auction benchmark contracts are in \apnref{ap-rresult}.


\begin{figure}[!t]
	\centering
	{\includegraphics[width=.73\textwidth]{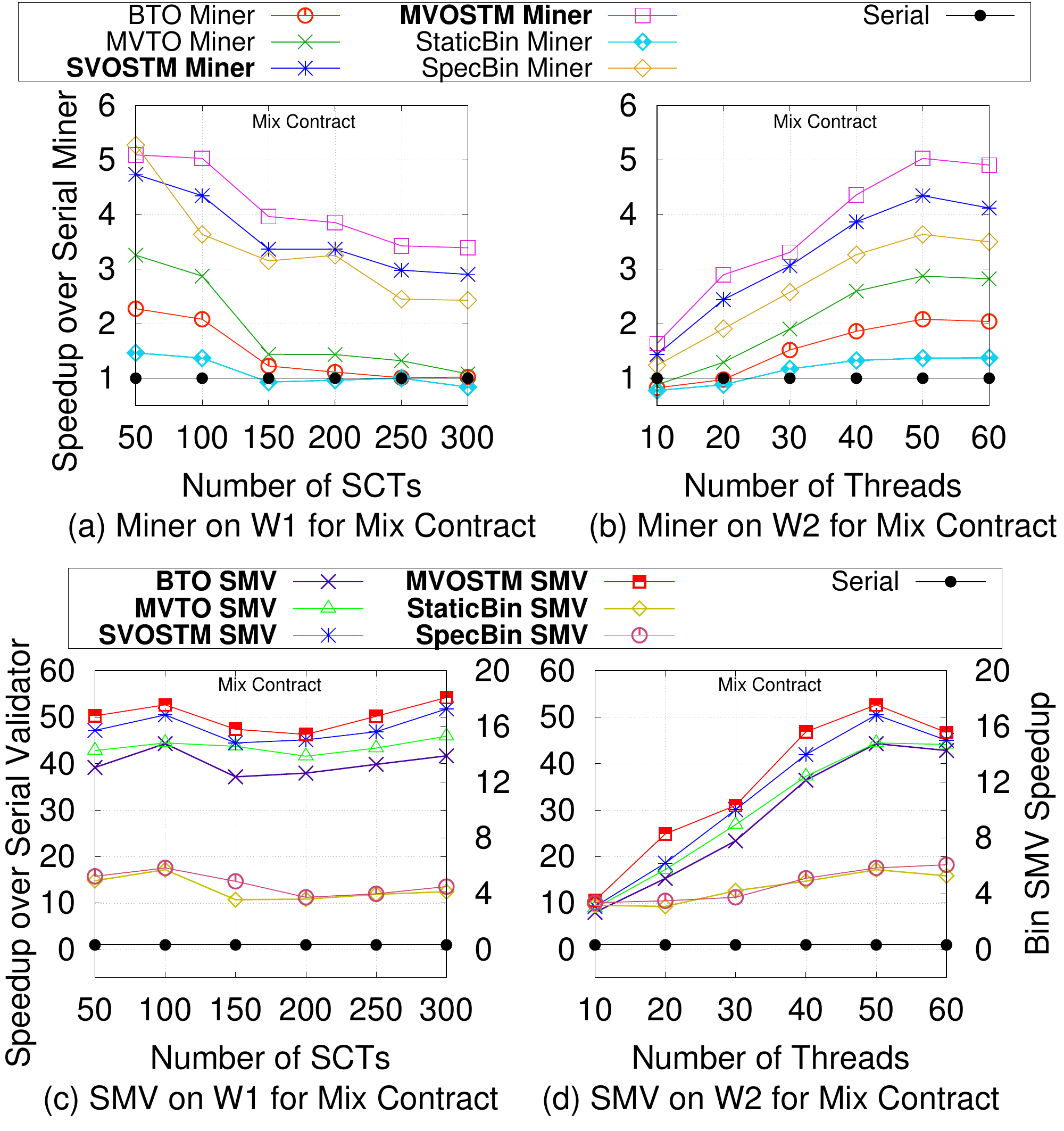}}\vspace{-.35cm}
	\caption{\Mthr and SMVs Speedup over Serial Miner and Validator for Mix Contract on W1 and W2}
	\label{fig:w12-mix}
\end{figure}

\noindent
\textbf{Experimental Results:} \figref{w12-mix} (a) and \figref{w12-mix} (b) show the speedup of \mvotm, \svotm, \mvto, \bto, \spec, and \stat miner over serial miner for mix contract on workloads W1 and W2, respectively.\footnote{In the figures, legend items in bold.} The average speedup achieved by \mvotm, \svotm, \mvto, \bto, \spec, and \stat miner over serial miner is 3.91$\times$, 3.41$\times$, 1.98$\times$, 1.5$\times$, 3.02$\times$, and 1.12$\times$, respectively.

As shown in \figref{w12-mix} (a), increasing the number of \sctrn{s} leads to high contention (because shared data-items are fixed to 500). So the speedup of \mthr miner reduces. \mvotm and \svotm miners outperform \spec miner because \mvotm and \svotm miners use optimistic object-based STMs to execute \sctrn{s} concurrently and construct the \blg whereas \spec uses locks to execute \sctrn{s} concurrently and constructs two bins using the pessimistic approach. \spec miner does not release the locks until the construction of the concurrent bin, which gives less concurrency. However, for the smaller numbers of \sctrn{s} in a block, \spec is slightly better than \mvotm and \svotm miners, which can be observed in the \figref{w12-mix} (a) at 50 \sctrn{s}. \mvotm and \svotm miners outperform \mvto and \bto miners because both of them are consider rwconflicts. It can also be observed that \mvotm miner outperforms all other STM miners as it has fewer conflicts, which gets reflected (see \figref{w1-w2-bg-mix}) as the least number of dependencies in the \blg as compared to other STM miners. For the multi-version (\mvotm and \mvto) miners, we did not limit the number of versions because the number of \sctrn{s} in a block is finite. The speedup by \stat miner is worse than serial miner for more than 100 \sctrn{s} because it takes time for \emph{static conflict prediction} before executing \sctrn{s}.  


\figref{w12-mix} (b) shows that speedup achieved by \mthr miner increases while increasing the number of threads, limited by the number of hardware threads available on the underlying experimental setup. Since, our system has 56 logical threads, the speedup decreases beyond 56 threads. \mvotm miner outperforms all other miners with similar reasoning, as explained for \figref{w12-mix} (a). Another observation is that when the number of threads is less, the serial miner dominates \bto and \mvto miner due to the overhead of the STM system.



The average number of dependencies in \blg by all the STM miners presented in \figref{w1-w2-bg-mix}. It shows that \blg constructed by the \mvotm has the least number of edges for all the contracts on both workloads. However, there is no \blg for bin-based approaches (both \spec and \stat). So, from the block size perspective, bin-based approaches are efficient. But the speedup of the validator obtained by the bin-based approaches is significantly lesser than STM validators. 

\figref{w12-mix} (c) and \figref{w12-mix} (d) show the speedup of Smart \Mthr Validators (SMVs) over serial validator on the workloads W1 and W2, respectively. The average speedup achieved by \mvotm, \svotm, \mvto, \bto, \spec, and \stat decentralized SMVs are 48.45$\times$, 46.35$\times$, 43.89$\times$, 41.44$\times$, 5.39$\times$, and 4.81$\times$ over serial validator, respectively. 


It can be observed that decentralized  \mvotm SMV is best among all other STM validators due to fewer dependencies in the \blg. Though the block size is less in bin-based approaches as compared to STM based approaches due to the absence of \blg, however, STM validators outperform bin-based validators because STM validators precisely determines the concurrent \sctrn{s} based on \blg. In contrast, bin-based validator gives less concurrency using a lock-based pessimistic approach. 

The speedup of SMV is significantly higher than \mthr miner because the miner has to execute the \sctrn{s} concurrently either using STMs (including the retries of aborted transactions) and constructs the \blg or prepare two bins (concurrent and sequential bin using locks in \spec and static analysis in \stat). On the other hand, the validator executes the \sctrn{s} concurrently and deterministically relying on \blg (without any retries) or bins provided by miner.

\begin{figure}[!t]%
	\centering
	{\includegraphics[width=.73\textwidth]{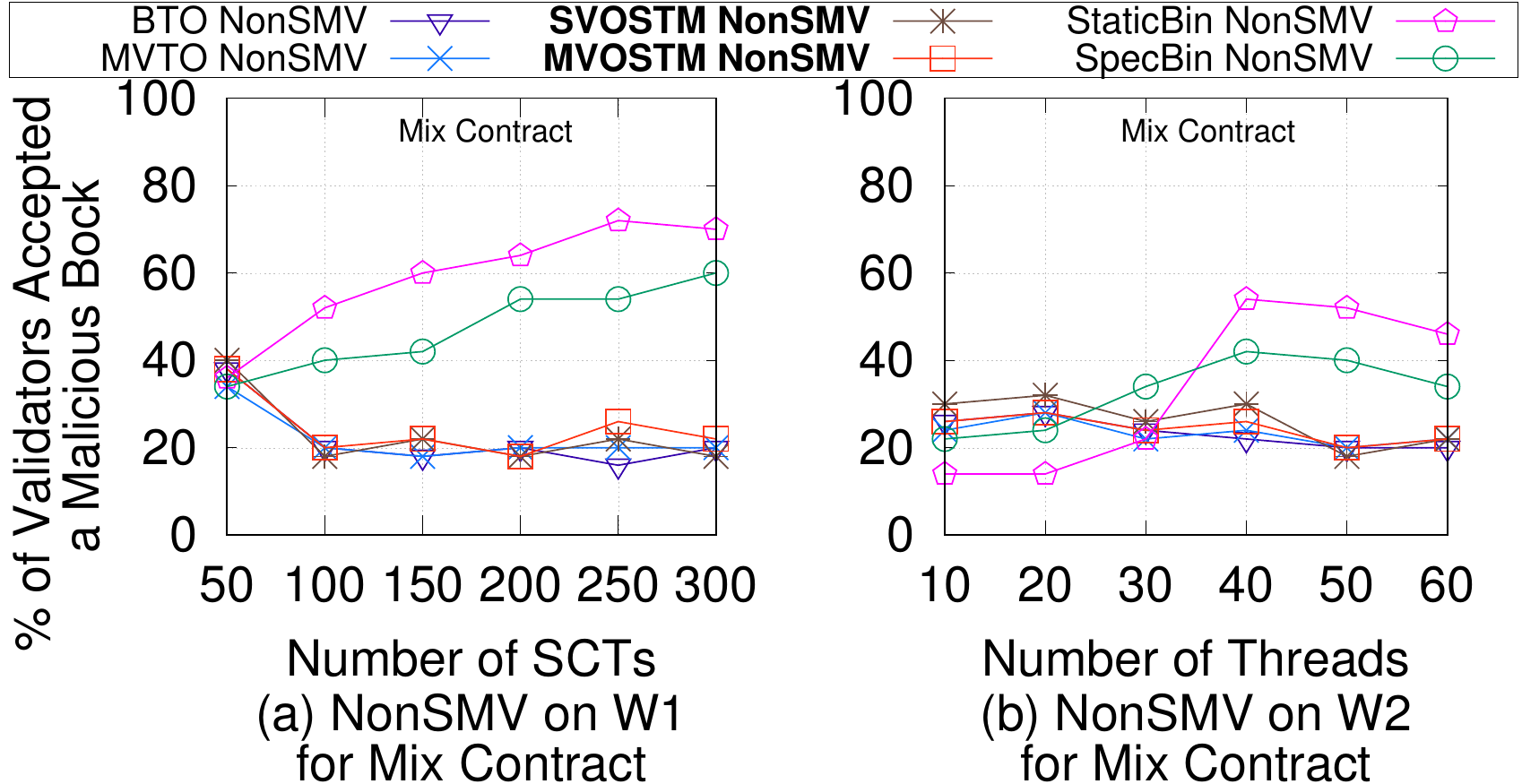}}\vspace{-.35cm}
	\caption{{\% of average \mthr validator (NonSMV) accepted a malicious block for Mix Contract on W1 and W2}}
	\label{fig:w1w2mixmminer}

	{\includegraphics[width=.73\textwidth]{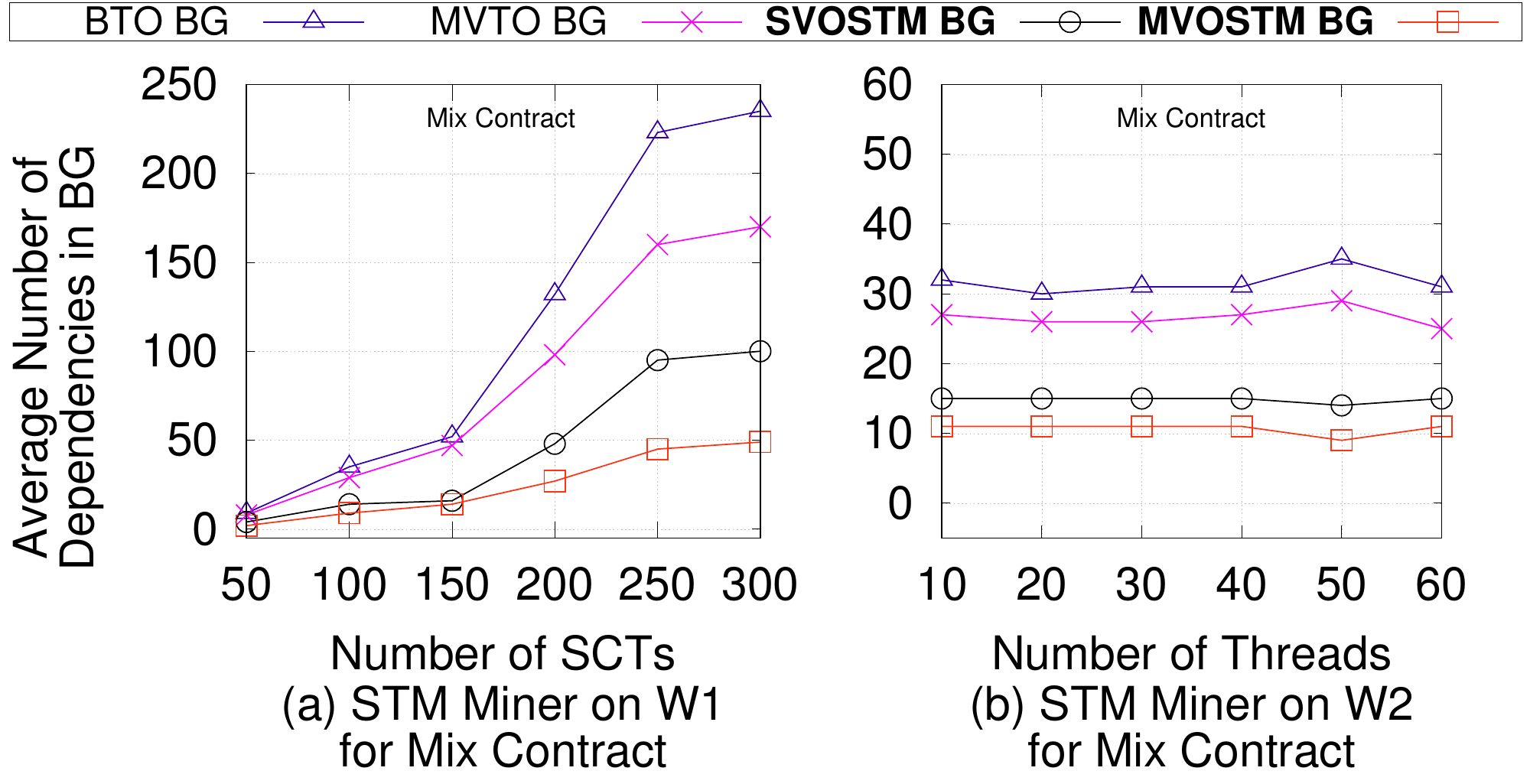}}\vspace{-.35cm}
	\caption{Average Number of Dependencies in BG for Mix Contract on W1 and W2}
	\label{fig:w1-w2-bg-mix}
\end{figure}

A malicious miner may cause either \emb or \fb errors in a block. \figref{w1w2mixmminer} illustrates the percentage of validators without \scv logic embedded, i.e., NonSMVs accepting a malicious block on workloads W1 and W2, respectively. Here, we considered 50 validators and ran the experiments for the mix contract. The \figref{w1w2mixmminer} shows that less than $50$\% of validators (except bin-based NonSMV) accept a malicious block. However, \spec and \stat NonSMVs show more than $50$\% acceptance of malicious blocks. Though, it is to be noted that the acceptance of even a single malicious block result in the blockchain going into inconsistent state.


To solve this problem, we developed a Smart \Mthr Validator (SMV), which identifies the malicious miner (described in \secref{malminer}). We prove that the SMV detects malicious block with the help of $counter$ and rejects it. In fact all the validators shown in \figref{w12-mix} (c) \& (d) are SMV based. Another advantage of SMV is that once it detects a malicious miner during the concurrent execution of \sctrn{s}, it can immediately reject the block and need not execute the remaining \sctrn{s} in the block thus saving time.

\apnref{ap-rresult} presents additional experiments that cover the average number of dependencies in the \blg and additional space required to store the \blg into the block. In addition to W1 and W2, we consider a third workload, W3 in which the number of shared data-items varied from 100 to 600 while the number of threads, \sctrn{s}, and hash-table size is fixed to 50, 100, and 30, respectively. We have shown that the performance of SMV validators for Mix contract on W3 and several other experiments for all the benchmarks. We compared the time taken by the SMV and NonSMV. We analyzed the speedup of fork-join validator for all the three workloads. We showed the actual time (microseconds) taken by all the miners and validators on W1 for the aforementioned four smart contract benchmarks in Tables \ref{tab:w1-coin-miner} - \ref{tab:w1-mix-val}.

\vspace{-.2cm}
\section{Conclusion and Future Directions}
\label{sec:con}

This paper presents a framework for the concurrent execution of smart contracts by miner and validator, which has achieved better performance using object semantics. In blockchains that follow order-execute model \cite{Androulaki+:Hyperledger:Eurosys:2018} such as Ethereum \cite{ethereum:url}, Bitcoin \cite{Nakamoto:Bitcoin:2009}, each Smart Contract Transaction (\sctrn) is executed in two different contexts: first by the \mthr miner to propose a block and later by the \mthr validator to verify the proposed block by the miner as part of the consensus. To avoid FBR errors, the miner on concurrent execution of \sctrn{s} capture the dependencies among them in the form of a \blg as in \cite{Anjana:OptSC:PDP:2019,Dickerson+:ACSC:PODC:2017}. The validator then re-executes the \sctrn{s} concurrently while respecting the dependencies recorded in the \blg to avoid \fbr errors.

The miner executes the \sctrn{s} concurrently using STMs that exploit the object semantics: Single-Version Object-based STM (\svotm) and Multi-Version Object-based STM (\mvotm). The dependencies among the \sctrn{s} collected during this execution are used by the miner threads to construct the \blg concurrently. Due to the use of object semantics, the number of edges in the \blg is smaller, which benefits both miners and validators by enabling them to execute \sctrn{s} quickly in a concurrent setting. 

Another interesting aspect that we considered in this paper is the issue of malicious miners. Suppose that in the \blg approach, a malicious miner proposes an incorrect \blg which does not have all the edges resulting in edge missing \blg (\emb) error. With the bin-based approach, the miner could place the conflicting transactions in the concurrent bin \cite{Vikram&Herlihy:EmpSdy-Con:Tokenomics:2019} resulting in faulty bin (\fb) error. To handle malicious miner, we have proposed a smart \mthr validator (\scv) which can identify these errors and reject the corresponding blocks. 

Proposed \svotm and \mvotm miner achieve on average speedup of 3.41$\times$ and 3.91$\times$ over serial miner respectively. Proposed \svotm and \mvotm decentralized validator outperform with an average speedup of 46.35$\times$ and 48.45$\times$ over serial validator, respectively on Ethereum smart contracts. 

\vspace{1mm}
\noindent
\textbf{Future Directions: } There are several directions for future work. A natural question is whether the size of \blg can become an overhead. Currently, the average number of \sctrn{s} in a block is $\approx$ 100 in Ethereum. So, storing \blg inside the block does not consume much space. The \blg constructed by \mvotm{s} has fewer dependencies as compared with state-of-the-art \sctrn{} execution as shown in \figref{w1-w2-bg-mix}. However, the number of \sctrn{s} in a block can increase over time and as a result the \blg size can grow, and storing it will consume more space. Hence, constructing storage optimal \blg is an interesting challenge. Or achieving the concurrent execution of \sctrn{s} correctly without incurring any extra storage overhead without compromising with the speedup will be another interesting direction. So, a related relevant question is what the optimal storage required for achieving the best possible speedup?

Another interesting research direction is optimizing power consumption. Nowadays, multi-core systems are ubiquitous while serial execution fails to harness the power of multiple cores. So, as discussed in the paper concurrent execution of \sctrn{s} by invoking multiple threads on a multi-core system ensures better performance than serial. But, multi-threading on the multi-core system consumes more power. Additional power is consumed by the multiple miner and validator threads to propose and validate the blocks concurrently. Hence, we would like to explore trade-off between harnessing the number of cores and power consumption. 


Finally, since \emph{Ethereum Virtual Machine (EVM)} \cite{ethereum:url} does not support multi-threading, it is not possible to test the proposed approach on Ethereum. So, another research direction is to design multi-threaded EVM. We plan to test our proposed approach on other blockchains such as Bitcoin \cite{Nakamoto:Bitcoin:2009}, EOS \cite{eos:url} which follow the order-execute model and support multi-threading. 

\bibliographystyle{plain}
\bibliography{citations}

\clearpage
\appendix
\section*{Appendix}
\label{apn:appendix}


\noindent This section is organized as follows:
\begin{table}
\centering
\label{tab:ap-Org}
\resizebox{.9\textwidth}{!}{%
\begin{tabular}{|c |m{11cm}|}
\hline
\textbf{Section No.}  & \textbf{\hspace{5cm}Section Name}\\\hline
        \hspace{.5cm}\apnref{ostm-adv}\hspace{.5cm} & Advantage of \ostm{s} over RWSTMs\\
		\apnref{ap-model} & Remaining System Model\\
		\apnref{ap-pm} & Detailed Proposed Mechanism\\
		\apnref{ap-correctness} & Correctness of BG, \Mthr Miner and Validator\\
		\apnref{ap-rresult} & Detailed Experimental Evaluation\\
		\hline
	\end{tabular}
}	\captionsetup{justification=centering}
	\label{tbl:appen}
\end{table}

\section{Advantage of \ostm{s} over RWSTMs}
\label{apn:ostm-adv}


\cmnt{
\noindent
\subsection{Challenge in Concurrent execution of SCTs}
\label{apn:ap-challenge} 

There are challenges with concurrent execution. \figref{vfalse} illustrates the difficulty. Let us consider two accounts $A, B$ having \$10 as the \emph{Initial State (or IS)}. Suppose there are two \sctrn{s}, $T_1, T_2$ where $T_1$ is transferring \$10 from account $A$ to $B$ while $T_2$ is transferring \$20 from $B$ to $A$. Since both the \sctrn{s} are accessing common accounts (A and B) to transfer the amount, the order of SCTs execution becomes important. Suppose the miner executes them concurrently with the equivalent effect being $T_1$ followed by $T_2$ as shown in \figref{vfalse}(b). In this case, the \emph{Final State (or FS)} of $A$ will have \$20, while $B$ will have \$0. On the other hand, suppose validators execute in a different concurrent order which is equivalent to $T_2$ followed by $T_1$ as shown in \figref{vfalse}(c). $T_2$ executes first, but due to insufficient balance in B's account, a validator, say $v$ rejects this SCT. Then after executing $T_1$, $v$ transfers \$10 from A to B. With this execution, the final state of $A$ will have \$0 while $B$ will have \$20. Thus on receiving such a block, a validator will see that the final state in the block given by the miner is different from what it obtained and hence, falsely reject the block. We refer this problem as \emph{False Block Rejection} (or \emph{\fbr}) error. This can negate the benefits of concurrent executions.

\setlength{\intextsep}{0pt}
\begin{figure*}
	\centering
	\centerline{\scalebox{0.38}{\input{figs/validFalse.pdf_t}}}
	\caption{Challenges in Concurrent execution of SCTs}
	\label{fig:vfalse}
\end{figure*}

}


\setlength{\intextsep}{0pt}
\vspace{4mm}
\begin{figure*}
	\centerline{\scalebox{0.27}{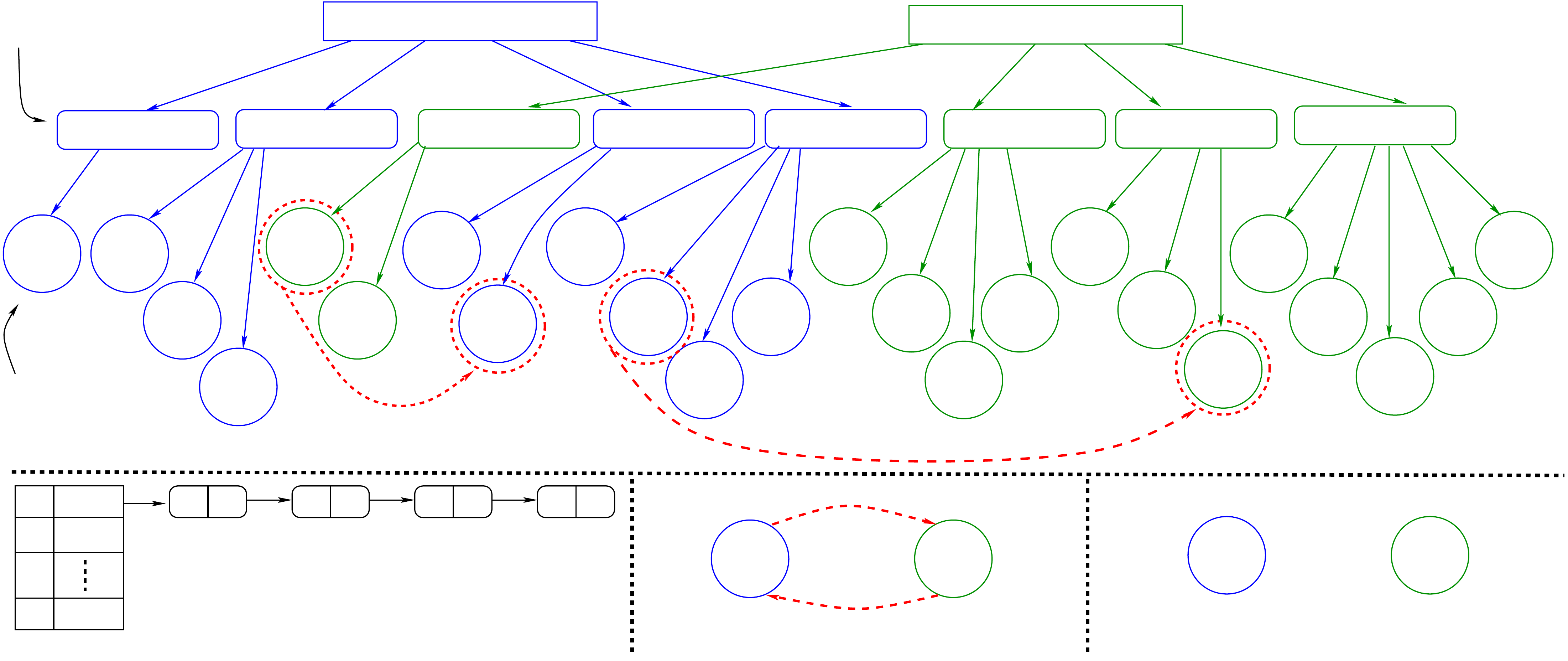}}\vspace{-.25cm}
	 \caption{Advantage of \ostm{s} over \rwstm{s} on \sctrn{s}}
	\label{fig:ex1}
\end{figure*}

We now illustrate the advantage of \ostm{s} over RWSTMs. Consider an \ostm for hash-table which invokes the following methods: (1) \begt{()}, (2) \tlu{($k$)} (or $l(k)$), (3) \tins{($k,v$)} (or $i(k,v)$), (4) \tdel{($k$)} (or $d(k)$), and (5) \tryc{()} explained in \secref{model}

Consider \figref{ex1}, which demonstrates the advantage of \ostm{s} over RWSTMs while executing \sctrn{s} concurrently by multiple miner threads. \figref{ex1} (a) shows two transactions $T_1$ and $T_2$ in the form of a tree structure which is working on a hash-table with $B$ buckets. \figref{ex1} (b) illustrates a bucket of the hash-table with four accounts (shared data-items) $A_1, A_2, A_3$ and $A_4$ which are accessed by these transactions. Accounts are stored in the form of a list. Thus to access account $A_4$, a thread has to accesses $A_1, A_2, A_3$ before access it.


Suppose $T_1$ wants to send \$50 from account $A_1$ to $A_3$ and $T_2$ wants to send \$70 from account $A_2$ to $A_4$. Before performing these transfers, the respective \sctrn{s} verify that each account has sufficient balance. After checking, the \sctrn $T_1$ deletes \$50 from $A_1$ and adds it to $A_3$. At a lower-level, these operations involve reading and writing to both accounts $A_1$ and $A_3$. The execution is shown in \figref{ex1} (a) in form of a tree following the notation used by Weikum et al. \cite[Chap 6]{WeiVoss:TIS:2002:Morg}. Here, level 0 (or $L_0$) shows the operations as read and write while $L_1$ shows higher-level \op{s} insert, delete and lookup. 

\noindent
Consider the execution at $L_0$ of \figref{ex1} (a). The dotted red circles represent conflicting operations: $r_2(A_1)$ conflicts with $w_1(A_1)$ while $r_1(A_2)$ conflicts with $w_2(A_2)$. As a result, this execution cannot be serialized as we cannot find any equivalent serial schedule because of cyclic conflict among $T_1$ and $T_2$ as shown in \figref{ex1} (c). Hence for \sbty \cite{Papad:1979:JACM} (or \opty \cite{GuerKap:Opacity:PPoPP:2008}) either $T_1$ or $T_2$ has to abort. However, execution at level $L_1$ depicts that both transactions are working on different accounts and the higher-level methods (insert and lookup) are isolated. So, we can prune \cite[Chap 6]{WeiVoss:TIS:2002:Morg} this tree and isolate the transaction executions \cite{Peri+:OSTM:Netys:2018} at the higher-level with equivalent serial schedule $T_1T_2$ or $T_2T_1$ as shown in \figref{ex1} (d). Essentially not all the conflicts of lower-level or read-write level matter at higher-level. In a typical execution, \emph{object-conflicts (or \oconf{s}))} \cite{Peri+:OSTM:Netys:2018} are fewer than \emph{read-write conflicts (or \rwconf{s})}. Therefore, \ostm{s} provides greater concurrency while reducing the number of aborts than RWSTMs.

\section{Remaining System Model}
\label{apn:ap-model}

\noindent
This section describes the remaining execution model and the notions of STMs used in this paper. 

\noindent
\textbf{History:} It is a sequence of invocations and responses of different transactional methods. In other words, a \emph{history} \cite{KuzSat:NI:TCS:2016,Papad:1979:JACM} $H$ is a sequence of events represented as $\evts{H}$. $H$ internally invokes multiple transactions by multiple threads concurrently. Each transaction calls higher-level \mth{s}, and each method comprises of read/write events. Here, we consider \emph{sequential history} in which invocation on each transactional method follows the immediate matching response. It helps to make each transactional method as an atomic event. We denote the total order of the transactional method as $<_H$, so history is represented as $\langle \evts{H},<_H \rangle$. 

In this paper, we consider only \emph{well-formed} histories in which a new transaction will not begin until the invocation of previous transaction has not been committed or aborted. History $H$ comprises of the set of transactions as $\txns{H}$. The set of \emph{committed} and \emph{aborted} transactions in $H$ is denoted as $\comm{H}$ and $\aborted{H}$ respectively. So, the set of \emph{incomplete} or \emph{live} transactions in $H$ is represented as $\incomp{H} = \live{H} = (\txns{H}-\comm{H}-\aborted{H})$. 

\noindent
\textbf{Transaction Real-Time Order:} 
Consider two transactions $T_i,T_j \in \txns{H}$, if $T_i$ terminates, i.e. either committed or aborted before \emph{$STM\_begin_j()$} of $T_j$ then $T_i$ and $T_j$ respects real-time \cite{Papad:1979:JACM} order represented as $T_i\prec_H^{RT} T_j$.

\noindent
\textbf{MVSR, VSR, and CSR:} A history $H$ is in \emph{Multi-Version View Serializable (or MVSR)} \cite[Chap. 5]{WeiVoss:TIS:2002:Morg}, if there exists a serial history $S$ such that $S$ is multi-version view equivalent to $H$. It keeps multiple versions with respect to each key. A history $H$ is in \emph{View Serializable (or VSR)} \cite[Chap. 3]{WeiVoss:TIS:2002:Morg}, if there exists a serial history $S$ such that $S$ is view equivalent to $H$. It has shown that verifying the membership of MVSR and VSR in the database is NP-Complete \cite{Papad:1979:JACM}. So, researchers came across with an efficient equivalence notion which is \emph{Conflict Serializable (or CSR)} \cite[Chap. 3]{WeiVoss:TIS:2002:Morg}. It is a sub-class of VSR which uses conflict graph characterization to verify the membership in polynomial time. A history $H$ is in CSR if there exists a serial history $S$ such that $S$ is conflict equivalent to $H$.

\noindent
\textbf{Serializability and Opacity:} Serializability\cite{Papad:1979:JACM} is a popular correctness criteria in databases. But it considers only \emph{committed} transactions. This property is not suitable for STMs. Hence, Guerraoui and Kapalka propose a new correctness criteria opacity \cite{GuerKap:Opacity:PPoPP:2008}  for STMs which considers \emph{aborted} transactions along with \emph{committed} transactions as well. A history $H$ is opaque \cite{GuerKap:Opacity:PPoPP:2008,tm-book}, if there exist an equivalent serial history $S$ with (1) set of events in $S$ and complete history of $H$ are same (2) $S$ satisfies the properties of legal history and (3) The real-time order of $S$ and $H$ are preserved.   

\noindent
\textbf{Linearizability:} A linearizable \cite{HerlihyandWing:1990:LCC:ACM} history $H$ has the following properties: (1) In order to get a valid sequential history, the invocation and response events can be reordered.
(2) The obtained sequential history should satisfy the sequential specification of the objects. (3) The real-time order should respect in sequential reordering as in $H$.

\noindent
\textbf{Lock Freedom:} It is a non-blocking progress property in which if multiple threads are running for a sufficiently long time, then at least one of the threads will always make progress.  Lock-free \cite{HerlihyShavit:Progress:Opodis:2011} guarantees system-wide progress, but individual threads may starve.
\section{Detailed Proposed Mechanism}
\label{apn:ap-pm}
This section describes the data structure and methods of concurrent BG in \apnref{ap-dsBG}. Then we describe the data structure of \svotm and \mvotm in \apnref{ap-sobjds}. Later, we describes the execution of \sctrn{s} by \Mthr Validator rely on the BG provided by the miner in \apnref{ap-cvalidator} and detection of malicious miner by \emph{Smart \Mthr Validator} in \apnref{ap-mm}.

\begin{figure}[t!]
	\centerline{\scalebox{0.35}{\input{figs/graph.pdf_t}}}\vspace{-.25cm}
	\caption{Construction of Block Graph}
	\label{fig:graph}
\end{figure}

\noindent
\subsection{Data Structure of the Block Graph} 
\label{apn:ap-dsBG}
We use \emph{adjacency list} to maintain the Block Graph $BG(V, E)$ inspired from \cite{Anjana:OptSC:PDP:2019,Chatterjee+:NbGraph:ICDCN:2019}. Here $V$ is the set of vertices (\vrtnode{s}) is stored as a vertex list ($\vrtlist$). Similarly E is the set of Edges (\egnode{s}) is stored as edge list ($\eglist$ or conflict list) as shown in the \figref{graph} (a), both $\vrtlist$ and $\eglist$ store between the two sentinel nodes \emph{Head}($-\infty$) and \emph{Tail}($+\infty$). Each \vrtnode{} maintains a tuple: \emph{$\langle$ts, scFun, indegree, egNext, vrtNext$\rangle$}. Here, \emph{ts} is the unique timestamp $i$ of the transaction $T_i$ to which this node corresponds to. scFun is the smart contract function executed by the transaction $T_i$ which is stored in \vrtnode. The number of incoming edges to the transaction $T_i$, i.e. the number of transactions on which $T_i$ depends, is captured by \emph{indegree}. Field \emph{egNext} and \emph{vrtNext} points the next \egnode{} and \vrtnode{} in the $\eglist$ and $\vrtlist$ respectively.      

Each \egnode{} of $T_i$ similarly maintains a tuple: \emph{$\langle$ts, vrtRef, egNext$\rangle$}. Here, \emph{ts} stores the unique timestamp $j$ of $T_j$ which has an edge coming from $T_i$ in the graph. BG maintains the conflict edge from lower timestamp transaction to higher timestamp transaction. This ensures that the \bg is acyclic. The \egnode{s} in $\eglist$ are stored in increasing order of the \emph{ts}. Field \emph{vrtRef} is a \emph{vertex reference pointer} which points to its own \vrtnode{} present in the $\vrtlist$. This reference pointer helps to maintain the \emph{indegree} count of \vrtnode{} efficiently.

\figref{graph} (b) demonstrates the high level overview of BG which consist of three transaction $T_1$, $T_2$ and $T_3$. Here, $T_1$, $T_2$ are in conflict while $T_3$ is independent. The underlying representation of it illustrated in \figref{graph} (a). For each transactions ($T_1, T_2$ and $T_3$) there exists a \vrtnode{} in the $\vrtlist$ of BG along with their conflicts. Since there is en edge from $T_1$ to $T_2$, an \egnode{} corresponding to $T_2$ is in the $\eglist$ of $T_1$. As mentioned earlier, the conflict edges go from lower timestamp to higher timestamp to ensure acyclicity of the \bg. After adding the \egnode, the \emph{indegree} of the \vrtnode of $T_2$ in the $\vrtlist$ is incremented as shown in \figref{graph} (a). 

\noindent
\textbf{Block Graph Library Methods Accessed by \Mthr Miner:} \Mthr miner uses multiple threads to build the \bg. Specifically, the \mthr miner uses two methods to build the \bg: \emph{addVertex()} and \emph{addEdge()}. These two methods are \emph{lock-free} \cite{HerlihyShavit:Progress:Opodis:2011}. Here, \emph{addVertex(i)}, as the names suggests adds a \vrtnode with $ts=i$ for respective $scFun$ to the $vrtList$ of the BG if such a vertex is not already present. This node is atomically added to $vrtList$ using CAS operations.


The $addEdge(u, v)$ \mth creates an \egnode for $v$ in $u$'s \vrtnode if it does not already exist. First, it identifies the \egnode{} in the $\eglist$ of \vrtnode{}. If \egnode{} does not exist then it creates the node and adds into the $\eglist$ of \vrtnode{} atomically using CAS. The edges from $u$ to $v$ captures the conflicts between these transactions. This implies that $v$ is dependent on $u$ and the scFun of $v$ has to be executed only after $u$'s execution. 


\noindent
\textbf{Block Graph Library Methods Accessed by \Mthr Validator:} \Mthr validator uses multiple threads to re-executes the \sctrn{s} concurrently and deterministically with the help of BG given by the \mthr miner. To execute the \sctrn{s}, validator threads use three methods of block graph library: \searchg{()}, \remex{()} and \searchl{()}. First a validator thread $Th_i$ invokes the \searchg() \mth which searches for a \vrtnode{} $n$ in the \bg having $indegree$ 0 (i.e., source node). Such a node corresponds to a \sctrn, which does not depend on other transactions and hence can be executed independently without worrying about synchronization issues. On identifying $n$, $Th_i$ atomically tries to claim it if not already claimed by some other thread. It does this by performing a CAS \op on the $indegree$ to -1. After successful execution of scFun of $n$, $Th_i$ invokes \remex \mth which decrements the \emph{inedgree} count for all the nodes which are have an incoming edge from $n$. This list of nodes is maintained in the \emph{\eglist} of $n$. 

While decrementing the \emph{indegree} count of conflicting nodes if the validator thread $Th_i$ finds any other \vrtnode{} with the \emph{indegree} as 0 then it adds that a reference to that node in its thread-local log $thLog_i$. The $thLog_i$ is used for optimization so that $Th_i$ needs not to search in the global \blg to find the next source node. If a reference to the source node exists in the local log of validator, it is identified by the \searchl{()} \mth. $Th_i$ on identifying such a node $n$, atomically claims $n$ (if not already claimed by another thread). Then it executes the scFun of $n$ and then \remex as explained above. A detailed description of BG methods, along with pseudocode is as follows:


$BG$(\emph{vrtNode}, STM): Miner builds a BG based on \oconf{} given by the STM for all \sctrn{s}. BG takes the \oconf{} from the STM at \Lineref{cg3} for \emph{vrtNode} of \sctrn{}. If $T_i$ have a conflict with $T_j$ then its adds both \sctrn{} \emph{vrtNode} in the BG at \Lineref{cg6} and \Lineref{cg7} using \emph{addVertex()}. To maintain the dependency among the \sctrn{} $T_i$ and $T_j$, the conflict edge goes from lower timestamp transaction ($T_i$) to higher timestamp transaction ($T_i$) to avoid the deadlock. It adds an edge using \emph{addEdge()} at \Lineref{cg9} or \Lineref{cg11}.
\vspace{1mm}
\begin{algorithm}[!htb]
    \scriptsize
    \label{alg:cg} 
    \caption{$BG$(\emph{vrtNode}, STM)}
    \begin{algorithmic}[1]
        \makeatletter\setcounter{ALG@line}{26}\makeatother
        \Procedure{$BG$(\emph{vrtNode}, STM)}{} \label{lin:cg1}
        \State /*STM provides \cl{} of committed transaction $T_i$*/\label{lin:cg2}
        \State conflist = STM\_\gconfl(\emph{vrtNode}.$ts_i$);\label{lin:cg3}
        \State /*$T_i$ $T_j$ are in conflict and $T_j$ exists in conflict list of $T_i$*/\label{lin:cg4}
        \ForAll{($ts_j$ $\in$ conflist)}\label{lin:cg5}
        \State \addv(\emph{$ts_j$}); \label{lin:cg6}
        \State \addv(\emph{vrtNode.$ts_i$});\label{lin:cg7}
        \If{($ts_j$  $>$ \emph{vrtNode}.$ts_i$)}\label{lin:cg8}
        \State \adde(\emph{vrtNode}.$ts_i$, $ts_j$);\label{lin:cg9}
        \Else\label{lin:cg10}
        \State \adde($ts_j$, \emph{vrtNode}.$ts_i$);\label{lin:cg11}
        \EndIf  \label{lin:cg12}
        \EndFor\label{lin:cg13}
        \EndProcedure\label{lin:cg14}
    \end{algorithmic}
\end{algorithm}

\noindent
\addv{($ts_i$)}: This BG method is called by the \mthr miner. First, it identifies the correct location of \vgn{} for transaction $T_i$ in the BG at \Lineref{addv2}. If \vgn{} is not exist in BG then it creates a \vgn node of $T_i$ at \Lineref{addv4}. Finally, It adds the \vgn{} of transaction $T_i$ in the \vl{[]} of BG atomically at \Lineref{addv5} with the help of \emph{compare and swap} operation. If CAS fails then \addv{()} again identifies the location of \vgn{} node in the \vl{[]} with the help of current vertex predecessor node (\emph{vrtpred}) at \Lineref{addv10}. Eventually, \emph{vrtNode} will be the part of BG. This method of the BG is \emph{lock-free}.

\begin{algorithm}[!htb]
    \scriptsize
    \label{alg:addv}     
    \caption{\addv{($ts_i$)}}   
    \begin{algorithmic}[1]
        \makeatletter\setcounter{ALG@line}{40}\makeatother
        \Procedure{\addv{($ts_i$)}}{} \label{lin:addv1}
        \State Search $\langle$\vp, \vc{}$\rangle$ of \vgn{} of $ts_i$ in \vl{[]} of $BG$;\label{lin:addv2}
        \If{(\vc.$ts_i$ $\neq$ \vgn.$ts_i$)}\label{lin:addv3}
        \State Create new BG Node (or \vgn) of $ts_i$ in \vl{[]};\label{lin:addv4}
        \If{(\vp.\vn.CAS(\vc, \vgn))}\label{lin:addv5}
        
        \State /*\vgn{} successfully added in \vl{[]}*/ \label{lin:addv6}
        \State return$\langle$\emph{Vertex added}$\rangle$;  \label{lin:addv7}
        \EndIf\label{lin:addv8}
        \State /*Start with current \vp{} to search the new $\langle$\vp, \vc{}$\rangle$*/\label{lin:addv9}
        \State goto \Lineref{addv2};  \label{lin:addv10}
        \Else\label{lin:addv11}
        \State /*\vgn{} is already exist in \vl{[]}*/ \label{lin:addv12}
        \State return$\langle$\emph{Vertex already exist}$\rangle$; \label{lin:addv13}
        \EndIf\label{lin:addv14}
        \EndProcedure\label{lin:addv15}
    \end{algorithmic}
\end{algorithm}

\noindent
\adde{\emph{($conflictNode_1$, $conflictNode_2$)}}: This BG method is called by the concurrent miner. First, It identifies the location of $conflictNode_2$ in the \el{[]} of $conflictNode_1$ at \Lineref{adde2}. If \egn{} of $conflictNode_2$ is not part of BG then it creates a \egn{} at \Lineref{adde4}. Atomically, it adds an \egn{} in the \el{[]} of $conflictNode_1$ with the help of CAS at \Lineref{adde5}. After successful addition of \egn{} it increments the \emph{indegree} atomically with the help of \egn.\emph{vrtRef} pointer to maintain the dependency of it at \Lineref{adde6}. If CAS fails than \adde{()} again identifies the location of \egn{} node in the \el{[]} with the help of current edge predecessor node (\emph{egPred}) at \Lineref{adde11}. Eventually, \emph{egNode} will be the part of BG. This method of the BG is \emph{lock-free}.

\begin{algorithm}
    \scriptsize
    \label{alg:adde}     
    \caption{\adde{\emph{($conflictNode_1$, $conflictNode_2$)}}}
    \begin{algorithmic}[1]
        \makeatletter\setcounter{ALG@line}{55}\makeatother
        \Procedure{\adde{\emph{($conflictNode_1$, $conflictNode_2$)}}}{}\label{lin:adde1}
        \State Search $\langle$\ep, \ec{}$\rangle$ of \emph{$conflictNode_2$} in \el{[]} of the \emph{$conflictNode_1$} vertex in $BG$;\label{lin:adde2}
        \If{(\ec.$ts_i$ $\neq$ $conflictNode_2$.$ts_i$)}\label{lin:adde3}
        \State Create new BG Node (or \egn) in \el{[]};\label{lin:adde4} 
        \If{(\ep.\en.CAS(\ec, \egn))}\label{lin:adde5}
        \State Increment the \inc{} atomically of \egn.\emph{vrtRef} in \vl{[]};\label{lin:adde6}
        \State /*$conflictNode_2$ is successfully inserted*/\label{lin:adde7}
        \State return$\langle$\emph{Edge added}$\rangle$;  \label{lin:adde8}
        \EndIf\label{lin:adde9}
        \State /*Start with current \ep{} to search the new $\langle$\ep, \ec{}$\rangle$*/\label{lin:adde10}
        \State goto \Lineref{adde4}; \label{lin:adde11}
        \Else\label{lin:adde12}
        \State /*$conflictNode_2$ is already exist in \el{[]}*/\label{lin:adde13}
        \State return$\langle$\emph{Edge already present}$\rangle$; \label{lin:adde14}
        \EndIf\label{lin:adde15}
        \EndProcedure\label{lin:adde16}
    \end{algorithmic}
\end{algorithm}

\noindent
\searchl{($thlog_i$)}: This BG method is called by the \mthr validator. Validator thread identifies the source node in threads local log $thLog_i$ at \Lineref{sl2}. If it finds any source node in $thLog_i$ then it claims that node and atomically sets its \emph{indegree} field to -1 so that no other \mthr validator threads claim this node at \Lineref{sl3}. After claiming of source node it executes smart contract function (associated with the identified source node) using \emph{executeScFun()} at \Lineref{sl6}.

\begin{algorithm}
    \scriptsize
    \caption{\searchl{($thlog_i$)}}
    \begin{algorithmic}[1]
        \makeatletter\setcounter{ALG@line}{71}\makeatother
        \Procedure{\searchl{($thlog_i$)}}{} \label{lin:sl1}
        \State Identify local log vertex($llVertex$) with indegree 0 in $thLog_i$.\label{lin:sl2}
        \If{($llVertex$.\inc.CAS(0, -1))} \label{lin:sl3} 
        \State \nc{} $\gets$ \nc{}.$get\&Inc()$; \label{lin:sl4}
        \State /*Concurrently execute \sctrn{} corresponds to $llVertex$*/.\label{lin:sl5}
        \State \exec{($llVertex$.scFun)}.\label{lin:sl6}
        \State return$\langle$ $llVertex$ $\rangle$;\label{lin:sl7}
        \Else\label{lin:sl8}
        \State return$\langle nil \rangle$;\label{lin:sl9}
        \EndIf\label{lin:sl10}
        \EndProcedure \label{lin:sl11}
    \end{algorithmic}
\end{algorithm}

\noindent
\searchg{(BG)}: The \mthr validator calls this BG method. Validator thread identifies the source node (with indegree 0) in BG at \Lineref{sg4}. If it finds any source node in BG, then it claims that node and atomically sets its \emph{indegree} field to -1 so that no other \mthr validator threads claim this node at \Lineref{sg5}. After claiming of source node it executes smart contract function (associated with the identified source node) using \emph{executeScFun()} at \Lineref{sg8}.

\begin{algorithm}
    \scriptsize
    \caption{\searchg{(BG)}}
    \begin{algorithmic}[1]
        \makeatletter\setcounter{ALG@line}{82}\makeatother
        \Procedure{\searchg{(BG)}}{} \label{lin:sg1}
        \State \vnode{} $\gets$ BG.\vh;    /*Start from the Head of the list*/\label{lin:sg2}
        \State /*Identify the \emph{vrtNode} with indegree 0 in BG*/\label{lin:sg3}
        \While{(\vnode.\vn{} $\neq$ BG.\vt)} \label{lin:sg4}
        \If{(\vnode.\inc.CAS(0, -1))}\label{lin:sg5}
        \State \nc{} $\gets$ \nc{}.$get\&Inc()$; \label{lin:sg6}
        \State /*Concurrently execute \sctrn{} corresponds to \emph{Node}*/.\label{lin:sg7}
        \State \exec{(\emph{vrtNode}.scFun)}.\label{lin:sg8}
        \State return$\langle \vnode \rangle$;\label{lin:sg9}
        \EndIf\label{lin:sg10}
        \State \vnode $\gets$ \vnode.\vn;    \label{lin:sg11}
        \EndWhile\label{lin:sg12}
        \State return$\langle nil \rangle$;\label{lin:sg13}
        \EndProcedure\label{lin:sg14}
    \end{algorithmic}
\end{algorithm}

\noindent
\emph{remExecNode(removeNode)}: This BG method is called by the \mthr validator. It atomically decrements the \emph{indegree} for each conflicting node of source node with the help of vertex reference pointer (\emph{vrtRef}) at \Lineref{ren3}. \emph{vrtRef} pointer helps to decrement the \emph{indegree} count of conflicting node efficiently because thread need not to travels from head of the \vl{[]} to identify the \vgn{} node for decrementing the \emph{indegree} of it. With the help of \emph{vrtRef}, it directly decrements the \emph{indegree} of \vgn. While decrementing the \emph{indegree} of \vgn{} if it identifies the new source node than it keeps that node information in thread local log $thLog_i$ at \Lineref{ren5}.

\begin{algorithm}
    \scriptsize
    \caption{remExecNode(removeNode)}
    \begin{algorithmic}[1]
        \makeatletter\setcounter{ALG@line}{96}\makeatother
        \Procedure{\emph{remExecNode(removeNode)}}{} \label{lin:ren1}
        \While{(removeNode.\en $\neq$ removeNode.\et)}    \label{lin:ren2}
        \State Atomically decrement the \emph{indegree} of conflicting node using removeNode.\emph{vrtRef} pointer.\label{lin:ren3} 
        \If{(removeNode.\emph{vrtRef}.\inc{} == 0)}\label{lin:ren4}
        \State Add removeNode.\emph{vrtRef} node into $thLog_i$.\label{lin:ren5}
        \EndIf\label{lin:ren6}
        \State removeNode $\gets$ removeNode.\en.\emph{vrtRef};    \label{lin:ren7}    
        \EndWhile \label{lin:ren8}
        \State return$\langle nil \rangle$; \label{lin:ren9}
        \EndProcedure    \label{lin:ren10}    
    \end{algorithmic}
\end{algorithm}

\noindent
\exec{(scFun)}: It executes the \sctrn{s} concurrently without the help of concurrency control protocol. First, it identifies the smart contract function (scFun) steps and executes them one after another at \Lineref{ex3}. If the current step (curStep) is lookup on key $k$, then it lookup the shared data item for key $k$ from the shared memory at \Lineref{ex5}. If curStep is insert on key $k$ with value as $v$, then it inserts the shared data item of key $k$ with value $v$ in the shared memory at \Lineref{ex7}. If curStep is delete on key $k$, then it deletes the shared data item of key $k$ from the shared memory at \Lineref{ex9}. All these curStep of scFun can run concurrently with the other validator threads because only non-conflicting transactions will execute concurrently with the help of BG given by the \mthr miner.
\vspace{2mm}
\begin{algorithm} 
    \scriptsize
    \label{alg:exec} 
    \caption{\exec{(scFun)}}
    \begin{algorithmic}[1]
        \makeatletter\setcounter{ALG@line}{106}\makeatother
        \Procedure{\exec{($scFun$)}}{} \label{lin:ex1}
        \While{(scFun.steps.hasNext())} /*scFun is a list of steps*/ \label{lin:ex2}
        \State curStep = scFun.steps.next(). /*Get next step to execute*/\label{lin:ex3}
        \Switch{(curStep)}
        \EndSwitch
        \Case{lookup($k$):}\label{lin:ex4}
        \State Lookup $k$ from a shared memory.\label{lin:ex5}
        \EndCase
        \Case{insert($k, v$):} \label{lin:ex6}
        \State Insert $k$ in shared memory with value $v$.\label{lin:ex7}
        \EndCase
        \Case{delete($k$):}\label{lin:ex8}
        \State Delete $k$ from shared memory.\label{lin:ex9}
        \EndCase
        \Case{default:}
        curStep is not lookup, insert and delete;\label{lin:ex10}
        \EndCase \label{lin:ex11}
        \EndWhile            \label{lin:ex12}
        \State return $\langle void \rangle$    \label{lin:ex13}
        \EndProcedure            \label{lin:ex14}
    \end{algorithmic}
\end{algorithm}


\cmnt{
\subsection{Concurrent Miner}
\label{apn:ap-cminer}
This section describes the efficient and concurrent execution of \sctrn{s} by miner using object semantics. The concurrent miner gets the \sctrn{s} from multiple clients to execute them and proposes a block. To utilize the multi-core resources, miner executes the non-conflicting \sctrn{s} concurrently. However, identifying the non-conflicting \sctrn{s} are not straightforward because smart contracts are written in Turing-complete language. Therefore, a concurrent miner takes the help of optimistic STMs to execute the \sctrn{s} concurrently. STMs are a convenient programming interface for programmers to access the shared memory concurrently using multiple threads without worrying about concurrency issues such as deadlock, livelock, priority inversion, etc. To achieve the greater concurrency miner uses \emph{Single-Version Object-based STM (\svotm)} protocol to execute the \sctrn{s} concurrently.

High level overview of \algoref{cminer} (\subsecref{cminer}) illustrates the execution of \sctrn{s} by concurrent miner using $m$ threads. 
A thread $Th_x$, first reads the current value of global index $\gind$ into a local value current index $curInd$ and increments $\gind$ atomically in \Lineref{index}. At \Lineref{sctrn}, $Th_x$ gets a $curTrn$ from list of \sctrn (or sctList[]). $Th_x$ gets the unique timestamp by STM\_begin() at \Lineref{beg-tx}. \figref{cminer} shows the working of concurrent miner in which clients 1, 2 and 3 sends the \sctrn $T_1, T_2$ and $T_3$ to concurrent miner. \figref{cminer}.(a) shows the concurrent execution of sctList[] which consists of three transactions $T_1, T_2$ and $T_3$ executing by three threads using STMs. Transaction $T_1$ runs a smart contract function (or scFun) of $curTrn$ as $send(A,B,\$20)$ to transfer \$20 from account A to B. Initially, each account has \$100 as initial state (IS) shown in \figref{cminer}.(b). \emph{Send} is a function of coin contract from Solidity Documentation \cite{Solidity} which consists of a set of steps to transfer the money from one account to other account. Before performing the transfers, it verifies that transferring account has sufficient balance or not. If balance is sufficient then it deletes the amount from one account and adds it to other account.

\begin{figure*}
	\centerline{\scalebox{0.4}{\input{figs/cminer.pdf_t}}}
	 \caption{Working of Concurrent Miner}
	\label{fig:cminer}
\end{figure*}

If current step (curStep) is lookup for shared data item $k$ at \Lineref{case-lookup} then miner thread takes the help of STM\_lookup($k$) to execute it concurrently. STM\_lookup($k$) identifies the location of $k$ in the shared memory and return the value as $v$ at \Lineref{lu-ret}. If $v$ is $abort$ means inconsistency has been occurred then  thread needs to restart from the STM\_begin() and get the new timestamp at \Lineref{beg-tx} to execute the $curTrn$ again else take the next curStep until the $curTrn$ executed successfully.

If curStep is insert for shared data item $k$ then miner thread takes the help of STM\_insert($k,v$) to execute it concurrently at \Lineref{stm-ins}. Concurrent miner uses optimistic \svotm{s} so effect of the STM\_insert($k,v$) will be visible in shared memory after successful STM\_tryC() at \Lineref{tryC25}. If curStep is delete for shared data item $k$ then miner thread takes the help of STM\_delete($k$) to execute it concurrently at \Lineref{del-ret}. Actually deletion of $k$ will happen after successful STM\_tryC(). 

After completion of all the curStep of scFun of $curTrn$, STM validates the scFun in STM\_tryC() at \Lineref{tryC25}. If validation will be successful for all the update methods (STM\_insert() and STM\_delete()) exist in the scFun then update the final state ($FS_m$) corresponding to each shared data item (or account) in the shared memory as shown in \figref{cminer}.(b). Otherwise, it returns $abort$ at \Lineref{tryC} and restart from the \Lineref{beg-tx} until it $commits$. After successful STM\_try(), \svotm maintains the \oconf{s} corresponding to each shared data item which helps the miner to build the BG concurrently. The internal details of maintaining the conflicts is shown in \apnref{ap-sobjds}. Similarly other threads execute the transactions $T_2$ and $T_3$ concurrently with the help of \svotm as shown in \figref{cminer}.(a) and update the $FS_m$ for each account in \figref{cminer}.(b).


After successful execution of each scFun miner thread creates a vertex node \emph{vrtNode} of \sctrn and add the dependencies  based on \oconf{s} provided by the \svotm to maintain the BG (at \Lineref{addv}) as shown in \figref{cminer}.(c). To build the BG, concurrent miner uses the \emph{addVertex()} and \emph{addEdge()} methods of block graph library as explained in \subsecref{bg}. 

Once the sctList[] has been executed successfully and done with the formation of BG, then miner computes the hash of the previous block as depicted in \figref{cminer}.(d). Eventually, concurrent miner proposes a block which consists of all the \sctrn{s}, BG, final state ($FS_m$) of all the shared data items and hash of the previous block as shown in \figref{cminer}.(e) and broadcast to other existing nodes in the distributed system to validate the proposed block. 

To achieve the greater concurrency further, concurrent miner uses \emph{Multi-Version Object-based STM (\mvotm)} \cite{Juyal+:MVOSTM:SSS:2018} protocol instead of \svotm protocol to execute the \sctrn{s} concurrently. \mvotm maintains multiple versions corresponding to each shared data item. 
Concurrent miner uses \mvotm to capture the \emph{multi-version \oconf{s}} (defined in \apnref{ap-model}) to construct the BG. The internal details of maintaining the conflicts is shown in \apnref{ap-sobjds}. It maintains the less number of dependency in the BG as compared to \svotm. So, construction of the BG by concurrent miner using \mvotm is faster than \svotm. Hence, concurrent execution of \sctrn{s} by miner using \mvotm reduces the number of aborts and surpasses the efficiency. Later, concurrent validators will also re-execute more \sctrn{s} concurrently and ensure better performance because of the lesser number of \mvoconf in the BG. 
\cmnt{
\setlength{\intextsep}{0pt}

\begin{algorithm}[H]
	\scriptsize 
	\caption{\cminer{(\sct, STM)}: $m$ threads concurrently execute the \sctrn{s} from \sct{} with the help of Optimistic STM.}
	\label{alg:cminer}	
	\begin{algorithmic}[1]
		\makeatletter\setcounter{ALG@line}{0}\makeatother
		\Procedure{\emph{\cminer{ (\sct, STM)}}}{}\label{lin:cminer1}
		\State /*Get the smart contract function (scFun) from \sct{} and atomically increment the \blank{.25cm} curInd to point the next scFun*/\label{lin:cminer2}
		\State scFun = \emph{curInd.get\&Inc}(\sct);  \label{lin:cminer3}
		\ForAll{(scFun $\in$ \sct)}\label{lin:cminer4}
		\State $T_i$ = \begtrans{()}; /*Get the unique timestamp $i$ for $T_i$*/\label{lin:cminer5}
		\ForAll{(curStep $\in$ scFun)} /*scFun is a list of steps*/\label{lin:cminer6}
		\If{(curStep == lookup($k$))}\label{lin:cminer8}
		\State /*Lookup data item $k$ from a shared memory*/
		\State $v$ $\gets$ STM\_$lookup_i${($k$)};\label{lin:cminer10} 
		\If{($v$ == $abort$)} goto \Lineref{cminer5}; \label{lin:cminer11}
		\EndIf\label{lin:cminer12}
		\ElsIf{(curStep == insert($k, v$))} \label{lin:cminer13}
		\State /*Insert $k$ into $T_i$ local memory with value $v$*/
		\State STM\_$insert_i$($k, v$);\label{lin:cminer14} 
		\ElsIf{(curStep == delete($k$))} \label{lin:cminer15}
		\State /*Actual deletion of $k$ happens in STM\_tryC() */
		\State STM\_$delete_i$($k$); \label{lin:cminer16}
		\Else{} 
		curStep is not lookup, insert and delete.\label{lin:cminer18}
		\EndIf
		\EndFor\label{lin:cminer19}
		\State $v$ $\gets$ STM\_\emph{$tryC_i()$}; /*Try to commit the transaction $T_i$*/\label{lin:cminer20}
		\If{($v == abort$)} goto \Lineref{cminer5};\label{lin:cminer21}
		\EndIf			\label{lin:cminer22}
		\State Create vertex node \vrtnode{} with $\langle$\emph{$i$, scFun, 0, nil, nil}$\rangle$ as a vertex of BG;	\label{lin:cminer23}
		\State $BG_i$(\emph{vrtNode}, STM); /*Build BG with conflicts of $T_i$*/		\label{lin:cminer24}
		\State scFun $\gets$ \emph{curInd.get\&Inc}(\sct); 
		\EndFor	\label{lin:cminer26}
		\EndProcedure\label{lin:cminer27}
		
	\end{algorithmic}
\end{algorithm}
}

\cmnt{
\noindent
Now, we have the following theorems
\vspace{-.1cm}
\begin{theorem}
	\label{thm:BGdep}
	All the dependencies between the conflicting nodes are captured in the BG.
\end{theorem}
\vspace{-.1cm}
\begin{theorem}
	\label{thm:co-ostm}
	Any history $H_m$ generated by concurrent miner using \svotm satisfies co-opacity.
\end{theorem}
\vspace{-.1cm}
\begin{theorem}
	\label{thm:o-mvostm}
	Any history $H_m$ generated by concurrent miner using \mvotm satisfies opacity.
\end{theorem}

\vspace{-.1cm}
\noindent
Please refer the \apnref{ap-correctness} for proof of \thmref{BGdep}, \thmref{co-ostm}, and \thmref{o-mvostm}.
}
}

\subsection{Data Structure of  \svotm and \mvotm}
\label{apn:ap-sobjds}
\cmnt{
\begin{figure}[H]
    \centering
    \includegraphics[width=.7\textwidth]{figs/ostmCL.pdf_t}
     \caption{Underlying Data Structure of Shared Data Items in \svotm}
    \label{fig:dsOSTM}
\end{figure}

\begin{figure}[H]
    \centering
    \includegraphics[width=.8\textwidth]{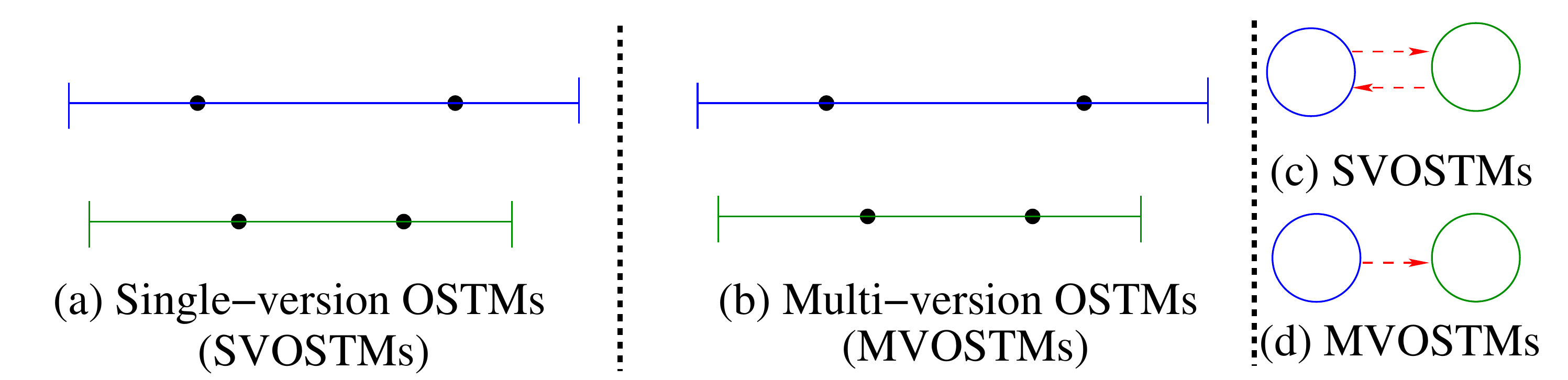}
    \caption{Underlying Data Structure of Shared Data Items in \mvotm}
    \label{fig:dsMVOSTM}
\end{figure}

}

\noindent 
This subsection describes the internal details about the data structure used to store shared data items in \svotm and \mvotm. 

As shown in \figref{ex1} (b) of \apnref{ostm-adv}, we have used a hash-table with a fixed size of buckets, where each bucket consists of a list of corresponding shared data items. The data structure used to store the data items depends on the protocol (\svotm, \mvotm).

\begin{figure}[!t]
    \centerline{\scalebox{0.38}{\input{figs/ostmCL.pdf_t}}}\vspace{-.25cm}
     \caption{Underlying Data Structure of Shared Data Items in \svotm}
\label{fig:confgostm}
\end{figure}

\vspace{1mm}
\noindent 
\textbf{Data Structure of \svotm: }\figref{confgostm} (a) demonstrates the structure of a shared data item in \svotm. Each shared data item consist of eight fields $\langle Account,$ $val,$ $Lock, max_L,$ $max_U, cL_{list}[], cU_{list}[], next\rangle$. Where \emph{Account} is a unique identifier which represents the shared data item or \emph{key} or \emph{account (e.g., $A_1$)}, \emph{val} field stores the value corresponding to data item. $Lock$ is use to provide synchronization among the operations of different transactions working on same data items. $Lock$ is acquired by the transaction before updating (inserting/deleting) the shared data item and released after the successful execution. Next two fields, i.e., \emph{max$_L$} and \emph{max$_U$} are counter variable 

initialize as 0. Whenever a transaction performs any operation (STM\_lookup()/ STM\_insert()/ STM\_delete()) on the shared data item, they update the corresponding  value of the \emph{max$_L$/max$_U$}. If current value of \emph{max$_L$/max$_U$} is smaller than the timestamp of the transaction $i$ then it updates \emph{max$_L$/max$_U$} with the $i$. Otherwise, transaction with timestamp $i$ returns abort and retry again. Field \emph{cL$_{list}[]$} and {cU$_{list}[]$} store the timestamp of all committed transactions (transaction ids) who has performed the STM\_lookup() and update operations (STM\_insert()/ STM\_delete()) on the shared data item respectively. These fields are used to generate the conflicts (dependencies) between the transactions. Finally, the \emph{next} field is point to the next shared data item in the list of the respective bucket. The shared data item in the list of the corresponding bucket is stored in the increasing order of the keys.

To understand conflicts generation for concurrent execution of SCTs in \svotm, we consider three transactions T$_0$, T$_5$, and T$_7$. In  \svotm, a transaction T$_i$ conflict with transaction T$_j$, if both are accessing a common data item $k$ and at least one of them is update operation (insert or delete). The conflicts of  \svotm is defined in \secref{bg}.

\cmnt{
\begin{enumerate}
    \item if T$_i$ is performing \emph{\tlu{(k)}} and  T$_j$ is performing \emph{\tins{(k)}} (or \tdel{(k)}) then there will be an edge from T$_i$ to T$_j$.
    \item if T$_i$ is performing \emph{\tins{(k)}} (or \tdel{(k)})  and  T$_j$ is performing \emph{\tlu{(k)}} then there will be an edge from T$_j$ to T$_i$.
    \item if T$_i$ is performing \emph{\tins{(k)}} (or  \tdel{(k)}) and  T$_j$ is performing \emph{\tins{(k)}} (or \tdel{(k)}) then there will be an edge from T$_i$ to T$_j$.
\end{enumerate}
}

\figref{confgostm} (b) shows the timeline view and respective operations of these transactions. Here, T$_0$ performed insert operation on shared data item as account $A_1$ with value $v_0$ (i.e., $i_{0}\langle A_1, v_{0}\rangle$), and committed successfully. Therefore, the timestamp (ts) of $T_0$ is inserted in {cU$_{list}[]$} as shown in \figref{confgostm} (a). Similarly, T$_5$, and T$_7$ are two concurrent execution performing insert (i.e., $i_{5}\langle A_1, v_{5}\rangle$) and lookup (i.e., $l_{7}\langle A_1, v_{5}\rangle$) operation respectively. Since T$_5$ performed insert on account $A_1$, the $max_U$ field is set to $5$ and later it committed successfully so its transaction id $T_5$ is inserted in {cU$_{list}[]$}. Further, T$_7$ performed lookup on $A_1$ so $max_L$ field is set to $7$ and committed successfully so its transaction id $T_7$ is inserted in {cL$_{list}[]$}.

\figref{confgostm} (c) illustrates the transactions conflict list. As shown in the concurrent execution, T$_0$ is the first transaction, so it does not find any conflict with other transactions; hence, its conflict list is empty. Later, T$_5$ committed, so T$_5$ conflict list consists of T$_0$ since both T$_0$ and T$_5$ performed update operation on $A_1$. A transaction which performs update operation (insert or delete) conflicts with all the transactions present in both {cU$_{list}[]$} and {cL$_{list}[]$} lists corresponding to shared data item. At the commit time of $T_5$, {cL$_{list}[]$} was empty so, $T_5$ conflicts with $T_0$ only. Finally, transaction T$_7$ committed with the lookup on $A_1$.  A transaction which performs lookup operation, conflicts with all the transactions present in the {cU$_{list}[]$} list. So, $T_7$ conflict list consists all the transactions present in {cU$_{list}[]$} which is T$_0$ and T$_5$. Hence, T$_7$ conflicts with T$_0$ and T$_5$.

\vspace{1mm}
\noindent 
\textbf{Data Structure of \mvotm: } \svotm stores only one version corresponding to each data item; however,  \mvotm maintains multiple versions. In the proposed framework, we have a fixed number of SCTs in each block, so we do not restrict the number of versions with each shared data item. \figref{confgmvostm} (a) shows the data structure used to store the shared data item in  \mvotm protocol. Here, each shared data item consists of four fields as $\langle Account, Lock, vl, next \rangle$. Where $Account$, \emph{Lock}, and \emph{next} field are same as defined earlier for \svotm. A new field \emph{vl} stands for version list, which maintains version created by update operations (insert and delete) on the shared data item.

To store the versions of the shared data item, we used a list (\emph{version list or vl)}. 
Here, each entry of the version list consists of five fields $\langle ts, val, max_L, rvl, vNext\rangle$. It stores the version in increasing order of transaction's timestamps. The first field \emph{ts} shows the timestamp of the transaction which created this version (see \figref{confgmvostm} (a)). The next field is \emph{val}, which stores the value corresponding to that version. Field \emph{max$_L$} is used to store the maximum id/ts of the transaction that has performed lookup on this version. A transaction $T_i$ looks up from the version $j$ such that $j$ is \emph{largest version timestamp} smaller than $i$. The \emph{rvl[]} stands for \emph{return value list}, which stores the timestamp of the committed transactions that have lookup from a particular version. Finally, the last field \emph{vNext} is used to store a pointer to the next available version in the version list. As \figref{confgmvostm} (a) illustrates account $A_1$ maintains of three versions as 0, 5, and 10. \emph{Version 0, 5, and 10} is created by transaction T$_0$, T$_5$, and T$_{10}$ respectively.

\begin{figure}[!t]
    \centerline{\scalebox{0.38}{\input{figs/mvostmCL.pdf_t}}}\vspace{-.25cm}
     \caption{Underlying Data Structure of Shared Data Items in \mvotm}
    \label{fig:confgmvostm}
\end{figure}

Consider \figref{confgmvostm} (b) to understand conflicts generation in \mvotm, which demonstrates the timeline view and respective operations of four transactions T$_0$, T$_5$, T$_7$, and T$_{10}$. Here, transactions T$_0$, T$_5$, and T$_{10}$ perform insert operation while T$_{7}$ performs a lookup on account A$_1$. Due to insert operation by T$_0$, T$_5$, and T$_{10}$ three different versions of account A$_1$ has been created, as shown in \figref{confgmvostm} (a). Here, transaction T$_0$ executed first, so it created version 0 and then T$_5$ performed insert operation ($i_{5}\langle A_1, v_{5}\rangle $); therefore version 5 is created. After that T$_7$ performed lookup on A$_1$ which returns the value as $v_5$ (i.e., $l_{7}\langle A_1, v_5\rangle$). After the successful commit of T$_7$, it inserts its ts $7$ in \emph{max$_L$} of version 5 as shown in \figref{confgmvostm} (a). Here, transaction T$_{10}$ began after the beginning of T$_{7}$ and committed before T$_{7}$ but still transaction T$_{7}$ is allowed to commit. Due to multiple versions, T$_{7}$ finds the older value of A$_1$ as $v_5$ created by $T_5$ and hence not abort; otherwise, in \svotm transaction T$_{7}$ has to return abort. 

\figref{confgmvostm} (c) Illustrates the conflict list of transactions. Here, T$_0$ is the first transaction and created version 0 of A$_1$, so it does not conflict with any other transaction; hence, its conflict list is empty. Next, transaction T$_5$, which created a new version of A$_1$ and committed successfully, so T$_5$ conflict list consists of T$_0$. So, while generating conflict list, for an update (insert or delete) operation on account A$_1$ transaction first checks if the \emph{rvl[]} list is empty for the largest version smaller than its ts (transaction ts), then if the list is empty it adds that version ts in its conflict list otherwise adds all the ts in \emph{rvl[]} list of that version. For a lookup operation, a transaction adds the ts/id of the version which it has looked up and also the next version in the version list (if available) in its conflict list. Next is transaction T$_{10}$  since it committed before T$_7$, so in T$_{10}$ conflict list, T$_{5}$ is added. The reason why only T$_{5}$ and not T$_7$, this is because \emph{rvl[]} list consist only committed transaction ts and at the time when T$_{10}$ committed T$_{7}$ was still live and not yet committed so \emph{rvl[]}  of T$_{5}$ was empty. Finally, T$_7$ committed and as it performed a lookup on account A$_1$ from version 5, so it adds T$_5$ and the next version ts, which is T$_{10}$ in its conflict list. Hence, T$_7$ conflicts with T$_5$ and T$_{10}$.

\subsection{\Mthr Validator}
\label{apn:ap-cvalidator}

\Mthr validator re-executes the \sctrn{s} concurrently and deterministically rely on the BG provided by the \mthr miner. To access the BG, validator uses \emph{globalSearch()}, \emph{localSearch()}, and \emph{remExNode()} methods of block graph library. The descriptions of all these methods are given in \apnref{ap-dsBG}. 

High level overview of \algoref{val} shows the execution of \sctrn{s} by \mthr validator with the help of BG. First, multiple validator threads concurrently identify the source node (\emph{indegree} 0) in the BG using \emph{globalSearch()} at \Lineref{val5}. After identifying the source node, thread claims it (sets \emph{indegree} to -1) atomically so that other \mthr validator threads can not claim it. Then it executes the scFun of \sctrn corresponding to the source node. After successful execution of scFun, it decrements the \emph{indegree} count of conflicting node of source node using \emph{remExNode()} at \Lineref{val7}. While decrementing the \emph{indegree} of conflicting node validator thread checks if it found new source node then it store that node in its thread local log $thLog$ to execute next \sctrn at \Lineref{val10} efficiently. 

\vspace{.5mm}
\begin{algorithm}
	\scriptsize
	\caption{\Mthr validator{(\sctl, BG)}: $v$ threads concurrently and deterministically executes the \sctrn{s} using BG.} 
	\label{alg:val}
	
	
	\begin{algorithmic}[1]
		\makeatletter\setcounter{ALG@line}{120}\makeatother
		\Procedure{\emph{\Mthr validator{(\sctl, BG)}}}{}  \label{lin:val1}
		\State /*Execute until all the \sctrn{s} successfully completed*/\label{lin:val2}
		\While{(\nc{} $<$ size\_of(\sctl))}  \label{lin:val3} /*Initially, \emph{sctCount}=0 to maintain count.*/
		
		\State \emph{vrtNode} = globalSearch{(BG)}; /*Identify the source node (\emph{indegree} 0) in the BG*/\label{lin:val5}
		
		
		\State remExecNode(\emph{vrtNode}); /*Decrement the \emph{indegree} of conflicting nodes*/\label{lin:val7}
		\While{($thLog$ $\neq$ $nil$)} /*Identify source node in thread local log ($thLog$)*/ \label{lin:val9}
		\State \emph{vrtNode} = localSearch{($thLog$)};\label{lin:val10}
		\State remExecNode(\emph{vrtNode});\label{lin:val11}
		\EndWhile \label{lin:val12}
		\EndWhile\label{lin:val13}
		\EndProcedure \label{lin:val14}
	\end{algorithmic}
\end{algorithm}

\cmnt{
\begin{figure*}
	\centerline{\scalebox{0.4}{\input{figs/cvalidator.pdf_t}}}
	 \caption{Working of Concurrent Validator}
	\label{fig:cval}
\end{figure*}

\figref{cval}.(a) represents the BG in which conflicting edge is from transaction $T_1$ to $T_2$ because they accessed the same account as B. So \mthr validator needs to respect the conflicting order of the transaction to avoid the \emph{false block rejection (or FBR)} error and execute them serially as depicted in \figref{cval}.(b). However, transaction $T_3$ do not have any conflict with other transactions so it can be executed concurrently. After successful execution of \sctrn{s} validator threads compute the final state ($FS_v$) for all the shared data items (or accounts) as shown in \figref{cval}.(c). 
}

Finally, validator thread compares the $FS_m$ given by the \mthr miner and $FS_v$ computed by itself corresponding to each shared data item. If final state matches and proposed block reaches the global consensus, then it is added into the blockchain and respective miner awarded with the incentive.

\subsection{Detection of Malicious Miner by Smart \Mthr Validator (SMV)}
\label{apn:ap-mm}
In this subsection, we propose a technique to detect malicious miner using \emph{Smart \Mthr Validator}.

As we have seen the functionality of \mthr validator in \apnref{ap-cvalidator}, it executes the SCTs concurrently rely on the BG provided by the \mthr miner. Suppose the miner that produces a block is malicious and does not add some edges to the \blg. This can result in the \bc systems entering inconsistent states due to \emph{double spend}. We motivate this with an example. Consider three bank accounts $A, B, C$ maintained on the \bc with the current balance being \$100 in each of them. Now consider two \sctrn{s} $T_i, T_j$ which are conflicting where (a) $T_i$ transfers \$50 from $A$ to $B$; (b) $T_j$ transfers \$60 from $A$ to $C$. Considering the initial balance of \$100 in $A$ account, both transactions cannot be executed. 

If a malicious miner, say $mm$ does not add an edge between these two transactions in the \blg then both these \sctrn{s} can execute concurrently by validators. Then such execution could result in the final state with the balances in the accounts $A, B, C$ as 40, 150, 160 respectively or 50, 150, 160.  As we can see, neither of these final states can be obtained from any serial execution and are not correct states. Suppose the miner $mm$ stores 40, 150, 160 for $A, B, C$  in the final state, and a validator $v$ on concurrent execution arrives at the same state. Then, $v$ will accept this block, which results in its state becoming inconsistent. If the majority of validators similarly accept this block, then the state of the \bc essentially has become inconsistent. We denote this problem as \emph{edge missing \blg} or \emph{\emb}. 





\noindent
\textbf{Counter Based Solution to Catch the Malicious Miner:} So, to avoid this issue, we propose a 
a \emph{Smart \Mthr Validator (\scv)}, which uses the concept of $counters$ and identifies the malicious behavior of miner and rejects the proposed malicious block. Our algorithm is inspired by \bto in databases \cite[Chap. 4]{WeiVoss:TIS:2002:Morg}. \scv keeps track of each global data item that can be accessed across multiple transactions by different threads. Specifically, \scv maintains two global counters for each key of hash-table (shared data item) $k$ - (a) $\guc{k}$ (b) $\glc{k}$. These respectively keep track of number of \textbf{updates} and \textbf{lookups} that are concurrently performed by different threads on $k$. Both these global counters are initialized to 0.


When a \scv thread $Th_x$ is executing an \sctrn $T_i$ then \scv similarly maintains two local variables corresponding to each global data item $k$ which is accessible only by $Th_x$ - (c) $\luc{k}{i}$ (d) $\llc{k}{i}$. These respectively keep track of number of updates and lookups performed by $Th_x$ on $k$ while executing $T_i$. These counters are initialized to 0 before the start of $T_i$.  

Having described the counters, we will explain the high level design of \scv approach is shown in \algoref{cminermc}. To access the BG, validator uses the block graph library methods \searchg{()}, \searchl{()}, and remExecNode() as explained in \apnref{ap-dsBG}. Internally they use the \exec{()} method to execute the smart contract function (scFun). First, it identifies the scFun steps and executes them one after another at \Lineref{mc1}. 



If current step (curStep) is lookup (at \Lineref{mc5}) on shared data-item key $k$ then it checks the $\guc{k}$ counter value. If $\guc{k}$ counter value is not equal to $k.lLC$ at \Lineref{mc6}, that means another concurrent conflicting thread is also working on the same key $k$, i.e., conflict edge among them are missing in BG given by the miner. Then \scv reports the miner is malicious.

If $\guc{k}$ counter value is zero means equal to $k.lLC$ then it atomically increments the $\glc{k}$ counter of key $k$ in shared memory at \Lineref{mc7}, so, any other concurrent conflicting thread checks the value as non zero it will detect the malicious miner. It also increments the local $k.lLC$ value by one at \Lineref{mc711}. Finally, validator thread lookups the key $k$ from the shared memory and return the value as $v$ at \Lineref{mc8}.

If curStep is insert on key $k$ with value as $v$ (at \Lineref{mc13}) then before inserting the key $k$ with value $v$ in the shared memory it checks both global counter values ($\glc{k}$ == $k.lLC$) $\&\&$ ($\guc{k}$ == $k.lUC$) at \Lineref{mc14}. If anyone of the counter value is not equal to corresponding to the local variable value, that means another concurrent conflicting thread is also working on the same key $k$, i.e., conflict edge among them is missing in BG given by the miner. Then \scv reports the miner is malicious.

If both global counter value is equal to corresponding local variables value, then it atomically increments the $\guc{k}$ counter of key $k$ in shared memory at \Lineref{mc15}, so, any other concurrent conflicting thread checks the value as non zero it will detect the malicious miner. It also increments the local $k.lUC$ value by one at \Lineref{mc722}. Finally, validator thread inserts the key $k$ with value $v$ in the shared memory at \Lineref{mc16}. Same things works if curStep is deleted on key $k$ at \Lineref{mc21}.  

After successful execution of each scFun, thread atomically decrements $\guc{k}, \glc{k}$ by the value of $\luc{k}{i}, \llc{k}{i}$ respectively at \Lineref{mc34}. Then thread will reset $\luc{k}{i}, \llc{k}{i}$ to 0. Thus with the help of $counter$, validator threads are able to detect the malicious miner, and straightforward reject that block. 




\begin{algorithm} [!htb]
	\scriptsize
	
	\caption{\exec($scFun$): Execute the smart contract function (scFun) with atomic global lookup/update counter. Initially, \emph{lookup counter ($\glc{k}$)} and \emph{update counter ($\guc{k}$}) value is 0 corresponding to each shared data-items key $k$. Each transaction maintains local \emph{$\llc{k}{i}$} and local \emph{$\luc{k}{i}$} as 0 in transaction local log, \emph{txLog} corresponding to each key.}
	\label{alg:cminermc} 
	\begin{algorithmic}[1]
		\makeatletter\setcounter{ALG@line}{131}\makeatother
		\While{(scFun.steps.hasNext())} /*Assume that scFun is a list of steps*/ \label{lin:mc1}
		\State curStep = scFun.steps.next(); /*Get the next step to execute*/ \label{lin:mc2}
		\Switch{(curStep)} \label{lin:mc3}
		\EndSwitch \label{lin:mc4}
		\Case{lookup($k$):} \label{lin:mc5}
		\If{($\guc{k}$ == \emph{$\llc{k}{i}$})} /*Check for update counter ($uc$) value*/ \label{lin:mc6}
		\State Atomically increment the lookup counter, $\glc{k}$; \label{lin:mc7}
		\State Increment \emph{$\llc{k}{i}$} by 1. /*Maintain \emph{$\llc{k}{i}$} in transaction local log \emph{txLog}*/ \label{lin:mc711}
		\State Lookup $k$ from a shared memory; \label{lin:mc8}
		\Else \label{lin:mc9}
		\State return $\langle$\emph{Miner is malicious}$\rangle$; \label{lin:mc10}
		\EndIf \label{lin:mc11}
		\EndCase \label{lin:mc12}
		\Case{insert($k, v$):}  \label{lin:mc13}
		\If{(($\glc{k}$ $==$ \emph{$\llc{k}{i}$}) \&\& ($\guc{k}$ == \emph{$\luc{k}{i}$}))} /*Check lookup/update counter value*/ \label{lin:mc14}
		\State Atomically increment the update counter, $\guc{k}$; \label{lin:mc15}
		\State Increment \emph{$\luc{k}{i}$} by 1. /*Maintain \emph{$\luc{k}{i}$} in transaction local log \emph{txLog}*/ \label{lin:mc722}
		\State Insert $k$ in shared memory with value $v$; \label{lin:mc16}
		\Else \label{lin:mc17}
		\State return $\langle$\emph{Miner is malicious}$\rangle$; \label{lin:mc18}
		\EndIf \label{lin:mc19}
		\EndCase \label{lin:mc20}
		\Case{delete($k$):}  \label{lin:mc21}
		\If{(($\glc{k}$ $==$ \emph{$\llc{k}{i}$}) \&\& ($\guc{k}$ == \emph{$\luc{k}{i}$}))} /*Check lookup/update counter value*/ \label{lin:mc22}
		\State Atomically increment the update counter, $\guc{k}$; \label{lin:mc23}
		\State Increment \emph{$\luc{k}{i}$} by 1. /*Maintain \emph{$\luc{k}{i}$} in transaction local log \emph{txLog}*/
		\State Delete $k$ in shared memory. \label{lin:mc24}
		\Else \label{lin:mc25}
		\State return $\langle$\emph{Miner is malicious}$\rangle$; \label{lin:mc26}
		\EndIf \label{lin:mc27}
		\EndCase \label{lin:mc28}
		\Case{default:} \label{lin:mc29}
		\State curStep is not lookup, insert and delete; \label{lin:mc30}
		\State execute curStep; \label{lin:mc31}
		\EndCase \label{lin:mc32}
		\EndWhile	 \label{lin:mc33}
		\State Atomically decrement the \emph{$\glc{k}$} and \emph{$\guc{k}$} corresponding to each shared data-item key $k$.	\label{lin:mc34}
		
	\end{algorithmic}
\end{algorithm}

\section{Correctness of BG, \Mthr Miner and Validator}
\label{apn:ap-correctness}
This section describes the proof of theorems stated for the correctness of BG, \mthr miner, and validator in \secref{pm}. In order to define the correctness of BG, we identify the linearization points (LPs) of each method as follows:

\begin{enumerate}
    \item \addv{(\vgn)}: (\vp.\vn.CAS(\vc, \vgn)) in \Lineref{addv5} is the LP point of \addv{()} method if \vnode{} is not exist in the BG. If \vnode{} exist in the BG then (\vc.$ts_i$ $\neq$ \vgn.$ts_i$) in \Lineref{addv3} is the LP point.
    
    \item \adde{\emph{($conflictNode_1$, $conflictNode_2$)}}:         (\ep.\en.CAS(\ec, \egn)) in \Lineref{adde5} is the LP point of \adde{()} method if \enode{} is not exist in the BG. If \enode{} is exist in the BG then         (\ec.$ts_i$ $\neq$ $conflictNode_2$.$ts_i$) in \Lineref{adde3} is the LP point.
    
    \item \searchl{(thLog)}: ($llVertex$.\inc.CAS(0, -1)) in \Lineref{sl3} is the LP point of \searchl{()} method.
    
    \item \searchg{(BG)}: (\vnode.\inc.CAS(0, -1)) in \Lineref{sg5} is the LP point of \searchg{()} method.
    
    \item \emph{remExecNode}(removeNode): \Lineref{ren3} is the LP point of \emph{remExecNode()} method.
\end{enumerate}

\begin{theorem}
    The BG captures all the dependencies between the conflicting nodes.
\end{theorem} 

\begin{proof}
    \secref{bg} represents the construction of BG by \mthr miner using \svotm and \mvotm protocol. BG considers each committed \sctrn{} as a vertex and edges (or dependencies) depends on the used STM protocols such as \svotm and \mvotm. So, the underlying STM protocol ensures that all the dependencies have been covered correctly among the conflicting nodes of BG. Hence, all the dependencies between the conflicting nodes are captured in the BG.
\end{proof}

\begin{theorem}
    A history $H_m$ generated by the \mthr miner with \svotm satisfies co-opacity.
\end{theorem}

\begin{proof}
 Multiple miner threads execute \sctrn{s} concurrently using \svotm and generate a concurrent history $H_m$. The underlying \svotm protocol ensures the correctness of concurrent execution of $H_m$. \svotm proves that any history generated by it satisfies co-opacity \cite{Peri+:OSTM:Netys:2018}. So, implicitly \svotm proves that the history $H_m$ generated by \mthr miner using \svotm satisfies co-opacity. 
\end{proof}

\begin{theorem}
\label{thm:hmo}
    A history $H_m$ generated by a \mthr miner with \mvotm satisfies opacity.
\end{theorem}

\begin{proof}
In order to achieve the greater concurrency further, a \mthr miner uses the \mvotm protocol, which maintains multiple versions corresponding to each shared data-item. \mvotm ensures the correct concurrent execution of the history $H_m$ with the equivalent serial history $S_m$. Any history generated by \mvotm satisfies opacity \cite{Juyal+:MVOSTM:SSS:2018}. So, implicitly \mvotm proves that the history $H_m$ generated by \mthr miner using \mvotm satisfies opacity.
\end{proof}

\begin{theorem}
\label{thm:hmve}
    	A history $H_m$ generated by the \mthr miner with \svotm and history $H_v$ generated by a \mthr validator are view equivalent.
\end{theorem}
\begin{proof}
\Mthr miner execute the \sctrn{s} concurrently using \svotm protocol to propose a block and generates $H_m$ along with BG. After that \mthr miner broadcasts $H_m$ and BG to \mthr validators to verify the proposed block. \Mthr validator applies the topological sort on BG and obtained an equivalent serial schedule $H_v$. Since BG constructed from $H_m$ while considering all the \oconf{s} and $H_v$ obtained from the topological sort on BG. So, $H_v$ will be equivalent to $H_m$. Similarly, $H_v$ will also follow the \emph{return value from} relation to $H_m$. Hence, $H_v$ is legal. Since $H_v$ and $H_m, $ are equivalent to each other, and $H_v$ is legal. So, $H_m$ and $H_v$ are view equivalent. 
\end{proof}

\begin{theorem}
A history $H_m$ generated by the \mthr miner with MVOSTM and history $H_v$ generated by a \mthr validator are multi-version view equivalent.
\end{theorem}
\begin{proof}
Following the proof of \thmref{hmve}, \mthr miner executes $H_m$ using \mvotm and constructs the BG. \mvotm maintains multiple-version corresponding to each shared data-item while executing the \sctrn{s} by \mthr miner. Later, \mthr validator  obtained $H_v$ by applying topological sort on BG given by miner. Since, $H_v$ obtained from topological sort on BG so, $H_v$ will be equivalent to $H_m$. Similarly, BG maintains the \emph{return value from} relations of $H_m$. So, from \mvotm protocol if $T_j$ returns a value of the method for shared data-item $k$ say $rv_j(k)$  from $T_i$ in $H_m$ then $T_i$ committed before $rv_j(k)$ in $H_v$. Therefore, $H_v$ is valid. Since $H_v$ and $H_m$ are equivalent to each other and $H_v$ is valid. So, $H_m$ and $H_v$ are multi-version view equivalent. 
\end{proof}

\begin{theorem}
    \label{thm:apHmm}
    Smart \Mthr Validator rejects  malicious blocks with \blg that allow concurrent execution of dependent \sctrn{s}.
\end{theorem}
\begin{proof}

With the help of global $counter$ \emph{Smart \Mthr Validator (SMV)} identifies the concurrent execution of dependent \sctrn{s} at \Lineref{mc6}, \Lineref{mc14}, and \Lineref{mc22} of \algoref{cminermc} and reject the malicious block. Detail description of \scv is available in \apnref{ap-mm}. We have tested our proposed counter-based approach with the existence of malicious block shown in \apnref{ap-rresult}. \scv straightforward reject the malicious block and notify to the other nodes part of the network. Hence, \scv rejects malicious blocks with \blg that allow concurrent execution of dependent \sctrn{s}.
\end{proof}

\cmnt{
\subsection{The Linearization Points of Lock-free Graph Library Methods}

Here, we list the linearization points (LPs) of each method as follows:

\begin{enumerate}
\item \addv{(\vgn)}: (\vp.\vn.CAS(\vc, \vgn)) in \Lineref{addv5} is the LP point of \addv{()} method if \vnode{} is not exist in the BG. If \vnode{} is exist in the BG then (\vc.$ts_i$ $\neq$ \vgn.$ts_i$) in \Lineref{addv3} is the LP point.

\item \adde{\emph{(fromNode, toNode)}}:         (\ep.\en.CAS(\ec, \egn)) in \Lineref{adde7} is the LP point of \adde{()} method if \enode{} is not exist in the BG. If \enode{} is exist in the BG then         (\ec.$ts_i$ $\neq$ toNode.$ts_i$) in \Lineref{adde5} is the LP point.
 
\item \searchl{(cacheVer, $AU_{id}$)}: (cacheVer.\inc.CAS(0, -1)) in \Lineref{sl1} is the LP point of \searchl{()} method.

\item \searchg{(BG, $AU_{id}$)}: (\vnode.\inc.CAS(0, -1)) in \Lineref{sg1} is the LP point of \searchg{()} method.

\item \emph{decInCount}(remNode): \Lineref{ren1} is the LP point of \emph{decInCount()} method.
\end{enumerate}
}
\section{Detailed Experimental Evaluation}
\label{apn:ap-rresult}

This section presents a detailed description of the benchmark contracts that we have considered in this paper. It also includes the additional experiments which show the performance benefits of proposed \mthr miner and validator over state-of-the-art miners and validators on various workloads. Along with this, we proposed \emph{smart \mthr validator} to identify malicious miners.

\noindent
\textbf{Smart Contracts:} Clients (possibly different) send transactions to the miners in the form of complex code known as smart contracts. It provides several complex services such as managing the system state, ensuring rules, or credentials checking of the parties involved. \cite{Dickerson+:ACSC:PODC:2017}. For better understanding, we have described \emph{Coin, Ballot, Simple Auction} Smart Contracts from Solidity documentation \cite{Solidity}. We consider one more smart contract as \emph{Mix Contract}, which is the combination of the three contracts as mentioned above in equal proportion and seems more realistic.

\noindent(1) \textit{Coin Contract:} It is a sub-currency contract which implements simplistic form of a cryptocurrency and is used to transfer coins from one account to another account using  \emph{send()}, or used to check the account balance using \emph{get\_balance()} function. Accounts (unique addresses in Ethereum) are shared objects. A conflict will occur when two or more transaction consists of at least one common account, and one of them is updating the account balance.

\algoref{cc1} shows the functionality of the coin contract, where \textit{mint()}, \textit{send()}, and \textit{get\_balance()} are the functions of the contract. These functions can be called by the miners or through other contracts. It initialized by the contract creator (or contract deployer) to a special public state variable \textit{minter} (\Lineref{c2}). Accounts are identified by Solidity mapping data structure essentially a $\langle \emph{key-value} \rangle$ pair (\Lineref{c3}), where a key is the unique Ethereum address and value is unsigned integer depicts the coins (or balance) in respective account. Initially, the contract deployer (aka \textit{minter}) creates new coins and allocate it to each receiver (\Lineref{c9}).  

Further, \emph{send()} function is used to transfer the coin from the sender account to the receiver account. The function ensures that the sender has sufficient balance in his account (\Lineref{c11}). If sufficient balance found in the sender's account, the coin transferred from the sender account to the receiver account. By calling \textit{get\_balance()}, anyone can query the specific account balance (\Lineref{c15}).


\noindent(2) \textit{Ballot Contract:} This contract is used to organize electronic voting where voters and proposals are the shared objects and stored at unique Ethereum addresses. At the beginning of voting, the chairman of the ballot gives rights to voters to vote. Later, voters either delegate their vote to other voter using \emph{delegate()} or directly vote to specific proposal using \emph{vote()}. Voters are allowed to delegate or vote only once per ballot. A conflict will occur when two or more voters vote for the same proposal, or they delegate their votes to the same voter simultaneously.

\vspace{.2cm}
\setlength{\intextsep}{0pt}
\begin{algorithm}[!htb]
	\scriptsize
	\caption{Coin(): A sub-currency contract used to depict the simplest form of a cryptocurrency.}
	\label{alg:cc1}
	\begin{algorithmic}[1]
		\makeatletter\setcounter{ALG@line}{163}\makeatother
		\Procedure{$Coin()$}{} \label{lin:c1}
		\State address public minter;/*Minter is a unique public address*/\label{lin:c2}
		\State /*Map $\langle \emph{key-value} \rangle$ pair of hash-table as $\langle \emph{address-balance} \rangle$*/
		\State mapping(address $=>$ uint) balances. \label{lin:c3} 
		
		\State \textbf{Constructor()} public\label{lin:c4}
		\State{\hspace{.3cm} minter = msg.sender.} /*Set the sender as minter*/\label{lin:c5}
		
		\Function {}{}mint(address receiver, uint amount )\label{lin:c6}
		\If{(msg.sender == minter)} \label{lin:c7}
		\State /*Initially, add the balance into receiver account*/
		\State balances[receiver] += amount. \label{lin:c9}    
		\EndIf
		
		\EndFunction
		
		\Function {}{}send(address receiver, uint amount)\label{lin:c10}
		\State /*Sender don't have sufficient balance*/
		\If{(balances[msg.sender] $<$ amount)} \label{lin:c11}
		return $\langle fail \rangle$;\label{lin:c12}
		\EndIf
		\State balances[msg.sender] -= amount;\label{lin:c13}
		\State balances[receiver] += amount;\label{lin:c14}
		\EndFunction
		
		\Function{}{}get\_balance(address account)\label{lin:c15}
		\State return $\langle$balance$\rangle$;
		\EndFunction
		\EndProcedure
		
	\end{algorithmic}
\end{algorithm}

\noindent(3) \textit{Simple Auction Contract:} In this contract, an auction is conducted where a bidder places their bids. Here bidders, maxBid, maxBidder, and auction end time are the shared object which can be accessed by multiple threads. The auction will end when the bidding period (or end time) of the auction is over. The auction end time is initialized at the beginning by the auction master. A \emph{bid()} function is used to bid the amount by a bidder for the auction. In the end, the bidder with the highest bid amount will be the winner, and all other bidders amount is then returned to them using \emph{withdraw()}. A conflict will occur when more than two bidders try to bid using \emph{bidPlusOne()} at the same time.

\noindent(4) \textit{Mix Contract: }In this contract, aforementioned smart contracts are executed simultaneously. This contract is designed to show real-time scenarios in which a block consists of \sctrn{s} from different contracts. For the experiment, we combined \sctrn{s} for three contracts in a block.

\noindent\textbf{Performance Analysis:} To analyze the proposed approach further, we show the performance analysis on remaining benchmark contracts and workloads. Additionally, we consider one more workload W3, in which the number of shared objects (data-items) vary from 100 to 600, while threads, \sctrn{s}, and hash-table size are fixed to 50, 100, and 30, respectively. For W3, with the increase in the number of shared objects, contention will decrease. 

In \figref{all-speedup-coin} to \figref{all-speedup-auction} numbering (a), (b), (c) show the \mthr miner speedup over serial miner on W1, W2, and W3 for coin, ballot, and auction contract respectively. Further, (d), (e), (f) shows the smart \mthr validator speedup over serial validator on W1, W2, and W3 for the coin, ballot, and auction contract respectively. It shows that the speedup decreases for \mthr miner on W1; however, it increases on W2 on all benchmark contracts. The observation for W1 and W2 are the same as explained in \secref{result}. Finally, \figref{w3-mix} (a) and \figref{w3-mix} (b) shows the speedup for \mthr miner and SMV for mix contract on workload W3.

For W3, as shown in \figref{all-speedup-coin} (c), \figref{all-speedup-ballot} (c), \figref{all-speedup-auction} (c), and \figref{w3-mix} (a) the speedup increase for \mthr miner with increase in shared objects (contention decreases). However in mix contract (as shown in \figref{w3-mix} (a)) small decrements for \bto and \mvto miner can be observed. Also, static bin miner is performing worse than serial due to the overhead of static conflict prediction before executing \sctrn{s} speculatively. Similarly \figref{all-speedup-coin} (f), \figref{all-speedup-ballot} (f), \figref{all-speedup-auction} (f), and \figref{w3-mix} (b) shows the speedup achieve by \scv over serial validator on W3. The speedup of bin-based \scv is less than STM validator. Thus for the better visualization, we have shown speedup for STM validator on $y1$ axis whereas for bin-based \scv on $y2$ axis. As we can observe, \mvotm \scv outperforms all other validators; however, performance decreases with increasing shared objects. 

\begin{figure}
	\includegraphics[width=\textwidth]{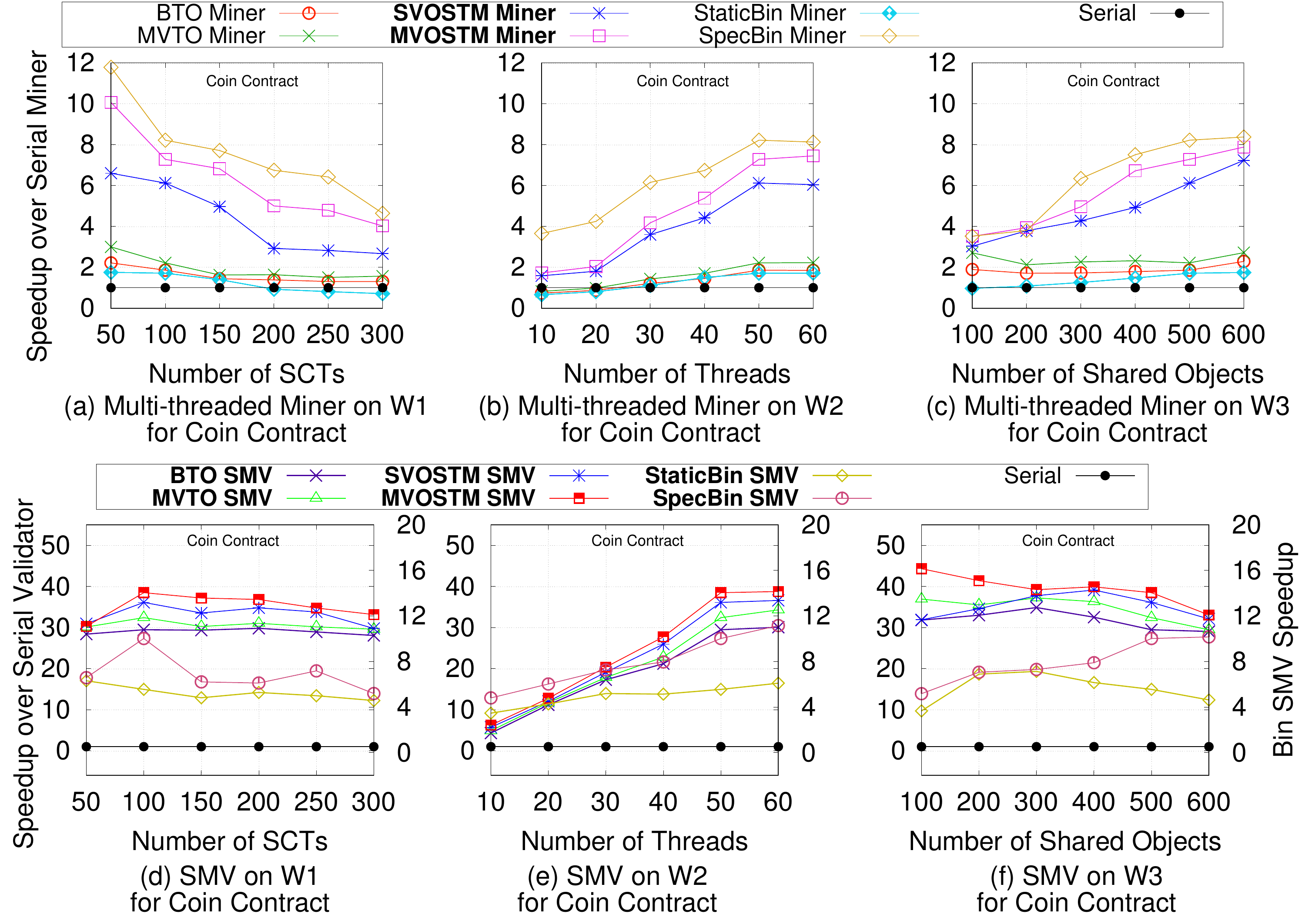}\vspace{-.3cm}
	 \caption{\Mthr Miner and Smart \Mthr Validator Speedup Over Serial Miner and Validator Across all Workloads for Coin Contract}
	\label{fig:all-speedup-coin}

\vspace{.25cm}

	\includegraphics[width=\textwidth]{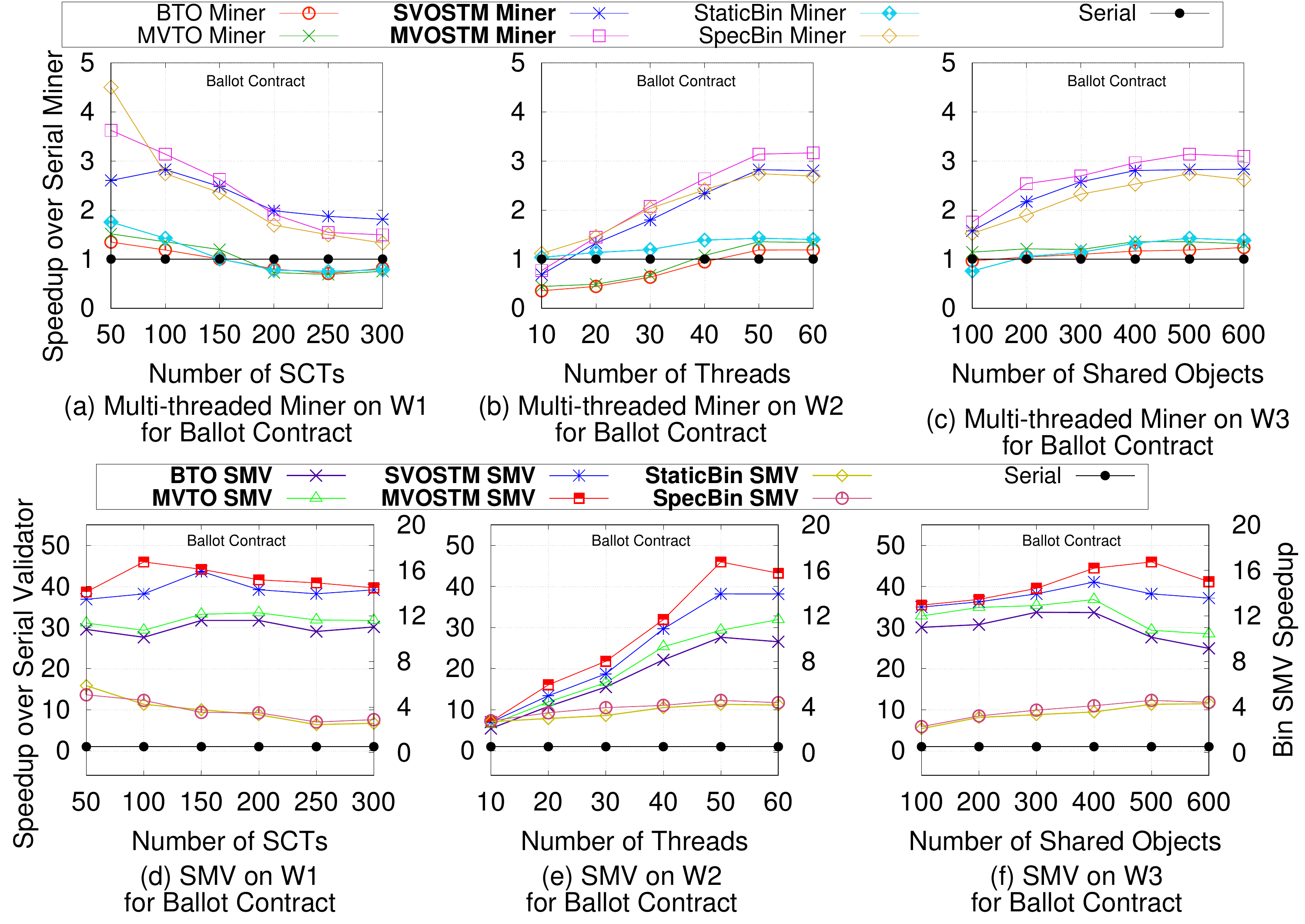}\vspace{-.3cm}
	 \caption{\Mthr Miner and Smart \Mthr Validator Speedup Over Serial Miner and Validator Across all Workloads for Ballot Contract}
	\label{fig:all-speedup-ballot}
\end{figure}

\begin{figure}
	\includegraphics[width=\textwidth]{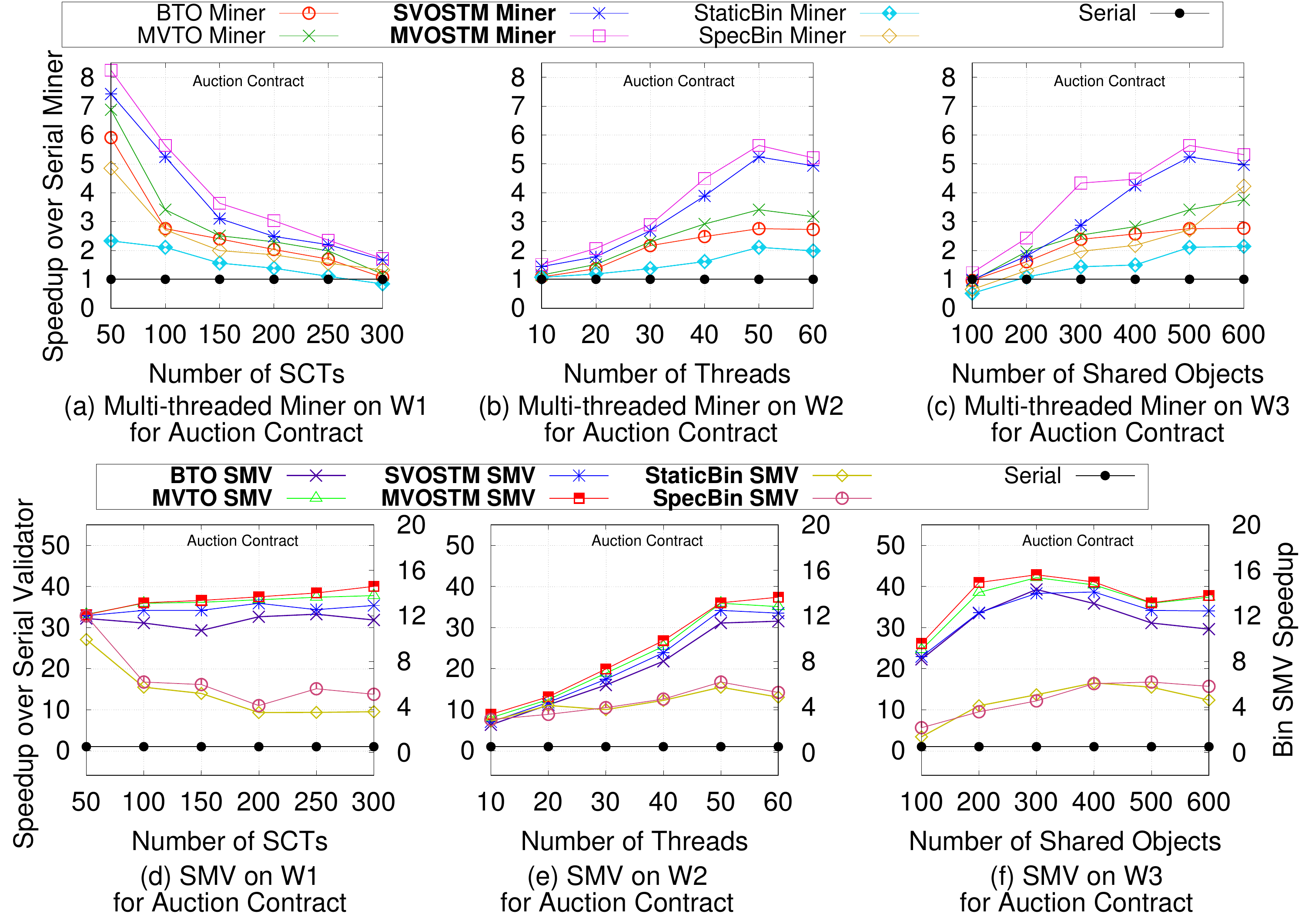}\vspace{-.3cm}
	 \caption{\Mthr Miner and Smart \Mthr Validator Speedup Over Serial Miner and Validator Across all Workloads for Auction Contract}
	\label{fig:all-speedup-auction}

\vspace{.45cm}

	\centering
	\includegraphics[width=.66\textwidth]{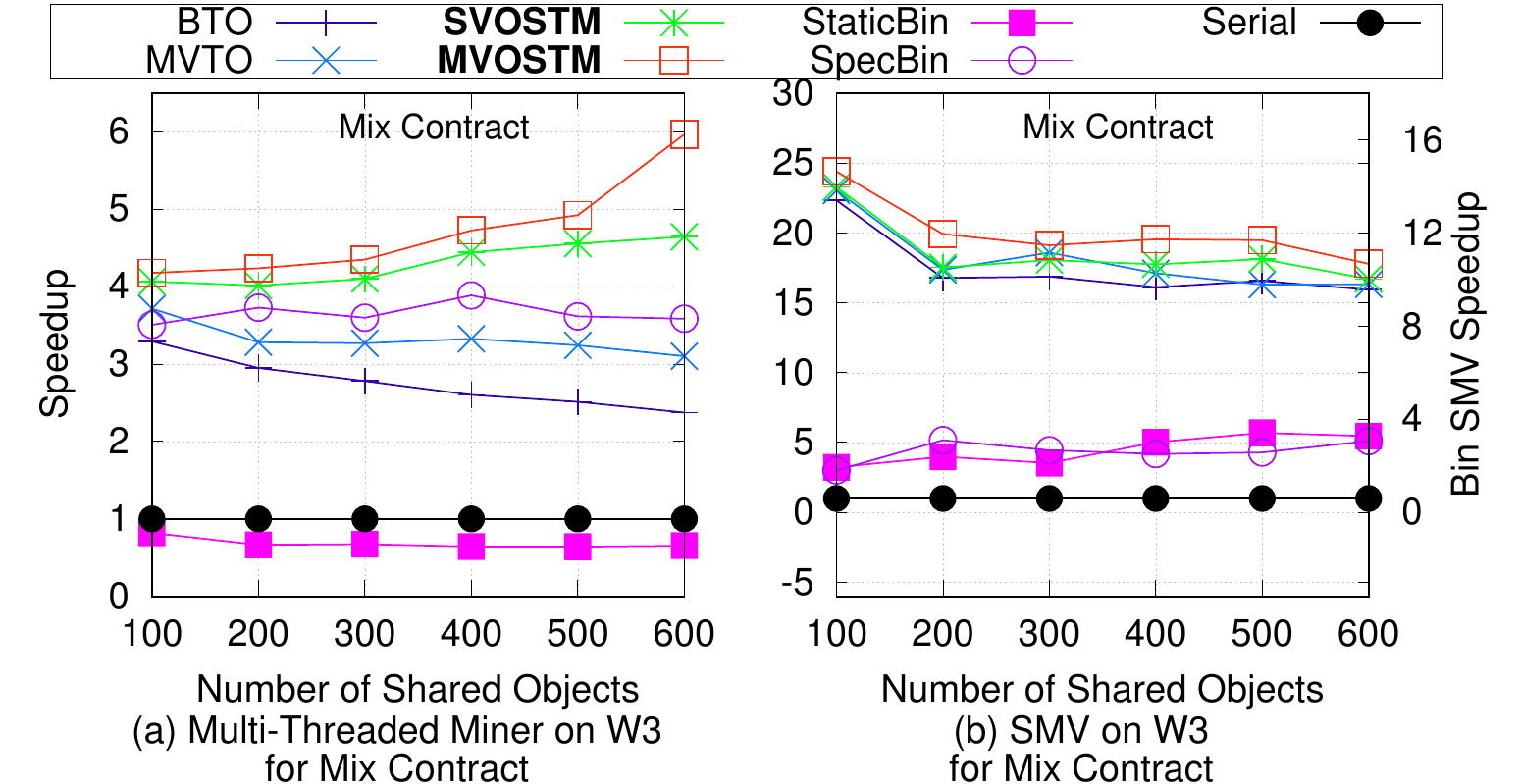}\vspace{-.3cm}
	 \caption{\Mthr Miner and Smart \Mthr Validator Speedup Over Serial Miner and Validator on W3 for Mix Contract}
	\label{fig:w3-mix}
\end{figure}

\noindent
\textbf{Dependencies in the BG: }\figref{all-depd-BG-coin} to \figref{all-depd-BG-auction} shows the average dependencies in the Block Graph (BG) generated by STM based \mthr miners for the coin, ballot, and auction contract on all workloads. While \figref{w3-BG-MM-mix}.(a) shows the average number of dependencies in the BG for mix contract on W3. There is no BG in bin-based static bin and speculative bin miner; instead, there is a sequential and concurrent bin. So from block size consideration, bin-based approach is efficient, though, from validator speedup consideration, STM based approach is better as shown in all smart \mthr validators figures.

As shown in figures for W1 with the increase in SCTs, the number of dependencies increases in BG. However, for W2, there is no much variation since we fixed the number of \sctrn{s} to 100 in W2. Moreover, the analysis of W3 is quite impressive. Here for the ballot and mix contracts with the increase in shared data items, the dependencies in BG increases for \bto and \mvto. However, it decreases for \svotm and \mvotm as shown in \figref{all-depd-BG-ballot} (c). Also, for coin contract, dependencies in BG decreases with an increase in shared data-item. In the Auction contract, it depends on the highest bid; if the highest bid is bided in the beginning, then there will be least dependencies in BG. 

The average speedup (averaged across the workloads) achieved by the \mthr miners and smart \mthr validators on workload W1, W2, and W3 on all benchmarks are shown in Table \ref{tbl:avgMiner} and Table \ref{tbl:avgValidator} respectively. Note that the average speedup result shown in \secref{result} for mix contract is averaged on workload W1 and W2.

\begin{table}[H]
\centering
\caption{Overall Average Speedup on all Workloads by \Mthr Miner over Serial}
\label{tbl:avgMiner}
\resizebox{.6\textwidth}{!}{%
\begin{tabular}{|c|c|c|c|c|c|c|}
\hline
& \multicolumn{6}{c|}{\textbf{\Mthr Miner}}                                            \\ \cline{2-7} 
\multirow{-2}{*}{\textbf{Contract}} &
  \textbf{\begin{tabular}[c]{@{}c@{}}\bto\\ Miner\end{tabular}} &
  \textbf{\begin{tabular}[c]{@{}c@{}}\mvto\\ Miner\end{tabular}} &
  {\color[HTML]{00009B} \textbf{\begin{tabular}[c]{@{}c@{}}\svotm\\ Miner\end{tabular}}} &
  {\color[HTML]{00009B} \textbf{\begin{tabular}[c]{@{}c@{}}\mvotm\\ Miner\end{tabular}}} &
  \textbf{\begin{tabular}[c]{@{}c@{}}SpecBin\\ Miner\end{tabular}} &
  \textbf{\begin{tabular}[c]{@{}c@{}}StaticBin\\ Miner\end{tabular}} \\ \hline
\textbf{Coin}               & 1.596         & 1.959         & 4.391         & 5.572         & 1.279         & 6.689         \\ \hline
\textbf{Ballot}             & 0.960         & 1.065         & 2.229         & 2.431         & 1.175         & 2.233         \\ \hline
\textbf{Auction}            & 2.305         & 2.675         & 3.456         & 3.881         & 1.524         & 2.232         \\ \hline
\textbf{Mix}                & 1.596         & 2.118         & 3.425         & 3.898         & 1.102         & 3.080         \\ \hline
\textbf{Total Avg. Speedup} & \textit{1.61} & \textit{1.95} & \textit{3.38} & \textit{3.95} & \textit{1.27} & \textit{3.56} \\ \hline
\end{tabular}%
}
\end{table}

\begin{table}[H]
\centering
\caption{Overall Average Speedup on all Workloads by \scv over Serial}
\label{tbl:avgValidator}
\resizebox{.6\textwidth}{!}{%
\begin{tabular}{|c|c|c|c|c|c|c|}
\hline
 & \multicolumn{6}{c|}{\textbf{Smart \Mthr Validator (\scv)}}                                \\ \cline{2-7} 
\multirow{-2}{*}{\textbf{Contract}} &
  {\color[HTML]{00009B} \textbf{\begin{tabular}[c]{@{}c@{}}\bto\\ \scv\end{tabular}}} &
  {\color[HTML]{00009B} \textbf{\begin{tabular}[c]{@{}c@{}}\mvto\\ \scv\end{tabular}}} &
  {\color[HTML]{00009B} \textbf{\begin{tabular}[c]{@{}c@{}}\svotm\\ \scv\end{tabular}}} &
  {\color[HTML]{00009B} \textbf{\begin{tabular}[c]{@{}c@{}}\mvotm\\ \scv\end{tabular}}} &
  {\color[HTML]{00009B} \textbf{\begin{tabular}[c]{@{}c@{}}SpecBin\\ \scv\end{tabular}}} &
  {\color[HTML]{00009B} \textbf{\begin{tabular}[c]{@{}c@{}}StaticBin\\ \scv\end{tabular}}} \\ \hline
\textbf{Coin}               & 26.576         & 28.635         & 30.344         & 32.864         & 5.296         & 7.565         \\ \hline
\textbf{Ballot}             & 26.037         & 28.333         & 33.695         & 36.698         & 3.570         & 3.780         \\ \hline
\textbf{Auction}            & 27.772         & 31.781         & 29.803         & 32.709         & 4.694         & 5.214         \\ \hline
\textbf{Mix}                & 36.279         & 39.304         & 42.139         & 45.332         & 4.279         & 4.463         \\ \hline
\textbf{Total Avg. Speedup} & \textit{29.17} & \textit{32.01} & \textit{34.00} & \textit{36.90} & \textit{4.46} & \textit{5.26} \\ \hline
\end{tabular}%
}
\end{table}

\vspace{.25cm}
\begin{figure}[H]
	\includegraphics[width=\textwidth]{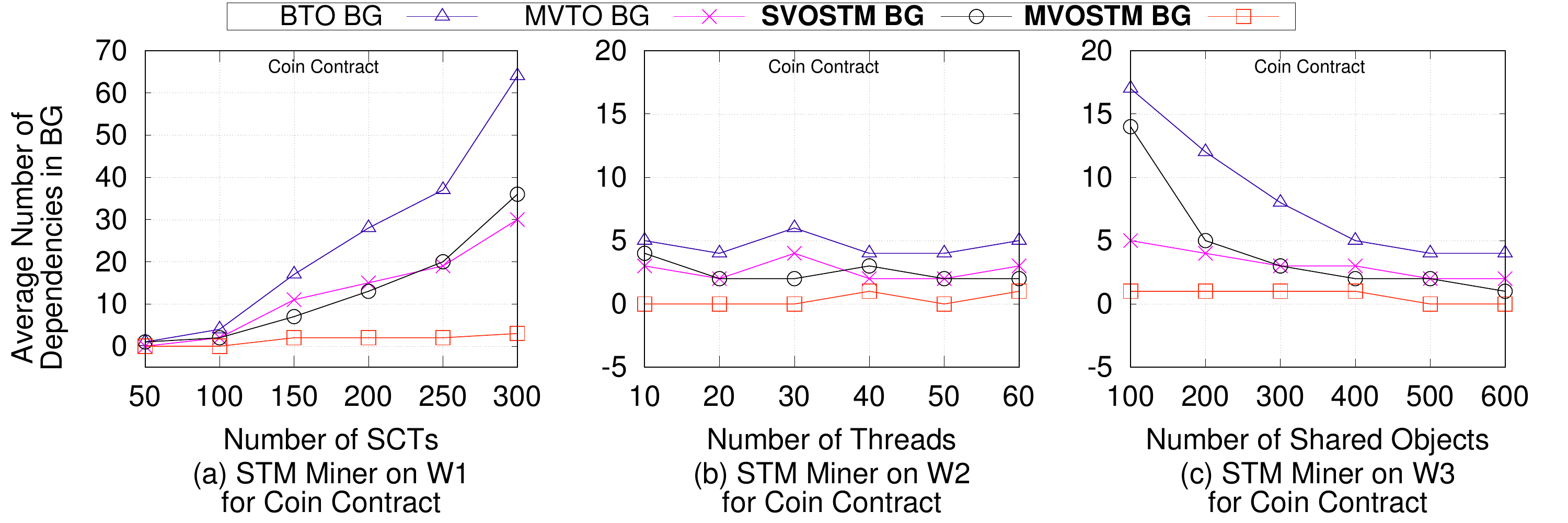}\vspace{-.35cm}
	 \caption{Average Number of Dependencies in Block Graph for Coin Contract}
	\label{fig:all-depd-BG-coin}
\end{figure}

\begin{figure}[H]

	\includegraphics[width=\textwidth]{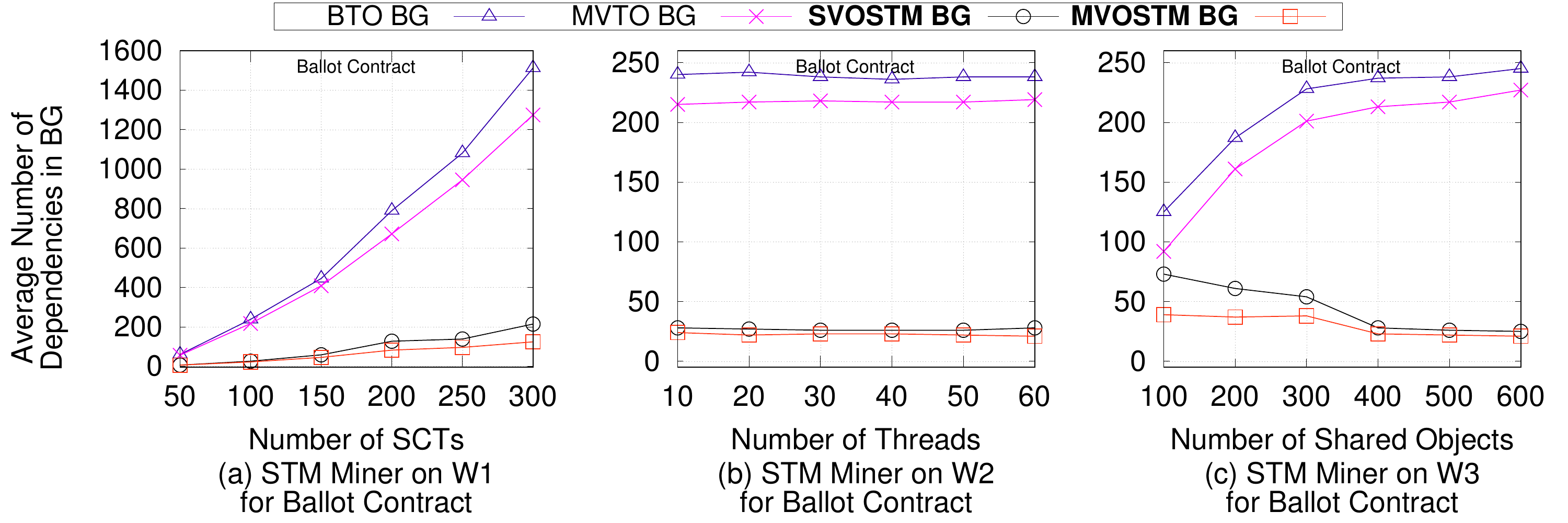}\vspace{-.35cm}
	 \caption{Average Number of Dependencies in Block Graph for Ballot Contract}
	\label{fig:all-depd-BG-ballot}

\vspace{.4cm}

	\includegraphics[width=\textwidth]{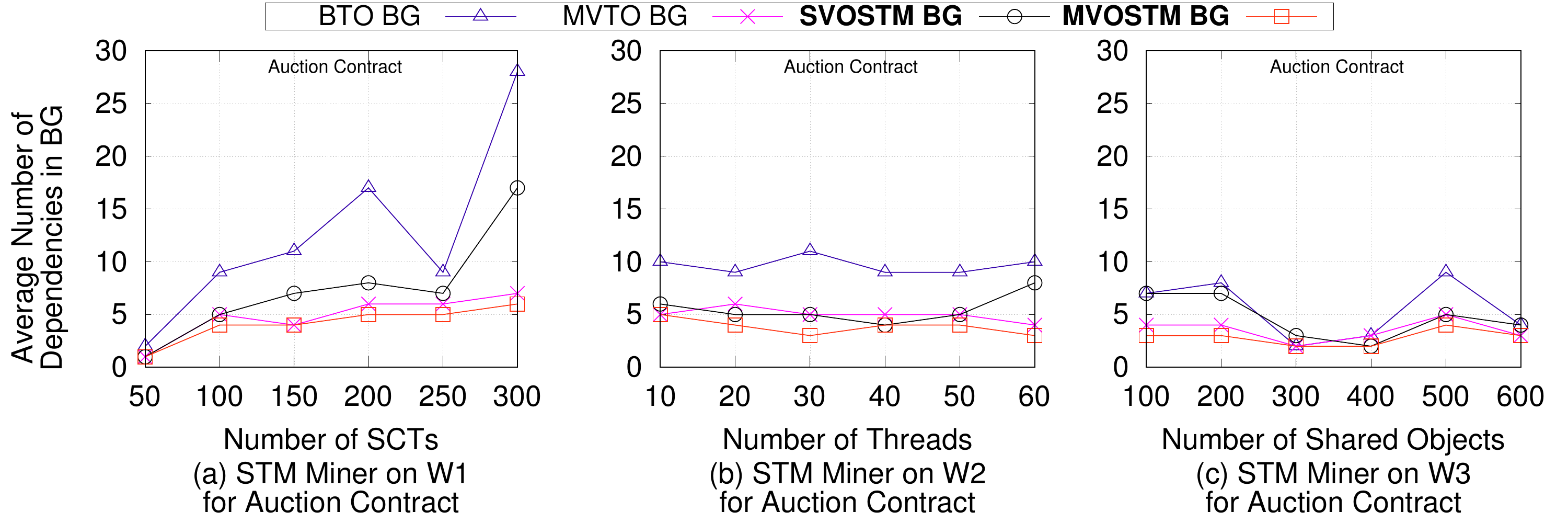}\vspace{-.35cm}
	 \caption{Average Number of Dependencies in Block Graph for Auction Contract}
	\label{fig:all-depd-BG-auction}

\vspace{.4cm}

    \centering
	\includegraphics[width=.92\textwidth]{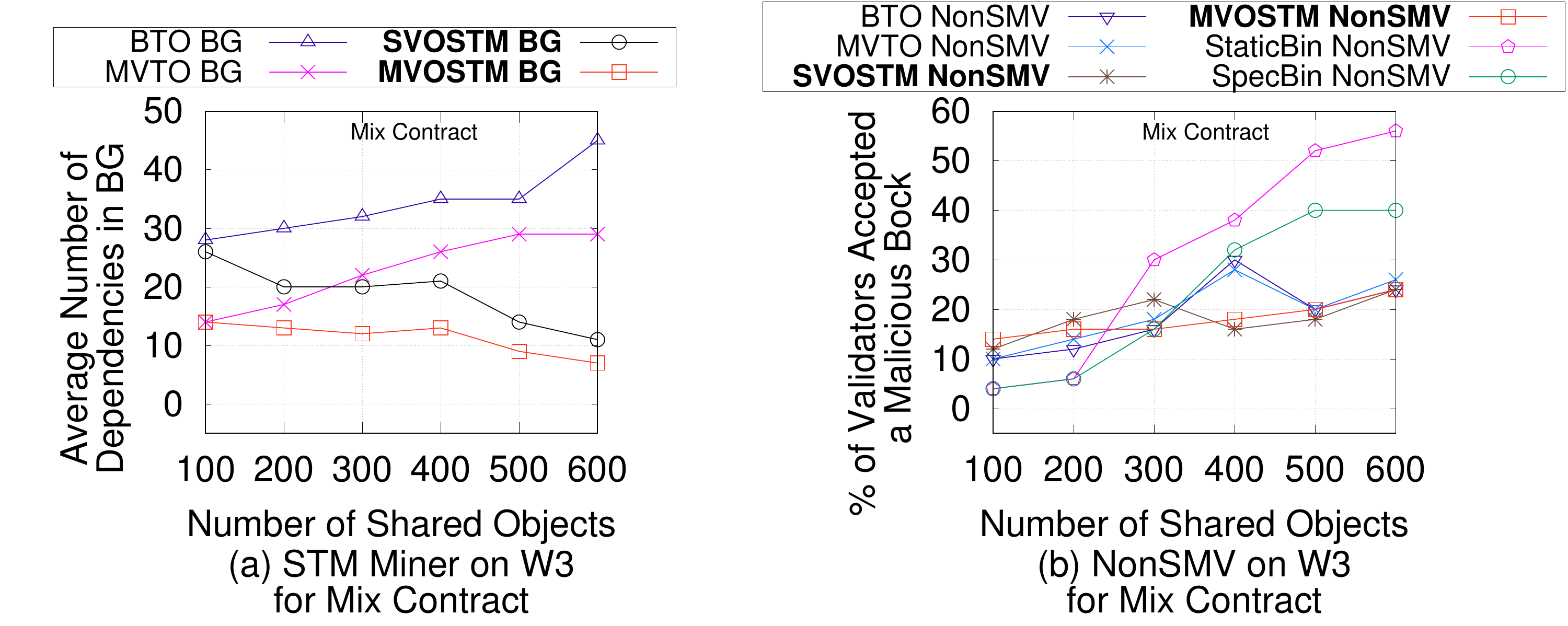}\vspace{-.35cm}
	 \caption{Average Number of Dependencies in Block Graph and Percentage of Average \Mthr Validator (NonSMV) Accepted a Malicious Block on W3 for Mix Contract}
	\label{fig:w3-BG-MM-mix}    

\end{figure}

\noindent
\textbf{Experiments on Malicious Miner:} A \mthr validator deterministically executes the \sctrn{s} rely on the BG provided by the miner in the block. However, what if a miner is malicious and embeds an incorrect BG? To answer this question, we have done experiments for the malicious miner. As explained earlier, we proposed a \emph{Smart \Mthr Validator (\scv)} to prevent such malicious activity due to concurrent execution of \sctrn{s} of the block. This experiment shows how many validators (Non\scv) accept malicious block proposed by a malicious miner.

To obtain the malicious miner activity, we generate two \sctrn{s} of double-spending (explained in \apnref{ap-mm}) in coin contract and ballot contract (double voting: a voter votes two different proposals with one voting right). After that, malicious miner added such \sctrn{s} into the block, manipulate the final state accordingly, but did not add the respective dependencies in BG, i.e., for these two \sctrn{s}, indegrees will be 0. Finally, malicious miner broadcast the malicious block in the network. Then \mthr Non\scv validators re-executes the \sctrn{s} concurrently using BG provided by the malicious miner. However, the validators may execute double-spending \sctrn{s} concurrently and compute the same final state as provided by the malicious miner. So, some of the validators accept the malicious block. If they reach a consensus, then they will add this malicious block into the blockchain. It may cause a severe issue in the blockchain.

\figref{w1w2mixmminer}, \figref{w3-BG-MM-mix}.(b), \figref{w1-w2-w3-mminer-coin}, and \figref{w1-w2-w3-mminer-ballot} demonstrates the average percentage of validators accepting a malicious block on different workloads and benchmark contracts. Here, we consider 50 validators and run the experiments for the Coin, Ballot, and Mix contract. So, we can conclude that if the malicious miner is present in the network, then some validators may agree on the malicious block, and harm the blockchain. Therefore, we should ensure the rejection of such a malicious block in the blockchain. To address this issue, we proposed \emph{Smart \Mthr Validator} (describe in \apnref{ap-mm}), which always detects the malicious block at the time of concurrent execution of malicious \sctrn{s} (double-spending and double voting) with the help of $counter$ and straightforward reject that block. Analysis of SMVs is presented for mix contract in \secref{result} and other contracts at the beginning of this section.

So, the next obvious question is, how much extra time does \emph{SMVs} is taking to serve the purpose of identifying the malicious miner over \emph{NonSMVs}? 
We observe that the counter-based \mthr validator (i.e., Smart \Mthr Validator (SMV)) approach is giving a bit less speedup than without counter-based \mthr validator (i.e., NonSMV) however this decrement in speedup is very low on considered workloads. Instead of small decrement in speedup, it is evident to use counter-based \mthr validators (\emph{SMVs}) to preserve the correctness of the blockchain.

\begin{figure}[H]
	\includegraphics[width=\textwidth, height= 4.5cm]{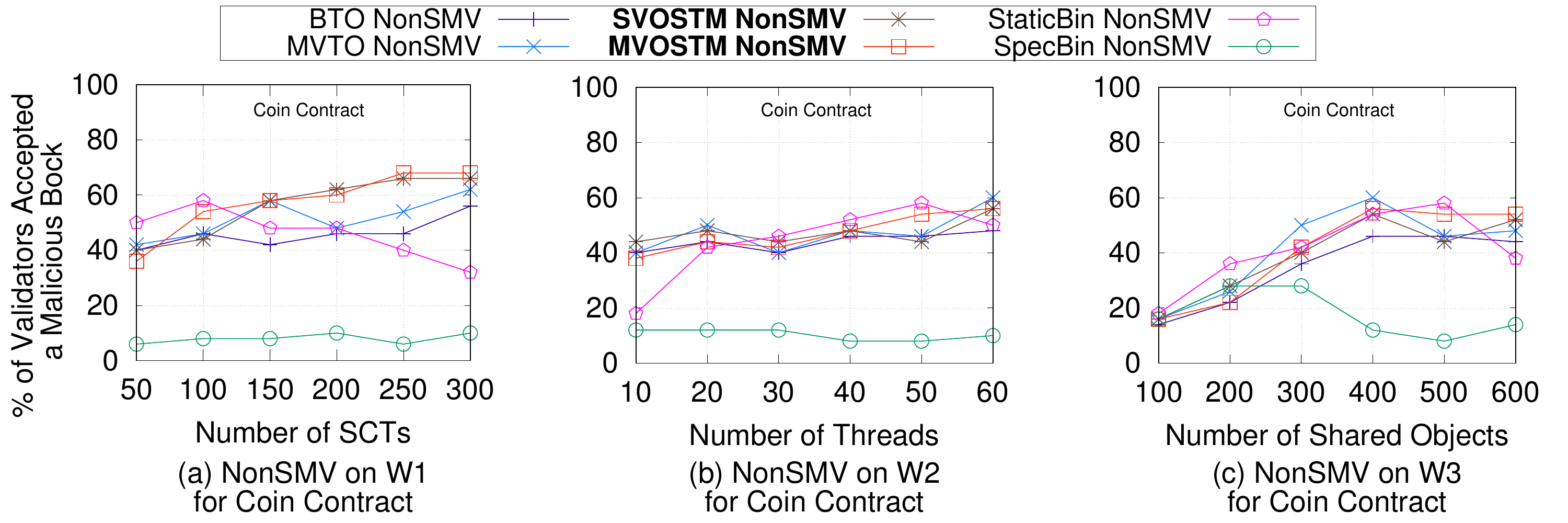}\vspace{-.3cm}
	 \caption{Percentage of Average \Mthr Validator (NonSMV) Accepted a Malicious Block for Coin Contract}
	\label{fig:w1-w2-w3-mminer-coin}
\vspace{.45cm}
	\includegraphics[width=\textwidth, height= 4.5cm]{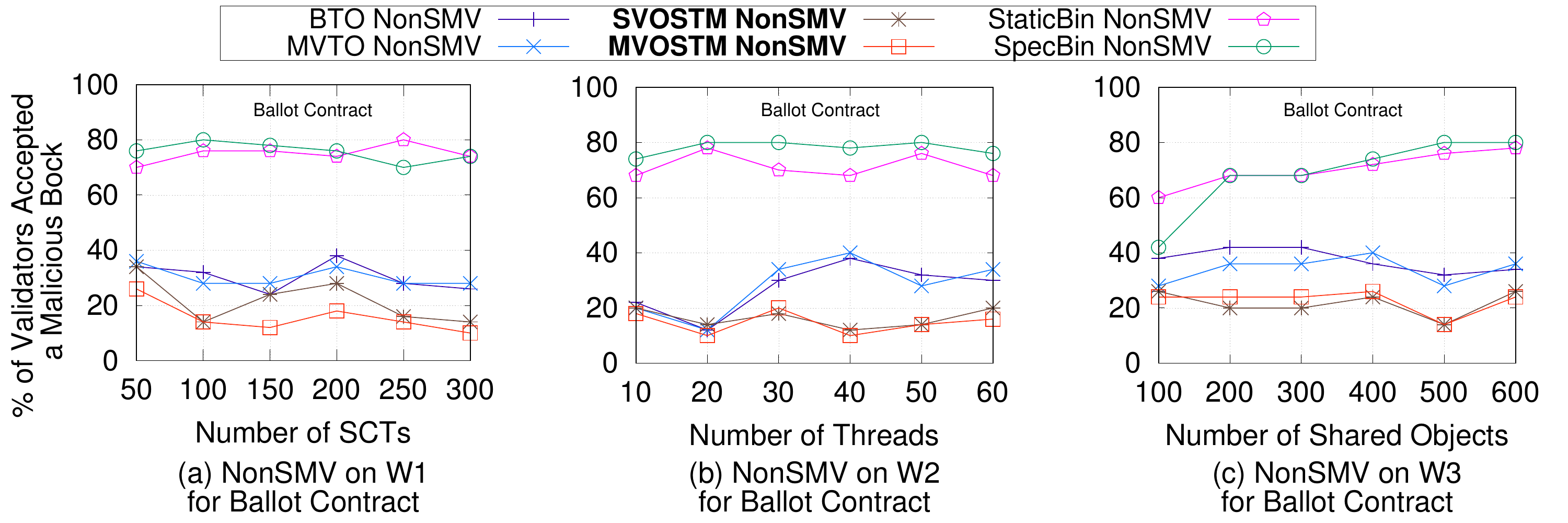}\vspace{-.3cm}
	 \caption{Percentage of Average \Mthr Validator (NonSMV) Accepted a Malicious Block for Ballot Contract}
	\label{fig:w1-w2-w3-mminer-ballot}
\end{figure}

\noindent
\textbf{Experiments on Block Graph Size:} We also measure the additional space required to append the BG into the block. In Ethereum and Bitcoin average block size is $\approx 20.98$ \cite{EthereumAvgBlockSize} and $\approx 1123.34$ KB \cite{BitcoinAvgBlockSize} respectively for the interval of 1$^{st}$ Jan. 2019 to 8$^{th}$ March 2020, which is keep on increasing every year. The average number of transactions in a block of Ethereum is $\approx 100$ \cite{EthereumAvgBlockSize}. So, on an average, each transaction requires $0.2$ KB ($\approx 200$ bytes) in Ethereum. Based on this simple calculation, we have computed block size with an increase in \sctrn{s} per block for workload W1. To compute the block size \emph{Equation \ref{eq:blocksize}} is used.
\begin{equation}
    B = 200 * N_{\sctrn{s}}
\label{eq:blocksize}
\end{equation}
Where, $B$ is block size in bytes, $N_{\sctrn{s}}$ number of smart contract transactions (\sctrn{s}) in block, and $200$ is the average size of an \sctrn{} in bytes.

We use \emph{adjacency list} to maintain the Block Graph $BG(V, E)$ inspired from \cite{Anjana:OptSC:PDP:2019,Chatterjee+:NbGraph:ICDCN:2019}. Here $V$ is the set of vertices (\vrtnode{s}) is stored as a vertex list, $\vrtlist$. Similarly E is the set of Edges (\egnode{s}) is stored as edge list ($\eglist$ or conflict list) as shown in the \figref{graph} (a) of \apnref{ap-dsBG}. Both $\vrtlist$ and $\eglist$ store between the two sentinel nodes \emph{Head}($-\infty$) and \emph{Tail}($+\infty$). Each \vrtnode{} maintains a tuple: \emph{$\langle$ts, scFun, indegree, egNext, vrtNext$\rangle$}. Here, \emph{ts (an integer)} is the unique timestamp $i$ of the transaction $T_i$ to which this node corresponds to. \emph{scFun (an integer)} is the ID of smart contract function executed by the transaction $T_i$ which is stored in \vrtnode. The number of incoming edges to the transaction $T_i$, i.e. the number of transactions on which $T_i$ depends, is captured by \emph{indegree (an integer)}. Field \emph{egNext (an address)} and \emph{vrtNext (an address)} points the next \egnode{} and \vrtnode{} in the $\eglist$ and $\vrtlist$ respectively. So a vertex node $V_s$ size is $28$ bytes in the experimental system, which is sum of the size of 3 integer variables and 2 pointers.

Each \egnode{} of $T_i$ similarly maintains a tuple: \emph{$\langle$ts, vrtRef, egNext$\rangle$}. Here, \emph{ts  (an integer)} stores the unique timestamp $j$ of $T_j$ which has an edge coming from $T_i$ in the graph. BG maintains the conflict edge from lower timestamp transaction to higher timestamp transaction. This ensures that the \bg is acyclic. The \egnode{s} in $\eglist$ are stored in increasing order of the \emph{ts}. Field \emph{vrtRef (an address)} is a \emph{vertex reference pointer} which points to its own \vrtnode{} present in the $\vrtlist$. This reference pointer helps to maintain the \emph{indegree} count of \vrtnode{} efficiently. The \emph{egNext (an address)} is a pointer to next edge node, so edge node $E_{s}$ requires a total of $20$ bytes in the experimental system.


The experimental results on the percentage of additional space required to store BG in the block are shown in \figref{all-inc-BS-coin} to \figref{all-inc-BS-mix} for all benchmark contracts and workloads. The size of BG ($\beta$) in bytes is computed using \emph{Equation \ref{eq:BGSize}}, while to compute the percentage of additional space ($\beta_{p}$) required to store BG in the block is calculated using \emph{Equation \ref{eq:perBG}}.
\begin{equation}
    \beta = (V_{s} * N_{\sctrn{s}}) + (E_{s} * M_{e})
    \label{eq:BGSize}
\end{equation}
Where, $\beta$ is size of Block Graph (BG) in bytes, $V_s$ is size of a vertex node of $BG$ in bytes, $N_{\sctrn{s}}$ are number of smart contract transactions (\sctrn{s}) in a block, $E_{s}$ is size of a edge node in bytes of $BG$, and $M_{e}$ is number of edges in $BG$.
\begin{equation}
    \beta_{p} = ({\beta*100})/{B}
\label{eq:perBG}
\end{equation}

As shown in \figref{all-inc-BS-coin} to \figref{all-inc-BS-mix}, it can be observed that with an increase in the number of dependencies, the space requirements also increase. The number of dependencies in the Ballot contract (\figref{all-depd-BG-ballot} (a)) for $W1$ is higher compared to other contracts, so the space requirement is also high. In all the figures, the space requirements of BG by \mvotm, \svotm is smaller than \mvto and \bto miner. The average space required for BG in \% concerning block size is $14.24\%$, $14.95\%$, $21.20\%$, and $22.70 \%$ by \mvotm, \svotm, \mvto, and \bto miner, respectively on W1 (As shown in Table \ref{tab:block-size}). Since the number of dependencies in BG developed by \mvotm is smaller than BG generated by other STM protocols, so it requires less space to store BG. In the future, we are planning to reduce space further to store the BG in the block by keeping the optimized or compressed BG.

\begin{table}
\centering
\caption{On W1 Average \% Increase in Block Size due to BG in Ethereum}
\label{tab:block-size}
\resizebox{.85\textwidth}{!}{%
\begin{tabular}{|c|c|c|c|c|}
\hline
& \multicolumn{4}{c|}{\textbf{Average \% of Increase in Block Size due to BG on W1}} \\ \cline{2-5} 
\multirow{-2}{*}{\textbf{Contract}} &
  \textbf{BTO Miner} &
  \textbf{MVTO Miner} &
  {\color[HTML]{00009B} \textbf{\svotm Miner}} &
  {\color[HTML]{00009B} \textbf{\mvotm Miner}} \\ \hline
\textbf{Coin}              & 14.225            & 13.702            & 13.712            & 13.220           \\ \hline
\textbf{Ballot}            & 44.542            & 40.633            & 17.377            & 16.073           \\ \hline
\textbf{Auction}           & \textit{13.811}   & \textit{13.427}   & \textit{13.534}   & \textit{13.392}  \\ \hline
\textbf{Mix}               & 18.238            & 17.043            & 15.180            & 14.264           \\ \hline
\textbf{Total Avg. Change} & \textit{22.70}    & \textit{21.20}    & \textit{14.95}    & \textit{14.24}   \\ \hline
\end{tabular}%
}
\end{table}

\begin{figure}[H]
	\includegraphics[width=\textwidth, height=4.5cm]{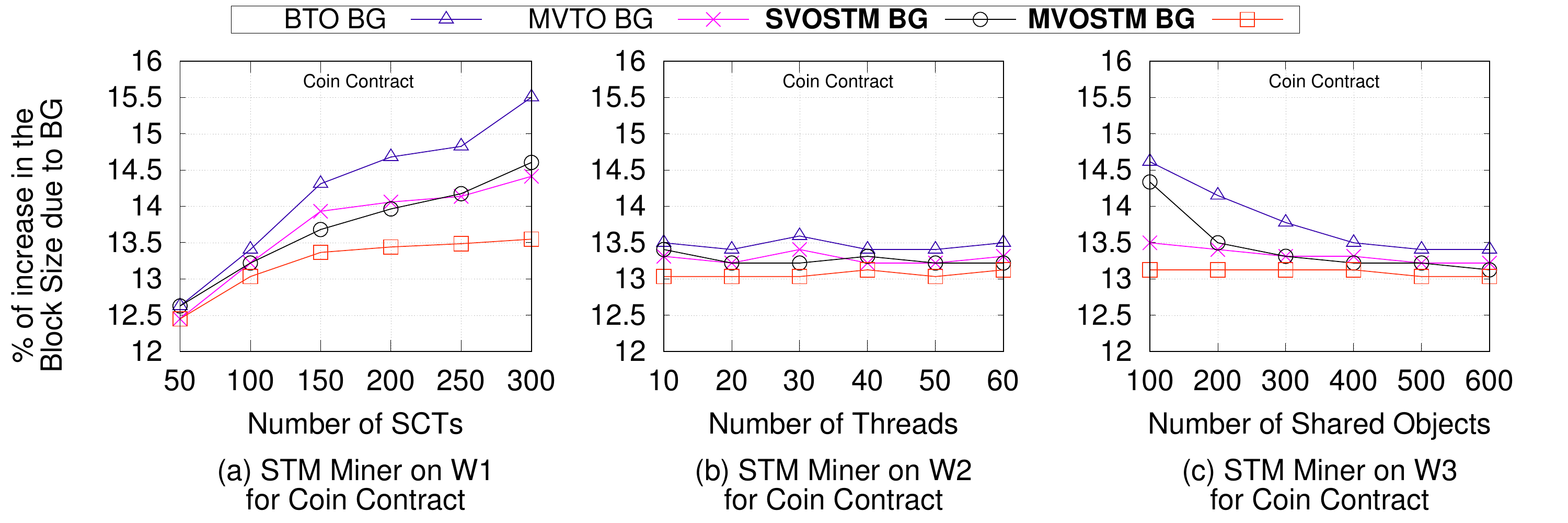}\vspace{-.3cm}
	 \caption{Percentage of Additional Space Required to Store Block Graph (BG) in Ethereum Block for Coin Contract}
	\label{fig:all-inc-BS-coin}

\vspace{.45cm}

	\includegraphics[width=\textwidth, height=4.5cm]{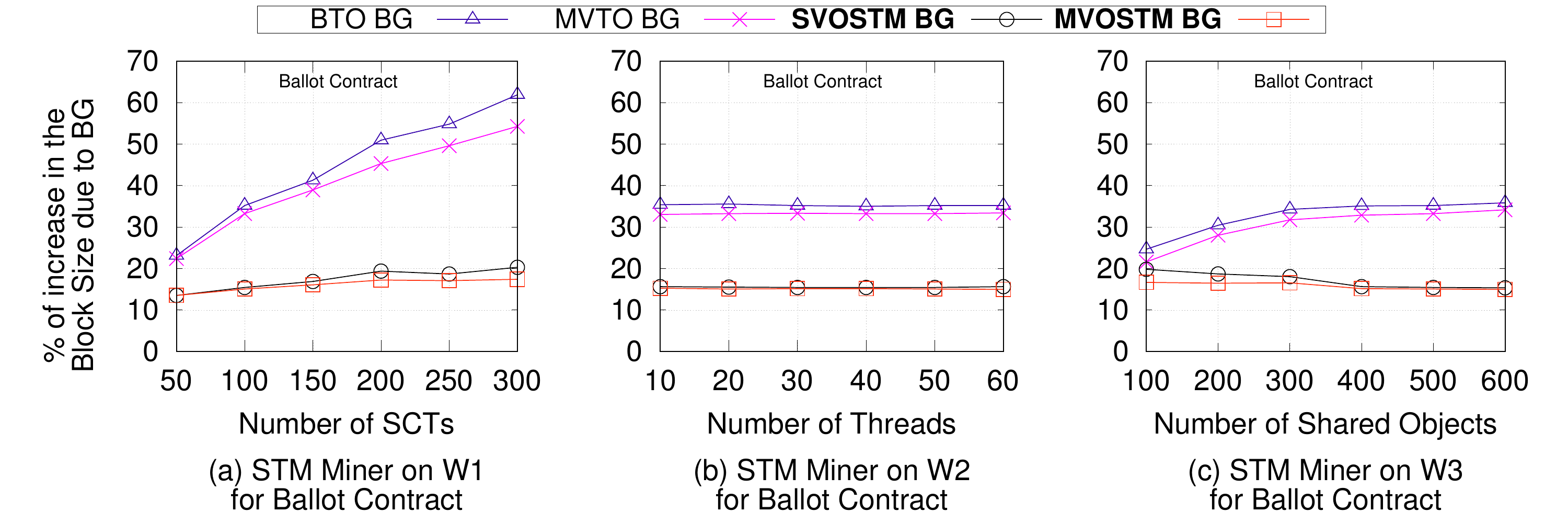}\vspace{-.3cm}
	 \caption{Percentage of Additional Space Required to Store Block Graph (BG) in Ethereum Block for Ballot Contract}
	\label{fig:all-inc-BS-ballot}
\end{figure}

\begin{figure}[H]
	\includegraphics[width=\textwidth, height=4.5cm]{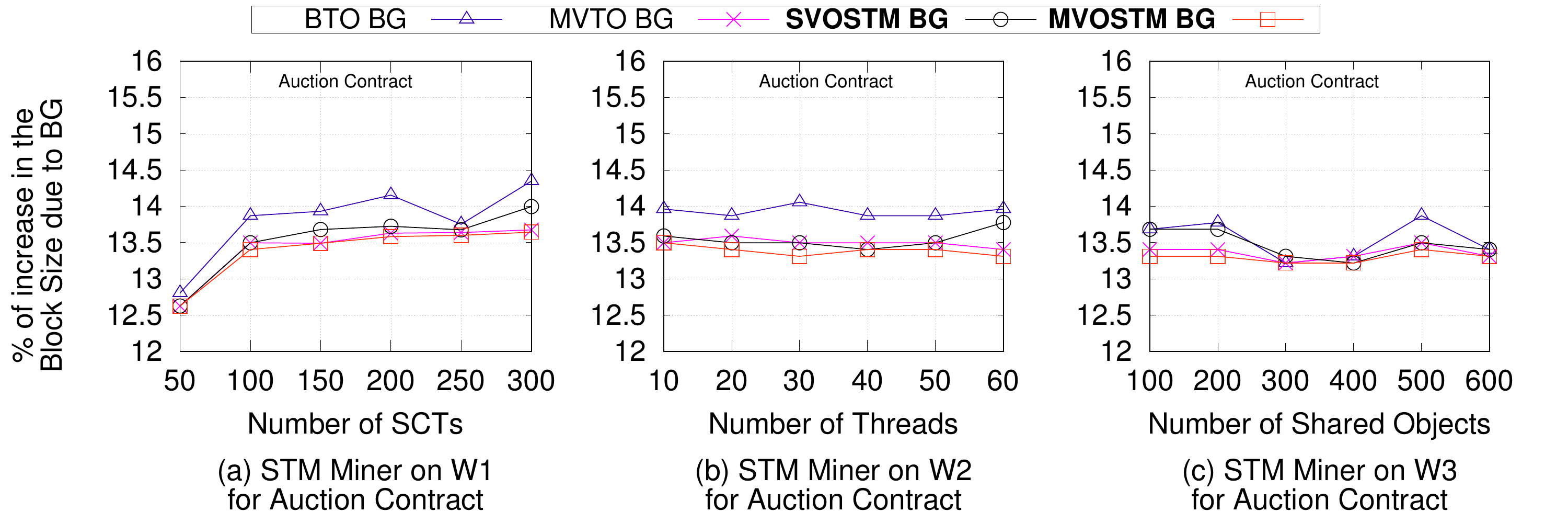}\vspace{-.3cm}
	 \caption{Percentage of Additional Space Required to Store Block Graph (BG) in Ethereum Block for Auction Contract}
	\label{fig:all-inc-BS-auction}

\vspace{.45cm}
	\includegraphics[width=\textwidth, height=4.5cm]{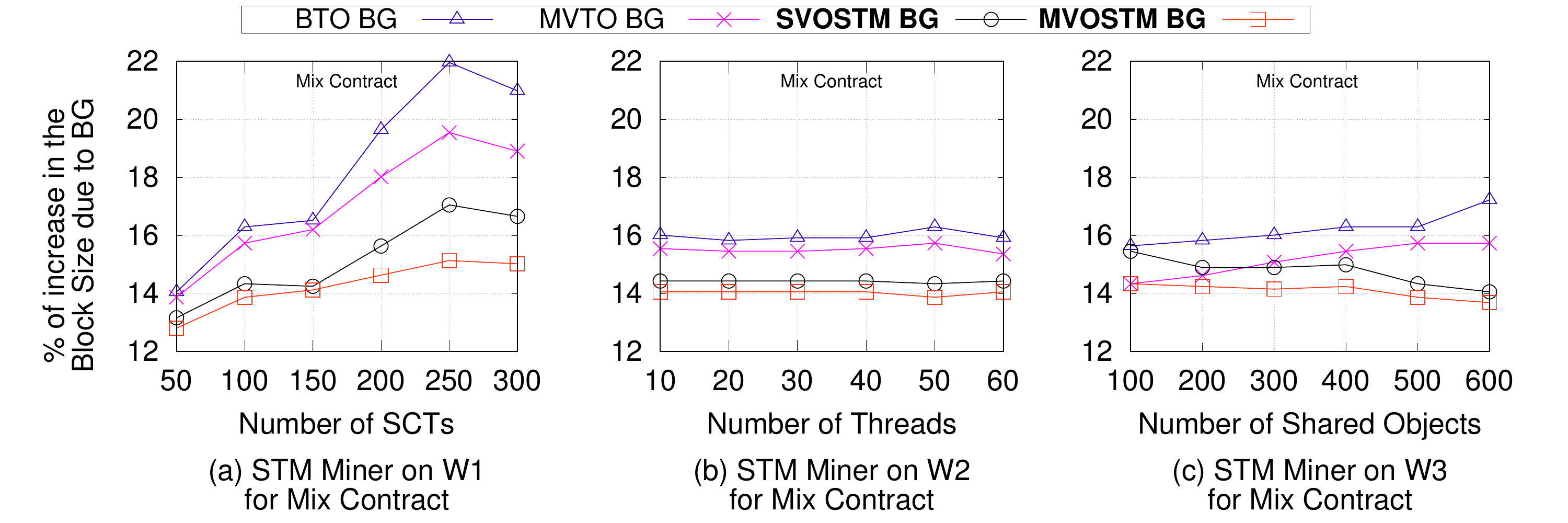}\vspace{-.3cm}
	 \caption{Percentage of Additional Space Required to Store Block Graph (BG) in Ethereum Block for Mix Contract}
	\label{fig:all-inc-BS-mix}
\end{figure}

\noindent\textbf{Performance Analysis of Decentralized NonSMV Validator:} \figref{validatorNonSMVcb} and \figref{validatorNonSMVam} show the performance of Decentralized NonSMV validator. Here we can observe that average speedup achieved by decentralized NonSMV validator is slightly better than \scv including bin-based and fork-join validators. However, NonSMV validators are prone to accepting a malicious block (the acceptance of malicious block is shown in \figref{w1w2mixmminer}, \figref{w1-w2-w3-mminer-coin}, and \figref{w1-w2-w3-mminer-ballot}). 
\begin{figure}[H]
    \centering
	\includegraphics[width=\textwidth]{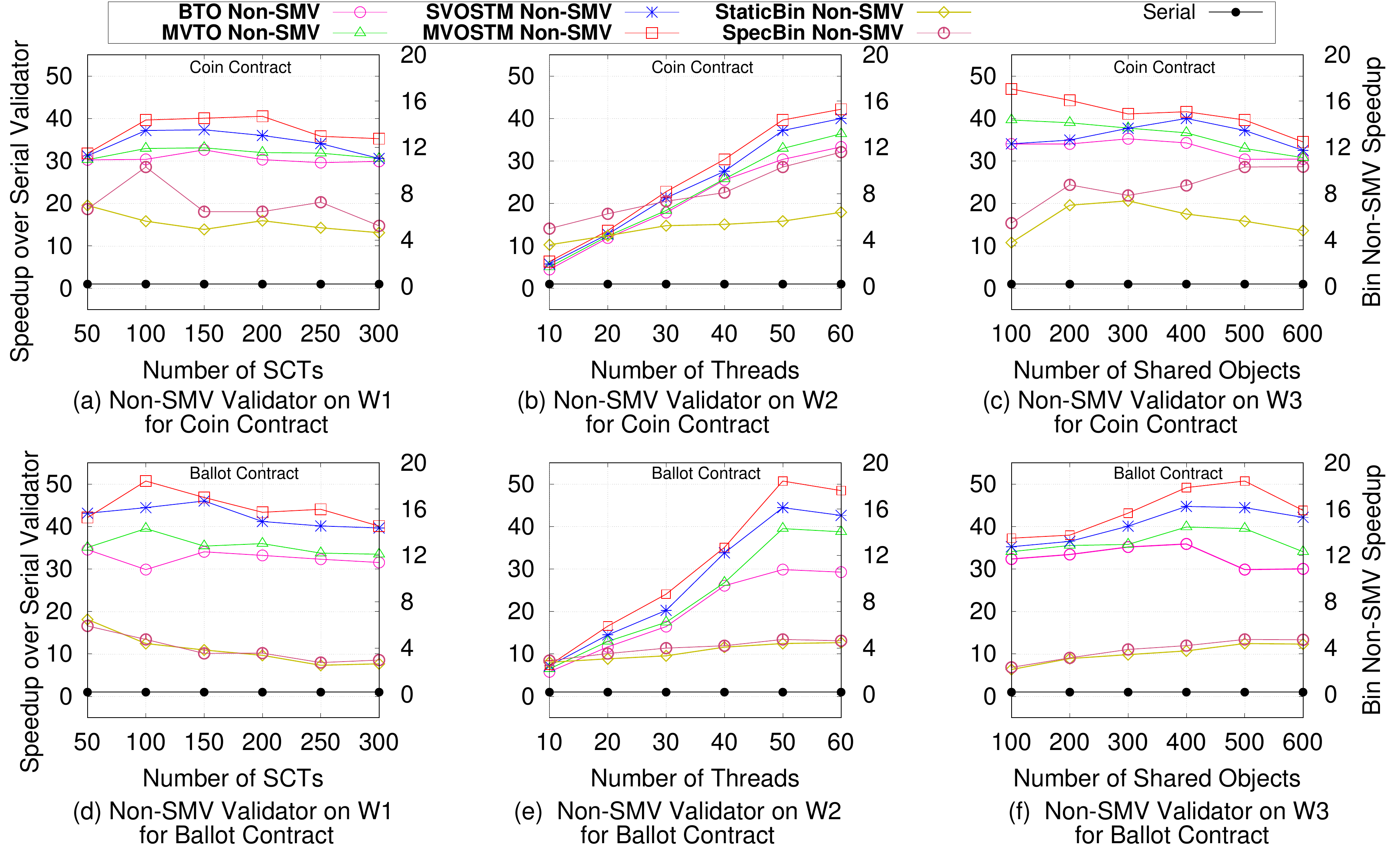}\vspace{-.3cm}
	 \caption{\Mthr Non\scv Validator Speedup Over Serial Validator for Coin and Ballot Contract}
	\label{fig:validatorNonSMVcb}

	\includegraphics[width=\textwidth]{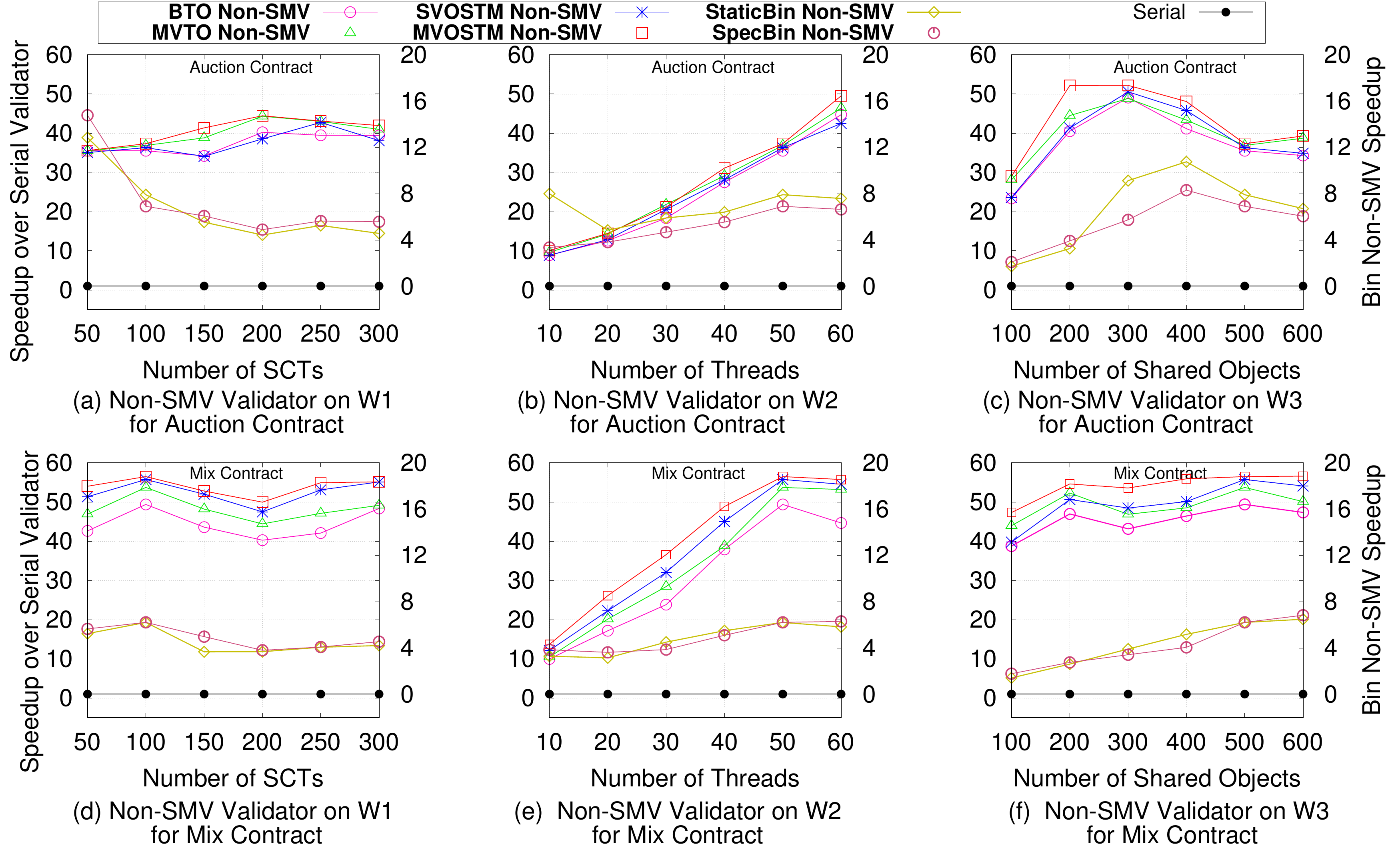}\vspace{-.3cm}
	 \caption{\Mthr Non\scv Validator Speedup Over Serial Validator for Auction and Mix Contract}
	\label{fig:validatorNonSMVam}
\end{figure}

\noindent\textbf{Performance Analysis of Fork-join \scv Validator:} \figref{w1-validatorFJ} and \figref{w2-validatorFJ} show the performance of the fork-join validator \cite{Anjana:OptSC:PDP:2019,Dickerson+:ACSC:PODC:2017}. 
Here we can observe that the average speedup achieved by the fork-join validator is very less compared to other \scv. The reason for low speedups by \mthr fork-join validators is possibly due to the working of the master thread, which becomes slow to allocate the \sctrn{s} to the slave threads and hence becomes the bottleneck. 
\begin{figure}[H]
    \centering
	\includegraphics[width=\textwidth]{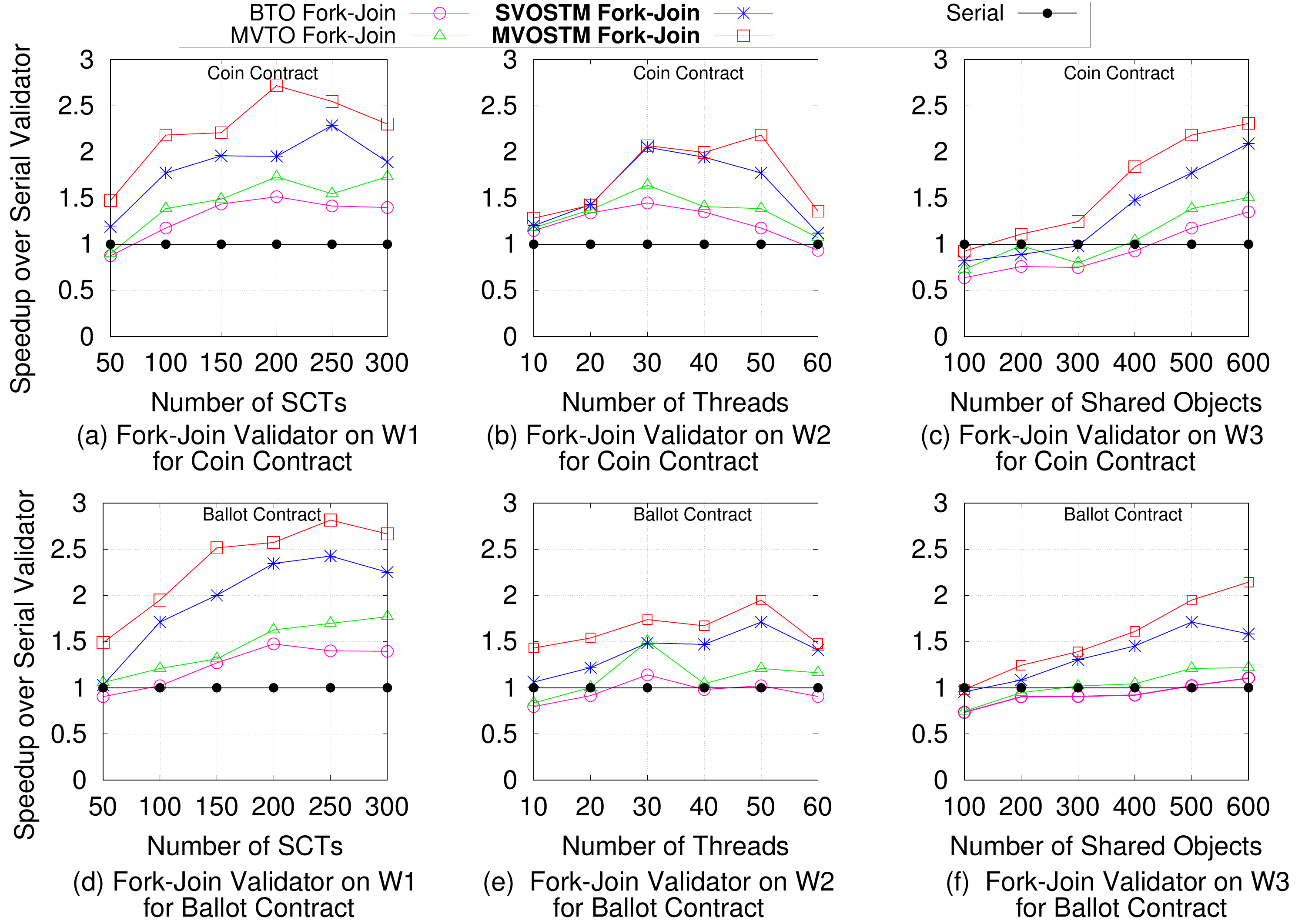}\vspace{-.3cm}
	 \caption{\Mthr Fork-join Validator Speedup over Serial Validator for Coin and Ballot Contract}
	\label{fig:w1-validatorFJ}

 	\includegraphics[width=\textwidth]{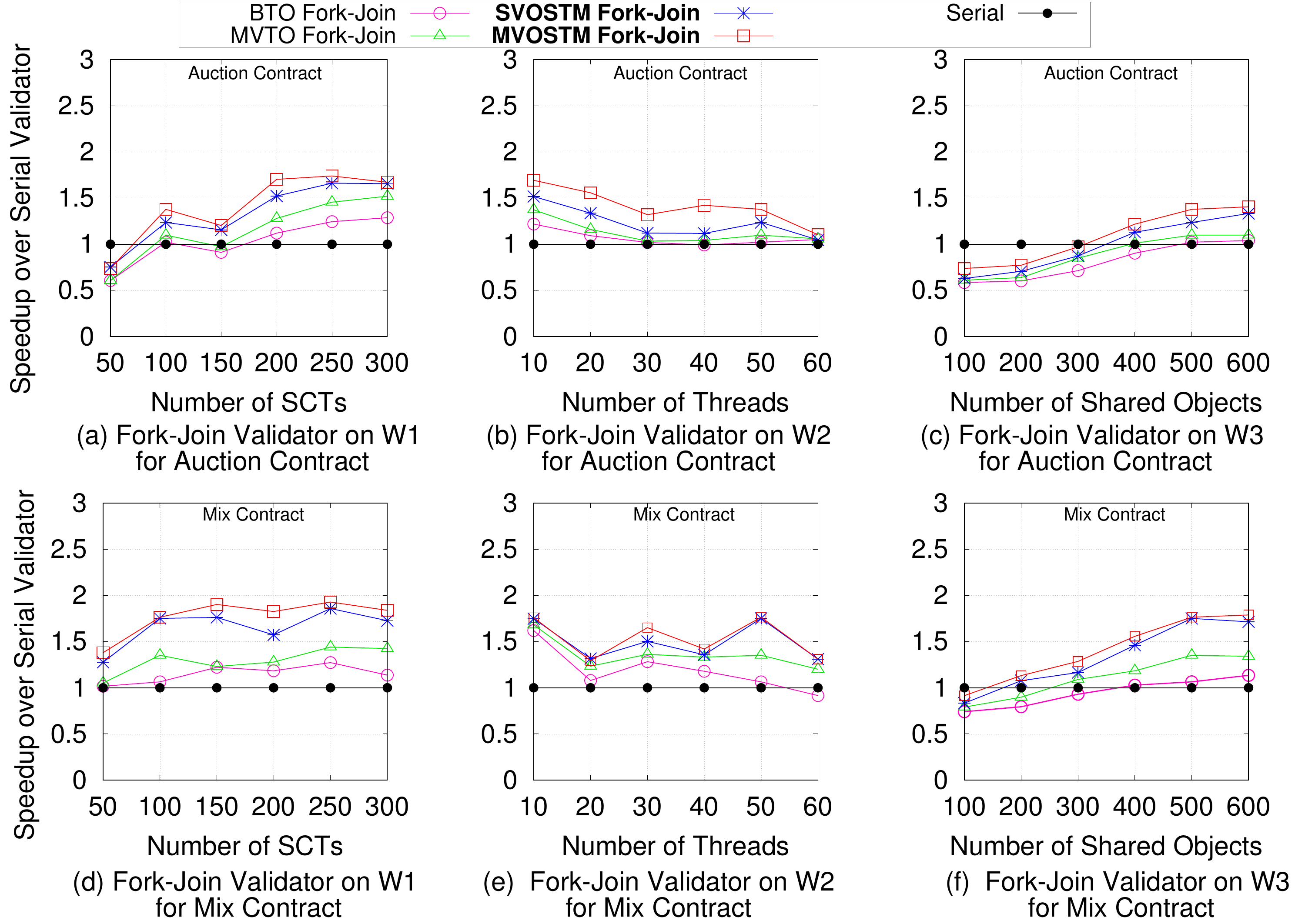}\vspace{-.3cm}
	 \caption{\Mthr Fork-join Validator Speedup over Serial Validator for Auction and Mix Contract}
	\label{fig:w2-validatorFJ}
\end{figure}

\noindent
\textbf{Time Taken by \Mthr Miner and \scv on Workload-1 Benchmark Contracts:} For the better clarity, we present the actual time taken by the miner and validators on W1. Table \ref{tab:w1-coin-miner} to Table \ref{tab:w1-mix-miner} show the time taken by the miners for all the four benchmark contracts, while Table \ref{tab:w1-coin-val} to Table \ref{tab:w1-mix-val} show the time taken by the validators. The time shown in the table is in the microsecond ($\mu s$) and averaged over 26 runs where the first run is considered as warm-up run and discarded.

\begin{table}[H]
\centering
\caption{\Mthr v/s Serial Miner Time on W1 for Coin Contract (in $\mu s$)}
\label{tab:w1-coin-miner}
\resizebox{.75\textwidth}{!}{%
\begin{tabular}{|c|c|c|c|c|c|c|c|}
\hline
 &  &  &  & {\color[HTML]{00009B} } & {\color[HTML]{00009B} } &  &  \\
\multirow{-2}{*}{\textbf{\begin{tabular}[c]{@{}c@{}}\# \sctrn{s}\\\end{tabular}}} & \multirow{-2}{*}{\textbf{Serial}} & \multirow{-2}{*}{\textbf{\begin{tabular}[c]{@{}c@{}}\bto\\ Miner\end{tabular}}} & \multirow{-2}{*}{\textbf{\begin{tabular}[c]{@{}c@{}}\mvto\\ Miner\end{tabular}}} & \multirow{-2}{*}{{\color[HTML]{00009B} \textbf{\begin{tabular}[c]{@{}c@{}}\svotm\\ Miner\end{tabular}}}} & \multirow{-2}{*}{{\color[HTML]{00009B} \textbf{\begin{tabular}[c]{@{}c@{}}\mvotm\\ Miner\end{tabular}}}} & \multirow{-2}{*}{\textbf{\begin{tabular}[c]{@{}c@{}}StaticBin\\ Miner\end{tabular}}} & \multirow{-2}{*}{\textbf{\begin{tabular}[c]{@{}c@{}}SpecBin\\ Miner\end{tabular}}} \\ \hline
	\textbf{50} & 150.65 & 68.1112 & 50.3176 & 22.8232 & 14.9664 & 86.1328 & 12.7848 \\ \hline
	\textbf{100} & 272.71 & 146.647 & 123.096 & 44.5568 & 37.464 & 159.595 & 33.1864 \\ \hline
	\textbf{150} & 379.18 & 262.93 & 233.871 & 76.3768 & 55.584 & 271.694 & 49.1736 \\ \hline
	\textbf{200} & 487.52 & 352.554 & 297.997 & 166.834 & 97.5192 & 527.921 & 72.2712 \\ \hline
	\textbf{250} & 587.215 & 450.446 & 390.727 & 208.166 & 122.653 & 724.494 & 91.472 \\ \hline
	\textbf{300} & 696.445 & 534.891 & 444.716 & 261.277 & 173.087 & 982.792 & 150.039 \\ \hline
\end{tabular}%
}
\end{table}

\begin{table}[H]
\centering
\caption{\Mthr v/s Serial Miner Time on W1 for Ballot Contract (in $\mu s$)}
\label{tab:w1-ballot-miner}
\resizebox{.75\textwidth}{!}{%
\begin{tabular}{|c|c|c|c|c|c|c|c|}
\hline
 &  &  &  & {\color[HTML]{00009B} } & {\color[HTML]{00009B} } &  &  \\
\multirow{-2}{*}{\textbf{\begin{tabular}[c]{@{}c@{}}\# \sctrn{s}\\\end{tabular}}} & \multirow{-2}{*}{\textbf{Serial}} & \multirow{-2}{*}{\textbf{\begin{tabular}[c]{@{}c@{}}\bto\\ Miner\end{tabular}}} & \multirow{-2}{*}{\textbf{\begin{tabular}[c]{@{}c@{}}\mvto\\ Miner\end{tabular}}} & \multirow{-2}{*}{{\color[HTML]{00009B} \textbf{\begin{tabular}[c]{@{}c@{}}\svotm\\ Miner\end{tabular}}}} & \multirow{-2}{*}{{\color[HTML]{00009B} \textbf{\begin{tabular}[c]{@{}c@{}}\mvotm\\ Miner\end{tabular}}}} & \multirow{-2}{*}{\textbf{\begin{tabular}[c]{@{}c@{}}StaticBin\\ Miner\end{tabular}}} & \multirow{-2}{*}{\textbf{\begin{tabular}[c]{@{}c@{}}SpecBin\\ Miner\end{tabular}}} \\ \hline
	\textbf{50} & 159.68 & 118.534 & 105.431 & 61.324 & 44.068 & 90.9888 & 35.4624 \\ \hline
	\textbf{100} & 270.72 & 228.384 & 200.039 & 95.8848 & 86.2352 & 189.968 & 98.7128 \\ \hline
	\textbf{150} & 426.24 & 425.461 & 357.151 & 171.871 & 162.211 & 424.124 & 181.074 \\ \hline
	\textbf{200} & 524.64 & 656.821 & 723.909 & 264.23 & 274.261 & 674.95 & 310.022 \\ \hline
	\textbf{250} & 633.32 & 897.225 & 919.69 & 338.452 & 410.184 & 846.106 & 423.346 \\ \hline
	\textbf{300} & 775.96 & 955.438 & 1033.51 & 428.401 & 519.42 & 990.503 & 584.227 \\ \hline
\end{tabular}%
}
\end{table}

\begin{table}[H]
\centering
\caption{\Mthr v/s Serial Miner Time on W1 for Auction Contract (in $\mu s$)}
\label{tab:w1-auction-miner}
\resizebox{.75\textwidth}{!}{%
\begin{tabular}{|c|c|c|c|c|c|c|c|}
\hline
 &  &  &  & {\color[HTML]{00009B} } & {\color[HTML]{00009B} } &  &  \\
\multirow{-2}{*}{\textbf{\begin{tabular}[c]{@{}c@{}}\# \sctrn{s}\\\end{tabular}}} & \multirow{-2}{*}{\textbf{Serial}} & \multirow{-2}{*}{\textbf{\begin{tabular}[c]{@{}c@{}}\bto\\ Miner\end{tabular}}} & \multirow{-2}{*}{\textbf{\begin{tabular}[c]{@{}c@{}}\mvto\\ Miner\end{tabular}}} & \multirow{-2}{*}{{\color[HTML]{00009B} \textbf{\begin{tabular}[c]{@{}c@{}}\svotm\\ Miner\end{tabular}}}} & \multirow{-2}{*}{{\color[HTML]{00009B} \textbf{\begin{tabular}[c]{@{}c@{}}\mvotm\\ Miner\end{tabular}}}} & \multirow{-2}{*}{\textbf{\begin{tabular}[c]{@{}c@{}}StaticBin\\ Miner\end{tabular}}} & \multirow{-2}{*}{\textbf{\begin{tabular}[c]{@{}c@{}}SpecBin\\ Miner\end{tabular}}} \\ \hline
\textbf{50} & 106.8 & 18.0744 & 15.5328 & 14.3912 & 12.9552 & 45.8104 & 22.0304 \\ \hline
\textbf{100} & 239.04 & 86.7432 & 70.0488 & 45.6224 & 42.373 & 113.296 & 88.044 \\ \hline
\textbf{150} & 341.28 & 142.105 & 135.964 & 109.977 & 93.825 & 218.606 & 170.866 \\ \hline
\textbf{200} & 441.04 & 217.049 & 191.782 & 177.425 & 145.63 & 318.754 & 238.448 \\ \hline
\textbf{250} & 534.88 & 315.15 & 269.503 & 242.992 & 227.251 & 486.099 & 343.722 \\ \hline
\textbf{300} & 636.24 & 593.994 & 541.058 & 381.154 & 370.016 & 756.739 & 479.808 \\ \hline
\end{tabular}%
}
\end{table}

\begin{table}[H]
\centering
\caption{\Mthr v/s Serial Miner Time on W1 for Mix Contract (in $\mu s$)}
\label{tab:w1-mix-miner}
\resizebox{.75\textwidth}{!}{%
\begin{tabular}{|c|c|c|c|c|c|c|c|}
\hline
 &  &  &  & {\color[HTML]{00009B} } & {\color[HTML]{00009B} } &  &  \\
\multirow{-2}{*}{\textbf{\begin{tabular}[c]{@{}c@{}}\# \sctrn{s}\\\end{tabular}}} & \multirow{-2}{*}{\textbf{Serial}} & \multirow{-2}{*}{\textbf{\begin{tabular}[c]{@{}c@{}}\bto\\ Miner\end{tabular}}} & \multirow{-2}{*}{\textbf{\begin{tabular}[c]{@{}c@{}}\mvto\\ Miner\end{tabular}}} & \multirow{-2}{*}{{\color[HTML]{00009B} \textbf{\begin{tabular}[c]{@{}c@{}}\svotm\\ Miner\end{tabular}}}} & \multirow{-2}{*}{{\color[HTML]{00009B} \textbf{\begin{tabular}[c]{@{}c@{}}\mvotm\\ Miner\end{tabular}}}} & \multirow{-2}{*}{\textbf{\begin{tabular}[c]{@{}c@{}}StaticBin\\ Miner\end{tabular}}} & \multirow{-2}{*}{\textbf{\begin{tabular}[c]{@{}c@{}}SpecBin\\ Miner\end{tabular}}} \\ \hline
\textbf{50} & 101.96 & 44.8776 & 31.352 & 21.5408 & 20.0272 & 69.6816 & 19.3464 \\ \hline
\textbf{100} & 192.2 & 92.4168 & 66.972 & 44.2504 & 38.2416 & 140.606 & 52.9 \\ \hline
\textbf{150} & 223.08 & 182.382 & 155.381 & 66.3584 & 56.3176 & 240.282 & 70.8464 \\ \hline
\textbf{200} & 318.04 & 286.35 & 221.89 & 94.5632 & 82.5976 & 330.153 & 97.808 \\ \hline
\textbf{250} & 421.32 & 418.263 & 319.533 & 141.498 & 123.091 & 421.934 & 172.095 \\ \hline
\textbf{300} & 515.56 & 505.855 & 477.872 & 177.78 & 152.14 & 615.446 & 212.537 \\ \hline
\end{tabular}%
}
\end{table}

\begin{table}[H]
\centering
\caption{\scv v/s Serial Validator Time on W1 for Coin Contract (in $\mu s$)}
\label{tab:w1-coin-val}
\resizebox{.75\textwidth}{!}{%
\begin{tabular}{|c|c|c|c|c|c|c|c|}
\hline
 &  & {\color[HTML]{00009B} } & {\color[HTML]{00009B} } & {\color[HTML]{00009B} } & {\color[HTML]{00009B} } & {\color[HTML]{00009B} } & {\color[HTML]{00009B} } \\
\multirow{-2}{*}{\textbf{\begin{tabular}[c]{@{}c@{}}\# \sctrn{s}\\\end{tabular}}} & \multirow{-2}{*}{\textbf{Serial}} & \multirow{-2}{*}{{\color[HTML]{00009B} \textbf{\begin{tabular}[c]{@{}c@{}}\bto\\ \scv\end{tabular}}}} & \multirow{-2}{*}{{\color[HTML]{00009B} \textbf{\begin{tabular}[c]{@{}c@{}}\mvto\\ \scv\end{tabular}}}} & \multirow{-2}{*}{{\color[HTML]{00009B} \textbf{\begin{tabular}[c]{@{}c@{}}\svotm\\ \scv\end{tabular}}}} & \multirow{-2}{*}{{\color[HTML]{00009B} \textbf{\begin{tabular}[c]{@{}c@{}}\mvotm\\ \scv\end{tabular}}}} & \multirow{-2}{*}{{\color[HTML]{00009B} \textbf{\begin{tabular}[c]{@{}c@{}}StaticBin\\ \scv\end{tabular}}}} & \multirow{-2}{*}{{\color[HTML]{00009B} \textbf{\begin{tabular}[c]{@{}c@{}}SpecBin\\ \scv\end{tabular}}}} \\ \hline
\textbf{50} & 141.63 & 4.9848 & 4.7008 & 4.5784 & 4.6784 & 22.4432 & 21.5008 \\ \hline
\textbf{100} & 263.4 & 8.936 & 8.1272 & 7.2896 & 6.8432 & 47.392 & 26.2848 \\ \hline
\textbf{150} & 359.83 & 12.2552 & 11.8888 & 10.7256 & 9.6768 & 74.6168 & 57.9472 \\ \hline
\textbf{200} & 438.83 & 14.7144 & 14.1384 & 12.6072 & 11.9136 & 83.0256 & 71.7624 \\ \hline
\textbf{250} & 562.24 & 19.4272 & 18.6376 & 16.6352 & 16.1696 & 112.474 & 78.2856 \\ \hline
\textbf{300} & 664.305 & 23.658 & 22.439 & 22.3 & 20.0271 & 145.223 & 127.9989 \\ \hline
\end{tabular}%
}
\end{table}

\begin{table}[H]
\centering
\caption{\scv v/s Serial Validator Time on W1 for Ballot Contract (in $\mu s$)}
\label{tab:w1-ballot-val}
\resizebox{.75\textwidth}{!}{%
\begin{tabular}{|c|c|c|c|c|c|c|c|}
\hline
 &  & {\color[HTML]{00009B} } & {\color[HTML]{00009B} } & {\color[HTML]{00009B} } & {\color[HTML]{00009B} } & {\color[HTML]{00009B} } & {\color[HTML]{00009B} } \\
\multirow{-2}{*}{\textbf{\begin{tabular}[c]{@{}c@{}}\# \sctrn{s}\\\end{tabular}}} & \multirow{-2}{*}{\textbf{Serial}} & \multirow{-2}{*}{{\color[HTML]{00009B} \textbf{\begin{tabular}[c]{@{}c@{}}\bto\\ \scv\end{tabular}}}} & \multirow{-2}{*}{{\color[HTML]{00009B} \textbf{\begin{tabular}[c]{@{}c@{}}\mvto\\ \scv\end{tabular}}}} & \multirow{-2}{*}{{\color[HTML]{00009B} \textbf{\begin{tabular}[c]{@{}c@{}}\svotm\\ \scv\end{tabular}}}} & \multirow{-2}{*}{{\color[HTML]{00009B} \textbf{\begin{tabular}[c]{@{}c@{}}\mvotm\\ \scv\end{tabular}}}} & \multirow{-2}{*}{{\color[HTML]{00009B} \textbf{\begin{tabular}[c]{@{}c@{}}StaticBin\\ \scv\end{tabular}}}} & \multirow{-2}{*}{{\color[HTML]{00009B} \textbf{\begin{tabular}[c]{@{}c@{}}SpecBin\\ \scv\end{tabular}}}} \\ \hline
\textbf{50} & 156.24 & 5.2896 & 5.0248 & 4.2376 & 4.0416 & 26.564 & 30.776 \\ \hline
\textbf{100} & 289.8 & 10.484 & 9.8752 & 7.5848 & 6.3024 & 68.3032 & 63.1488 \\ \hline
\textbf{150} & 425.2 & 13.4 & 12.792 & 9.7368 & 9.62 & 112.822 & 120.214 \\ \hline
\textbf{200} & 516.84 & 16.2848 & 15.3904 & 13.1784 & 12.4192 & 155.313 & 147.697 \\ \hline
\textbf{250} & 627.2 & 21.5944 & 19.6976 & 16.4096 & 15.3408 & 254.764 & 232.866 \\ \hline
\textbf{300} & 757.8 & 25.1328 & 23.8872 & 19.332 & 19.0984 & 293.702 & 261.422 \\ \hline
\end{tabular}%
}
\end{table}

\begin{table}[H]
\centering
\caption{\scv v/s Serial Validator Time on W1 for Auction Contract (in $\mu s$)}
\label{tab:w1-auction-val}
\resizebox{.75\textwidth}{!}{%
\begin{tabular}{|c|c|c|c|c|c|c|c|}
\hline
 &  & {\color[HTML]{00009B} } & {\color[HTML]{00009B} } & {\color[HTML]{00009B} } & {\color[HTML]{00009B} } & {\color[HTML]{00009B} } & {\color[HTML]{00009B} } \\
\multirow{-2}{*}{\textbf{\begin{tabular}[c]{@{}c@{}}\# \sctrn{s}\\\end{tabular}}} & \multirow{-2}{*}{\textbf{Serial}} & \multirow{-2}{*}{{\color[HTML]{00009B} \textbf{\begin{tabular}[c]{@{}c@{}}\bto\\ \scv\end{tabular}}}} & \multirow{-2}{*}{{\color[HTML]{00009B} \textbf{\begin{tabular}[c]{@{}c@{}}\mvto\\ \scv\end{tabular}}}} & \multirow{-2}{*}{{\color[HTML]{00009B} \textbf{\begin{tabular}[c]{@{}c@{}}\svotm\\ \scv\end{tabular}}}} & \multirow{-2}{*}{{\color[HTML]{00009B} \textbf{\begin{tabular}[c]{@{}c@{}}\mvotm\\ \scv\end{tabular}}}} & \multirow{-2}{*}{{\color[HTML]{00009B} \textbf{\begin{tabular}[c]{@{}c@{}}StaticBin\\ \scv\end{tabular}}}} & \multirow{-2}{*}{{\color[HTML]{00009B} \textbf{\begin{tabular}[c]{@{}c@{}}SpecBin\\ \scv\end{tabular}}}} \\ \hline
\textbf{50} & 103.4 & 3.2096 & 3.112 & 3.1424 & 3.1224 & 10.4136 & 8.6112 \\ \hline
\textbf{100} & 190.08 & 6.1088 & 5.2912 & 5.5608 & 5.2752 & 33.0736 & 30.668 \\ \hline
\textbf{150} & 290.6 & 9.916 & 8.0408 & 8.4984 & 7.9392 & 55.9136 & 48.576 \\ \hline
\textbf{200} & 406.48 & 12.4536 & 11.0552 & 11.324 & 10.8424 & 115.354 & 98.404 \\ \hline
\textbf{250} & 531.8 & 15.9936 & 14.2256 & 15.4752 & 13.8392 & 150.586 & 94.8384 \\ \hline
\textbf{300} & 606.4 & 19.048 & 16.0512 & 17.1448 & 15.1544 & 168.833 & 118.31 \\ \hline
\end{tabular}%
}
\end{table}

\begin{table}[H]
\centering
\caption{\scv v/s Serial Validator Time on W1 for Mix Contract (in $\mu s$)}
\label{tab:w1-mix-val}
\resizebox{.75\textwidth}{!}{%
\begin{tabular}{|c|c|c|c|c|c|c|c|}
\hline
 &  & {\color[HTML]{00009B} } & {\color[HTML]{00009B} } & {\color[HTML]{00009B} } & {\color[HTML]{00009B} } & {\color[HTML]{00009B} } & {\color[HTML]{00009B} } \\
\multirow{-2}{*}{\textbf{\begin{tabular}[c]{@{}c@{}}\# \sctrn{s}\\\end{tabular}}} & \multirow{-2}{*}{\textbf{Serial}} & \multirow{-2}{*}{{\color[HTML]{00009B} \textbf{\begin{tabular}[c]{@{}c@{}}\bto\\ \scv\end{tabular}}}} & \multirow{-2}{*}{{\color[HTML]{00009B} \textbf{\begin{tabular}[c]{@{}c@{}}\mvto\\ \scv\end{tabular}}}} & \multirow{-2}{*}{{\color[HTML]{00009B} \textbf{\begin{tabular}[c]{@{}c@{}}\svotm\\ \scv\end{tabular}}}} & \multirow{-2}{*}{{\color[HTML]{00009B} \textbf{\begin{tabular}[c]{@{}c@{}}\mvotm\\ \scv\end{tabular}}}} & \multirow{-2}{*}{{\color[HTML]{00009B} \textbf{\begin{tabular}[c]{@{}c@{}}StaticBin\\ \scv\end{tabular}}}} & \multirow{-2}{*}{{\color[HTML]{00009B} \textbf{\begin{tabular}[c]{@{}c@{}}SpecBin\\ \scv\end{tabular}}}} \\ \hline
\textbf{50} & 2.1936 & 2.0072 & 1.8232 & 1.7088 & 17.4024 & 16.3712 & 8.6112 \\ \hline
\textbf{100} & 4.1016 & 4.088 & 3.6 & 3.4512 & 31.7432 & 31.042 & 30.668 \\ \hline
\textbf{150} & 5.6592 & 4.812 & 4.7304 & 4.4384 & 58.816 & 42.975 & 48.576 \\ \hline
\textbf{200} & 7.3952 & 6.7504 & 6.224 & 6.0672 & 77.756 & 75.0021 & 98.404 \\ \hline
\textbf{250} & 9.4328 & 8.6696 & 8.0208 & 7.4904 & 94.9208 & 93.9526 & 94.8384 \\ \hline
\textbf{300} & 11.6016 & 10.5352 & 9.3392 & 8.9192 & 116.43 & 107.08 & 118.31 \\ \hline
\end{tabular}%
}
\end{table}


\cmnt{
This section presents a detailed description of smart contracts that we have considered in this paper. It also includes the additional experiments that we have done to show the performance benefits of concurrent execution of smart contracts by miner and validator using our proposed algorithm on various workloads. Along with this, we test our proposed \emph{Smart \Mthr Validator} to detect the malicious miner.

\noindent
\textbf{Smart Contract:} Clients send transactions to the miners in the form of complex code known as smart contracts. It provides several complex services such as managing the system state, ensuring rules, or credentials checking of the parties involved, etc. \cite{Dickerson+:ACSC:PODC:2017}. For better understanding, we have described \emph{Coin, Ballot, Simple Auction} Smart Contract from Solidity documentation \cite{Solidity}. We consider one more smart contract as \emph{Mixed Contract}, which is the combination of the three contracts as mentioned above in equal proportion and seems more realistic.

\noindent(1) \textit{Coin Contract:} It is a sub-currency contract which implements the simplistic form of a cryptocurrency and is used to transfer coins from one account to another account using  \emph{send()}, or to check the account balance using \emph{get\_balance()} function. Accounts (unique addresses in Ethereum) are shared objects. A conflict will occur when two or more transaction consists of at least one common account, and one of them is updating the account balance.

\algoref{cc1} shows the functionality of the coin contract, where \textit{mint()}, \textit{send()}, and \textit{get\_balance()} are the functions of the contract. These functions can be called by the miners or through other contracts. It permanently initialized by the contract creator (or contract deployer) to a special public state variable \textit{minter} \Lineref{c2}. Accounts are realized using Solidity mapping data structure essentially a $\langle \emph{key-value} \rangle$ pair at \Lineref{c3}, where a key is the unique Ethereum address and value is unsigned integer depicts the coins (or balance) in respective account. Initially, the contract deployer (aka \textit{minter}) creates new coins and allocate it to each receiver at \Lineref{c9}.  

Further, in \emph{send()} function, to transfer the coin from sender account to receiver account, function ensures that the sender has sufficient balance in his account at \Lineref{c11}. If sufficient balance found in senders account, the coin transferred from sender account to receiver account. By calling \textit{get\_balance()}, anyone can query the specific account balance at \Lineref{c15}.


\noindent(2) \textit{Ballot Contract: } This contract is used to organize electronic voting where voters and proposals are the shared objects and stored at unique Ethereum addresses. At the beginning of voting, the chairman of the ballot gives rights to voters to vote. Later, voters either delegate their vote to other voter using \emph{delegate()} or directly vote to specific proposal using \emph{vote()}. Voters are allowed to delegate or vote only once per ballot. A conflict will occur when two or more voters vote for the same proposal, or they delegate their votes to the same voter simultaneously.
\vspace{.2cm}
\setlength{\intextsep}{0pt}
\begin{algorithm}[!htb]
	\scriptsize
	\caption{Coin(): A sub-currency contract used to depict the simplest form of a cryptocurrency.}
	\label{alg:cc1}
	\begin{algorithmic}[1]
		\makeatletter\setcounter{ALG@line}{170}\makeatother
		\Procedure{$Coin()$}{} \label{lin:c1}
		\State address public minter;/*Minter is a unique public address*/\label{lin:c2}
		\State /*Map $\langle \emph{key-value} \rangle$ pair of hash table as $\langle \emph{address-balance} \rangle$*/
		\State mapping(address $=>$ uint) balances. \label{lin:c3} 
		
		\State \textbf{Constructor()} public\label{lin:c4}
		\State{\hspace{.3cm} minter = msg.sender.} /*Set the sender as minter*/\label{lin:c5}
		
		\Function {}{}mint(address receiver, uint amount )\label{lin:c6}
		\If{(msg.sender == minter)} \label{lin:c7}
		\State /*Initially, add the balance into receiver account*/
		\State balances[receiver] += amount. \label{lin:c9}    
		\EndIf
		
		\EndFunction
		
		\Function {}{}send(address receiver, uint amount)\label{lin:c10}
		\State /*Sender don't have sufficient balance*/
		\If{(balances[msg.sender] $<$ amount)} \label{lin:c11}
		return $\langle fail \rangle$;\label{lin:c12}
		\EndIf
		\State balances[msg.sender] -= amount;\label{lin:c13}
		\State balances[receiver] += amount;\label{lin:c14}
		\EndFunction
		
		\Function{}{}get\_balance(address account)\label{lin:c15}
		\State return $\langle$balance$\rangle$;
		\EndFunction
		\EndProcedure
		
	\end{algorithmic}
\end{algorithm}
\noindent(3) \textit{Simple Auction Contract: }In this contract, an auction is conducted where a bidder places their bids. Here bidders, maxBid, maxBidder, and auction end time are the shared object which can be accessed by multiple threads. The auction will end when the bidding period (or end time) of the auction is over. The auction end time is initialized at the beginning by the auction master. A \emph{bid()} function is used to bid the amount by a bidder for the auction. In the end, the bidder with the highest bid amount will be the winner, and all other bidders amount is then returned back to them using \emph{withdraw()}. A conflict will occur when more than two bidders try to bid using \emph{bidPlusOne()} at the same time.

\noindent(4) \textit{Mix Contract: }In this contract, aforementioned smart contracts are executed simultaneously. This contract is designed to intercept real-time scenarios in which a block consists of \sctrn{s} from different contracts. For the experiment, we combined \sctrn{s} for three contracts in a block.

\noindent\textbf{Performance Analysis of fork-join validator for W1 and W2:} \figref{w1-validatorFJ} and \figref{w2-validatorFJ} show the performance of the fork-join validator. \mvotm, \svotm, \mvto, and \bto fork-join validators obtain an average speedup of 1.76$\times$, 1.56$\times$, 1.29$\times$, and 1.15$\times$ over serial validator respectively. The reason for low speedups by fork-join validators as compared to decentralized validator is possibly due to the working of the master thread which becomes slow to allocate the \sctrn{s} to the slave threads and hence becomes the bottleneck.


\noindent\textbf{Performance Analysis for Workload (W3): } To test the performance of proposed approach further, we consider one more workload W3, in which the number of shared data-items vary from 100 to 600, while threads, \sctrn{s}, and hash table size are fixed to 50, 100, and 30 respectively. For W3, with the increase in the number of shared data-items, contention will reduce. 

\begin{figure*}
	\includegraphics[width=\textwidth,  height=5cm]{figs/w1validatorFJ.pdf}
	 \caption{(W1: Varying \sctrn{s}) Concurrent Fork-join Validator Speedup Over Serial Validator}
	\label{fig:w1-validatorFJ}
\end{figure*}

\begin{figure*}
	\includegraphics[width=\textwidth,  height=5cm]{figs/w2validatorFJ.pdf}
	 \caption{(W2: Varying Threads) Concurrent Fork-join Validator Speedup Over Serial Validator}
	\label{fig:w2-validatorFJ}
\end{figure*}

\begin{figure*}
	\includegraphics[width=\textwidth, height=5cm]{figs/w1depen.pdf}
	 \caption{(W1: Varying \sctrn{s}) Average Number of Dependencies in Block Graph}
	\label{fig:w1-BG}
\end{figure*}

\begin{figure*}
	\includegraphics[width=\textwidth,  height=5cm]{figs/w2depen.pdf}
	 \caption{(W2: Varying Threads) Average Number of Dependencies in Block Graph}
	\label{fig:w2depen}
\end{figure*}

\begin{figure*}
	\includegraphics[width=\textwidth,  height=5cm]{figs/w3depen.pdf}
	 \caption{(W3: Varying Shared data-items) Average Number of Dependencies in Block Graph}
	\label{fig:w3-BG}
\end{figure*}

\begin{figure*}
	\includegraphics[width=\textwidth,  height=5cm]{figs/w3miner.pdf}
	 \caption{(W3: Varying Shared data-items) \Mthr Miner Speedup Over Serial Miner}
	\label{fig:w3-miner}
\end{figure*}

\begin{figure*}
	\includegraphics[width=\textwidth,  height=5cm]{figs/w3validator.pdf}
	 \caption{(W3: Varying Shared data-items) Concurrent Decentralised Validator Speedup Over Serial Validator}
	\label{fig:w3-validator}
\end{figure*}
\begin{figure*}
	\includegraphics[width=\textwidth,  height=5cm]{figs/w3validatorFJ.pdf}
	 \caption{(W3: Varying Shared data-items) Concurrent Fork-join Validator Speedup Over Serial Validator}
	\label{fig:w3-validatorFJ}
\end{figure*}

\begin{table*}
	\centering
	 \caption{Overall Avg. Speedup by \Mthr Miner over Serial Miner}
	\label{tbl:avgMiner}
	\resizebox{.56\textwidth}{!}{%
		\begin{tabular}{|c|c|c|c|c|}
			\hline
			\multirow{2}{*}{\textbf{Contract}} & \multicolumn{4}{c|}{\textbf{\Mthr Miner}} \\ \cline{2-5} 
			& \textbf{BTO} & \textbf{MVTO} & \textbf{\svotm} & \textbf{\mvotm} \\ \hline
			\textbf{Coin} & 3.135 & 3.530 & 5.541 & 6.174 \\ \hline
			\textbf{Ballot} & 1.651 & 1.785 & 3.961 & 4.309 \\ \hline
			\textbf{Auction} & 1.898 & 2.376 & 2.713 & 3.310 \\ \hline
			\textbf{Mix} & 2.363 & 2.518 & 3.788 & 4.761 \\ \hline
			\textit{\textbf{Total Avg. Speedup}} & \textit{2.26} & \textit{2.55} & \textit{4.00} & \textit{4.64} \\ \hline
		\end{tabular}%
	}
\end{table*}

\begin{table*}
	\centering
	 \caption{Overall Avg. Speedup by \Mthr Validator over Serial Validator}
	\label{tbl:avgValidator}
	\resizebox{\textwidth}{!}{%
		\begin{tabular}{|c|c|c|c|c|c|c|c|c|}
			\hline
			\multirow{2}{*}{\textbf{Contract}} & \multicolumn{4}{c|}{\textbf{Decentralized Validator}} & \multicolumn{4}{c|}{\textbf{Fork-join Validator}} \\ \cline{2-9} 
			& \textbf{BTO} & \textbf{MVTO} & \textbf{\svotm} & \textbf{\mvotm} & \textbf{BTO} & \textbf{MVTO} & \textbf{\svotm} & \textbf{\mvotm} \\ \hline
			\textbf{Coin} & 18.941 & 21.432 & 23.357 & 25.983 & 1.156 & 1.293 & 1.589 & 1.852 \\ \hline
			\textbf{Ballot} & 23.187 & 24.904 & 27.859 & 29.698 & 1.045 & 1.200 & 1.567 & 1.841 \\ \hline
			\textbf{Auction} & 35.652 & 37.784 & 40.606 & 43.918 & 0.970 & 1.056 & 1.181 & 1.299 \\ \hline
			\textbf{Mix} & 24.030 & 27.172 & 31.580 & 36.169 & 1.096 & 1.255 & 1.496 & 1.571 \\ \hline
			\textit{\textbf{Total Avg. Speedup}} & \textit{25.45} & \textit{27.82} & \textit{30.85} & \textit{33.94} & \textit{1.07} & \textit{1.20} & \textit{1.46} & \textit{1.64} \\ \hline
		\end{tabular}%
	}
\end{table*}

So, concurrent execution of \sctrn{s} gives better performance. As shown in \figref{w3-miner}, \figref{w3-validator}, and \figref{w3-validatorFJ} the rise in the performance with the increase in shared data-item can be observed. 
However, for Ballot Contract, the average number of dependencies in BG for \bto and \mvto increase with the increase in shared data-items, while it decreases for \svotm and \mvotm as shown in \figref{w3-BG}. The average speedup achieved by the \mthr miner and validator (decentralized and fork-join) for workload W1, W2, and W3 on all benchmarks using different STM protocols are shown in Table \ref{tbl:avgMiner} and Table \ref{tbl:avgValidator} respectively.

\cmnt{
\noindent
\textbf{Experiments on Malicious Miner:} This experiments shows that how many validators accept the block proposed by the malicious miner. First, to make the malicious miner, we generate two double spending (explained in \apnref{ap-pm}) \sctrn{s}. After that malicious miner added these \sctrn{s} into Block Graph (BG) with indegree 0 and append it into the proposed block. 
Finally, malicious miner broadcast the malicious block in the network. 

Other existing node in the network called as validator re-executes the \sctrn{s} concurrently with the help of BG given by the malicious miner. The validator threads execute double spending \sctrn{s} concurrently then they may compute the same final state as given by the malicious miner. So, some of the validators accept the malicious block. If they reach consensus then they will add this malicious block into the blockchain. That may cause a serious issue in the blockchain network. \figref{mm}.(a) demonstrates the number of validators accepting the malicious block while varying the malicious blocks from 10 to 50. Here, we consider total number of validators are 500 to validate a block and run the experiments for the coin contract but similar observations can be seen for the other smart contracts as well. So, we can conclude that if the malicious miner is present in the network then some validators may agree on the malicious block which will violate the property of blockchain. So, we should ensure that such malicious block acceptance should not exist in the blockchain. 

To address this issue we propose \emph{Smart Concurrent Validator (SCV)} (describe in \apnref{ap-pm}) which always detects the malicious block with the help of $counter$ and straightforward reject that block. So, we applied $counter$ approach to all the proposed validators called as \emph{\svotm Decentralized (Dece.) SCV} and \emph{\mvotm Decentralized SCV}. Similarly, we name the state-of-the-art concurrent validators as \emph{BTO Decentralized SCV} and \emph{MVTO Decentralized SCV}. \figref{mm}.(b) represents the same experiments as defined above for all $SCV$ and it can be seen that all the SCVs are rejecting the malicious block, i.e., it is showing 100\% block rejection by all the $SCV$ which serves our purpose to identify the malicious miner in the network.

An obvious question is like, how much extra time does $SCV$ is taking to serve the purpose of identifying the malicious miner. \figref{vc} represents the speedup taken by the concurrent validations with counter and without counter over serial validator while varying the \sctrn{s} on workload $W1$. It can be seen that the counter-based concurrent validator approach is giving a bit less speedup than without counter. But it is appreciable to use counter-based $SCV$ to preserve the correctness of the blockchain.  
}

\cmnt{
\noindent
\textbf{Experiments on Malicious Miner: }In concurrent execution of the \sctrn{s}, the validators deterministically execute the \sctrn{s} based on the BG append by the miner in the block. But what if a miner is malicious and embeds an incorrect BG? To answer this question, we have done experiments for malicious miner. As explained earlier, we also proposed a technique to prevent such malicious activity due to concurrent execution of \sctrn{s} of the block. This experiment shows how many validators accept malicious block proposed by a malicious miner.

To make the malicious miner, we generate double-spending in coin contract (explained in \apnref{ap-pm}) and double voting (in Ballot contract: a voter votes two different proposals with one voting right). After that, malicious miner added these \sctrn{s} into the block, manipulate the Final State accordingly but did not add the respective dependencies in BG, i.e., for these two \sctrn{s} indegree will be 0. Finally, malicious miner broadcast the malicious block in the network. Other existing nodes in the network, i.e., validator re-executes the \sctrn{s} concurrently with the help of BG given by the malicious miner. The validators may execute double-spending \sctrn{s} concurrently and compute the same final state as provided by the malicious miner. So, some of the validators accept the malicious block. If they reach consensus, then they will add this malicious block into the blockchain. This may cause a severe issue in the blockchain. \figref{mm} demonstrates the number of validators accepting the malicious block while varying the malicious blocks from 0 to 50. Here, we consider 500 validators and run the experiments for the coin, ballot, and mix contract. So, we can conclude that if the malicious miner is present in the network, then some validators may agree on the malicious block which will violate the property of blockchain. Therefore, we should ensure that such malicious block acceptance should not exist in the blockchain.

To address this issue, we proposed \emph{Smart Concurrent Validator (SCV)} (describe in \apnref{ap-pm}) which always detects the malicious block at the time of concurrent execution of malicious \sctrn{s} (double-spending and double voting) with the help of $counter$ and straightforward reject that block. We applied $counter$ approach to all the proposed validators called as \emph{\svotm Decentralized (Dece.) SCV} and \emph{\mvotm Decentralized SCV}. Similarly, we name the state-of-the-art concurrent validators as \emph{BTO Decentralized SCV} and \emph{MVTO Decentralized SCV}. \figref{mm}.(b) represents the same experiments as defined above for all \emph{SCV} and it can be seen that all the \emph{SCVs} are rejecting the malicious block, i.e., it is showing 100\% block rejection by all the \emph{SCV} which serves our purpose to identify the malicious concurrent miner in the network. Another advantage of \emph{SCV} is that the validator may detect the malicious miner during the concurrent execution of malicious \sctrn{s} and need not to execute the remaining \sctrn{s} of the block since miner is malicious which saves the validator time. It can be observed in the \figref{w1-validatorFJ}.(b) that for ballot contract the number of validators accepting a malicious block is more than 50\% for the state-of-the-art \emph{BTO} and \emph{MVTO Decentralized Validator}, though its random process, but even a single malicious block may create vital damage to the blockchain.

So the next obvious question is, how much extra time does \emph{SCV} is taking to serve the purpose of identifying the malicious miner? \figref{vc} represents the speedup taken by the concurrent validations with counter and without counter over serial validator while varying the \sctrn{s} on workload $W1$. It can be seen that the counter-based concurrent validator (i.e., smart validator) approach is giving a bit less speedup than without counter. But it is appreciable to use counter-based \emph{SCV} to preserve the correctness of the blockchain.

\cmnt{
\begin{figure}
	\centering
	\includegraphics[width=9cm,  height=5cm]{figs/mm.pdf}
	 \caption{Percentage of Validator Accepted the Malicious Block without Counter vs With Counter}
	\label{fig:mm}
\end{figure}

\begin{figure}
	\centering
	\includegraphics[width=7cm,  height=5cm]{figs/vc.pdf}
	 \caption{Speedup of Counter and without Counter based Validator over Serial}
	\label{fig:vc}
\end{figure}
}

\begin{figure}
	\centering
	\includegraphics[width=\textwidth]{figs/speedupSCV.pdf}
	 \caption{Speedup of Counter (SCV) and Without Counter based Dec. Validator over Serial}
	\label{fig:vc}
\end{figure}

\begin{figure}
	\centering
	\includegraphics[width=\textwidth]{figs/acceptedMM.pdf}
	 \caption{Percentage of Validator Accepted the Malicious Block Without Counter vs With Counter}
	\label{fig:mm}
\end{figure}

\begin{figure}
	\centering
	\includegraphics[width=\textwidth,  height=5cm]{figs/blocksize.pdf}
	 \caption{Percentage of Additional Space Required to Store Block Graph (BG) in Block}
	\label{fig:blocksize}
\end{figure}

}

\noindent
\textbf{Experiments on Malicious Miner: }Validators deterministically and concurrently execute the \sctrn{s} based on the BG append by the miner in the block. However, what if a miner is malicious and embeds an incorrect BG? To answer this question, we have done experiments for malicious miner. As explained earlier, we also proposed a technique to prevent such malicious activity due to concurrent execution of \sctrn{s} of the block. This experiment shows how many validators accept malicious block proposed by a malicious miner.

To make the malicious miner, we generate two \sctrn{s} of double-spending (explained in \apnref{ap-pm}) in Coin contract and double voting (a voter votes two different proposals with one voting right) in Ballot contract. After that, malicious miner added these \sctrn{s} into the block, manipulate the final state accordingly but did not add the respective dependencies in BG, i.e., for this two \sctrn{s} indegree will be 0. Finally, malicious miner broadcast the malicious block in the network. Other existing nodes in the network, i.e., validator re-executes the \sctrn{s} concurrently with the help of BG given by the malicious miner. The validators may execute double-spending \sctrn{s} concurrently and compute the same final state as provided by the malicious miner. So, some of the validators accept the malicious block. If they reach consensus, then they will add this malicious block into the blockchain. This may cause a severe issue in the blockchain. \figref{mm} demonstrates the percentage of validators accepting the malicious block while varying the malicious blocks from 0 to 50. Here, we consider 500 validators and run the experiments for the Coin, Ballot, and Mix contract. So, we can conclude that if the malicious miner is present in the network, then some validators may agree on the malicious block which will violate the property of blockchain. Therefore, we should ensure that such malicious block acceptance should not exist in the blockchain. To address this issue, we proposed \emph{Smart \Mthr Validator (SMV)} (describe in \apnref{ap-mm}) which always detects the malicious block at the time of concurrent execution of malicious \sctrn{s} (double-spending and double voting) with the help of $counter$ and straightforward reject that block. Analysis of \scv is presented in Section \ref{sec:result}.

So, the next obvious question is, how much extra time does \emph{\scv} is taking to serve the purpose of identifying the malicious miner? \figref{vc} represents the speedup taken by the concurrent validations with counter and without counter-based approach over serial validator while varying the \sctrn{s} on workload $W1$. It can be seen that the counter-based \mthr validator (i.e., Smart \Mthr Validator \scv) approach is giving a bit less speedup than without counter-based \mthr validator. However, it is appreciable to use counter-based \emph{\scv} to preserve the correctness of the blockchain.

\begin{figure*}
	\centering
	\includegraphics[width=.85\textwidth, height= 5cm]{figs/speedupSCV.pdf}
	 \caption{Speedup of Counter (\scv) and Without Counter based Dec. Validator over Serial}
	\label{fig:vc}
\end{figure*}

\begin{figure*}
	\centering
	\includegraphics[width=\textwidth,  height=5cm]{figs/blocksize.pdf}
	 \caption{Percentage of Additional Space Required to Store Block Graph (BG) in Block}
	\label{fig:blocksize}
\end{figure*}

\noindent
\textbf{Experiments on Block Graph (BG) Size: }We also measure the additional space required to append the BG into the block. In Ethereum and Bitcoin average block size is $\approx 20.98$ \cite{EthereumAvgBlockSize} and $\approx 1123.34$ KB \cite{BitcoinAvgBlockSize} respectively for the interval of 1$^{st}$ Jan. 2019 to 2$^{st}$ Sept. 2019, which is keep on increasing every year. The average number of transactions in a block of Ethereum is $\approx 100$ \cite{EthereumAvgBlockSize}. So, on an average, each transaction requires $.2$ KB ($\approx 200$ bytes) in Ethereum. Based on this simple estimate, we have computed block size with an increase in \sctrn{s} per block for workload $W1$. To compute the block size \emph{Equation \ref{eq:blocksize}} is used.
\begin{equation}
    B = 200 * N_{\sctrn{s}}
\label{eq:blocksize}
\end{equation}
Where: $B$ is block size in bytes, $N_{\sctrn{s}}$ number of smart contract transactions (\sctrn{s}) in block, and $200$ is the average size of an SCT in bytes.

We use \emph{adjacency list} to maintain the Block Graph $BG(V, E)$ inspired from \cite{Chatterjee+:NbGraph:ICDCN:2019}. Here $V$ is the set of vertices (\vrtnode{s}) is stored as a vertex list, $\vrtlist$. Similarly E is the set of Edges (\egnode{s}) is stored as edge list ($\eglist$ or conflict list) as shown in the \figref{graph}.(a) of \subsecref{bg}. Both $\vrtlist$ and $\eglist$ store between the two sentinel nodes \emph{Head}($-\infty$) and \emph{Tail}($+\infty$). Each \vrtnode{} maintains a tuple: \emph{$\langle$ts, scFun, indegree, egNext, vrtNext$\rangle$}. Here, \emph{ts (an integer)} is the unique timestamp $i$ of the transaction $T_i$ to which this node corresponds to. \emph{scFun (an integer)} is the ID of smart contract function executed by the transaction $T_i$ which is stored in \vrtnode. The number of incoming edges to the transaction $T_i$, i.e. the number of transactions on which $T_i$ depends, is captured by \emph{indegree (an integer)}. Field \emph{egNext (an address)} and \emph{vrtNext (an address)} points the next \egnode{} and \vrtnode{} in the $\eglist$ and $\vrtlist$ respectively. So a vertex node $V_s$ size is $28$ bytes in the experimental system, which is sum of the size of 3 integer variables and 2 pointers.

Each \egnode{} of $T_i$ similarly maintains a tuple: \emph{$\langle$ts, vrtRef, egNext$\rangle$}. Here, \emph{ts  (an integer)} stores the unique timestamp $j$ of $T_j$ which has an edge coming from $T_i$ in the graph. BG maintains the conflict edge from lower timestamp transaction to higher timestamp transaction. This ensures that the \bg is acyclic. The \egnode{s} in $\eglist$ are stored in increasing order of the \emph{ts}. Field \emph{vrtRef (an address)} is a \emph{vertex reference pointer} which points to its own \vrtnode{} present in the $\vrtlist$. This reference pointer helps to maintain the \emph{indegree} count of \vrtnode{} efficiently. The \emph{egNext (an address)} is a pointer to next edge node, so edge node $E_{s}$ requires a total of $20$ bytes in the experimental system. 


The experimental results on the percentage of additional space required to store BG in the block are shown in \figref{blocksize} for workload $W1$. The size of BG ($\beta$) in bytes is computed using \emph{Equation \ref{eq:BGSize}}, while to compute the percentage of additional space ($\beta_{p}$) required to store BG in the block is calculated using \emph{Equation \ref{eq:perBG}}.
\begin{equation}
    \beta = (V_{s} * N_{\sctrn{s}}) + (E_{s} * M_{e})
    \label{eq:BGSize}
\end{equation}
Where: $\beta$ is size of Block Graph (BG) in bytes, $V_s$ is size of a vertex node of $BG$ in bytes, $N_{\sctrn{s}}$ are number of smart contract transactions (\sctrn{s}) in a block, $E_{s}$ is size of a edge node in bytes of $BG$, and $M_{e}$ is number of edges in $BG$.
\begin{equation}
    \beta_{p} = ({\beta*100})/{B}
\label{eq:perBG}
\end{equation}

As shown in \figref{blocksize}, it can be observed that with an increase in the number of dependencies, the space requirements also increase. The number of dependencies in Ballot contract (\figref{w1-BG}.(b)) for $W1$ is higher compared to other contracts, so the space requirement is also high. In all the figures the space requirements of BG by \mvotm, \svotm is smaller than \mvto and \bto miner. The average space required for BG in \% concerning block size is $14.30\%$, $14.57\%$, $21.31\%$, and $23.21 \%$ by \mvotm, \svotm, \mvto, and \bto miner, respectively. Since the number of dependencies in BG developed by \mvotm is smaller then BG generated by other STM protocols, so it requires less space to store BG. In the future, we are planning to reduce space further to store the BG in the block.
}

\end{document}